%% file: Thesis.tex
\documentclass[11pt]{scrreprt}
\include{mypackages}

\include{commenting}

\include{thmstyles}

\include{globalabbervations}
\include{symbols}

\include{myacronyms}

\begin{document}

\pagenumbering{gobble}

\include{titlepage}

\pagenumbering{roman}
\thispagestyle{empty}
\setcounter{page}{0}
\include{preamble}

\include{acknowledgements}

\include{abstract}

\clearpage
\setcounter{tocdepth}{1}
\tableofcontents

\glsaddall
\glossarystyle{long}
\printglossaries
%\printglossaries[type=\symboltype]
\clearpage
\pagenumbering{arabic}

\part{INTRODUCTION AND SETUP}

\include{introduction}

\include{adiabaticmodes}

\include{mantonapproximation}

\part{APPLICATION TO GENERAL RELATIVITY}

\include{mathgr}

\include{naturalformofgr}

\include{adiabaticgr}

\include{slnspace}

\include{GRonBall}

\include{discussion}

\appendix
\part{APPENDIX: MATHEMATICAL BACKGROUND}
\include{mathbackgnd}

%\listofexpl
%\listofque
%\listoflattech
%\listoflan
\printbibliography

%%%%FOR FBE%%%
%\bibliographystyle{styles/fbe_tez_v11}
%\bibliography{bibliography}

\end{document}

%% file: mypackages.tex
\usepackage[paper=a4paper,margin=3cm]{geometry}
\addtokomafont{disposition}{\rmfamily}
\usepackage[backend=biber,backref,sorting=none,isbn=false,url=false]{biblatex}

\DeclareSourcemap{
  \maps[datatype=bibtex]{
    \map[overwrite]{
      \step[fieldsource=eprint, final]
      \step[fieldset=doi, null]
    }  
  }
}

\usepackage{mathtools}
\usepackage[utf8]{inputenc}

\usepackage{graphicx}
\usepackage{wrapfig}
\usepackage{amsthm}
\usepackage{amsfonts}
\usepackage{bbold}
\usepackage[cal=cm,scr=dutchcal]{mathalfa}
\usepackage{relsize}

\usepackage[ddmmyyyy,24hr]{datetime}
\usepackage[defaultlines=3,all]{nowidow}
\usepackage[normalem]{ulem}

\usepackage[hidelinks]{hyperref}
\usepackage[dvipsnames]{xcolor}
\hypersetup
{
   colorlinks,
   linkcolor={white!30!black},
   citecolor={blue!50!black},
   urlcolor={blue!80!black}
}

\usepackage{enumitem}
\usepackage{makeidx}
\makeindex
\usepackage[toc=false,shortcuts,acronyms,symbols,nonumberlist,nogroupskip]{glossaries-extra}
\makeglossaries
\usepackage{tocloft}
\cftpagenumbersoff{part}

%% file: commenting.tex
\newcommand{\listoflan}{List of Possible Language Mistakes}
\newlistof[chapter]{lan}{lan}{\listoflan}

\newcommand{\listtoexplore}{List of Things to Explore Further}
\newlistof[chapter]{expl}{expl}{\listtoexplore}

\newcommand{\expl}[1]{%
\refstepcounter{expl}%
\addcontentsline{expl}{expl}{\protect\numberline{\theexpl}#1}}

\newcommand{\listofquestions}{List of Questions}
\newlistof[chapter]{que}{que}{\listofquestions}

\newcommand{\que}[1]{%
\refstepcounter{que}%
\addcontentsline{que}{que}{\protect\numberline{\theque}#1}}

\newcommand{\listoflattech}{Latex Technical Problems}
\newlistof[chapter]{lattech}{lattech}{\listoflattech}

\newcommand{\lattech}[1]{%
\refstepcounter{lattech}%
\addcontentsline{lattech}{lattech}{\protect\numberline{\thelattech}#1}}

%% file: thmstyles.tex
\theoremstyle{theorem}
\newtheorem{thm}{Theorem}[chapter]
\newtheorem{prop}[thm]{Proposition}
\newtheorem{cor}[thm]{Corollary}
\newtheorem{claim}{Claim}[chapter]
\newtheorem{assumption}{Assumption}[section]
\newtheorem{case}{Case}[assumption]
\newtheorem{lemma}[thm]{Lemma}

\theoremstyle{definition}
\newtheorem{defn}{Definition}[chapter]

\theoremstyle{remark}
\newtheorem{example}{Example}[chapter]

\newtheorem*{remark}{Remark}

%% file: globalabbervations.tex
\newcommand{\rp}{\right)}
\newcommand{\lp}{\left(}
\newcommand{\rb}{\right]}
\newcommand{\lb}{\left[}
\newcommand{\rbr}{\right\rbrace}
\newcommand{\lbr}{\left\lbrace}

\newcommand{\bra}{\right>}
\newcommand{\ket}{\left<}
\newcommand{\ra}{\rightarrow}

\newcommand{\pr}{\partial}
\renewcommand{\uline}{\textit} %an artifact from fbe rules

\DeclareMathOperator{\Lie}{Lie}
\DeclareMathOperator{\Aut}{Aut}
\DeclareMathOperator{\End}{End}
\DeclareMathOperator{\Ad}{Ad}
\DeclareMathOperator{\ad}{ad}
\DeclareMathOperator{\GL}{GL}
\DeclareMathOperator{\gl}{\mathfrak{gl}}
\DeclareMathOperator{\tr}{tr}
\DeclareMathOperator{\dir}{\mathcal{D}}
\DeclareMathOperator{\rmt}{Rm}
\DeclareMathOperator{\ric}{Ric}

\DeclareMathOperator{\diag}{diag}
\DeclareMathOperator{\diff}{Diff}
\DeclareMathOperator{\bdiff}{BDiff}
\DeclareMathOperator{\SO}{SO}
\DeclareMathOperator{\iso}{iso}

\newcommand{\cg}{\, , \quad}
\newcommand{\cgap}{\cg} %discard later on
\newcommand{\sgd}{\, .}
\newcommand{\sgc}{\, ,}
\newcommand{\gag}{\quad \mbox{and} \quad}
\newcommand{\gwg}{\quad \mbox{where} \quad}
\newcommand{\gig}{\quad \mbox{if} \quad}
\newcommand{\gstg}{\quad \mbox{such that} \quad}
\newcommand{\st}{\, : \,}

%% file: symbols.tex
\newcommand{\symb}[3][]{%
  \glsxtrnewsymbol[#1]{#2}{#3}%
 \expandafter\newcommand\csname #2\endcsname{\gls{#2}}%
 }

\newcommand{\G}{\mathcal{G}}
\newcommand{\e}{\textit{e}}
\newcommand{\1}{\mathbb{1}}
\newcommand{\E}{\mathcal{E}}
\newcommand{\m}{\mu}
\newcommand{\n}{\nu}

\symb[description={Adjoint exterior derivative}]{dad}{\ensuremath{d^{\dagger}}}
\newcommand{\dket}{\ket \ket}
\newcommand{\dbra}{\bra \bra}

\symb[description={Set of Smooth Vector Fields}]{vfs}{\ensuremath{\mathfrak{X}}}
\symb[description={Set of Smooth p-forms}]{pfs}{\ensuremath{\Omega^p}}
\symb[description={Pushforward}]{pfwd}{\ensuremath{\phi_*}}
\symb[description={Pushforward}]{dpfwd}{\ensuremath{d\phi}}
\symb[description={Pullback}]{pbck}{\ensuremath{\phi^*}}
\symb[description={Space of tensors with k vector-like and l covector-like components}]{klten}{\ensuremath{T^{(k,l)}TM}}

\symb[description={Lie Algebra of Group $\mathcal{G}$}]{g}{\ensuremath{\mathfrak{g}}}

\symb[description={Riemannian Volume form for a manifold with metric g}]{volg}{\ensuremath{\mathlarger{\epsilon}_g}}
\newcommand{\volk}{\ensuremath{\mathlarger{\epsilon}_k}}
\newcommand{\vol}{\ensuremath{\mathlarger{\epsilon}}}

\symb[description={deRham Cohomology group of degree p}]{derhamp}{\ensuremath{H^p_{dR}}} %(M)
\newcommand{\derham}{H_{dR}}
\symb[description={Ambient Manifold}]{amb}{\ensuremath{\tilde{M}}}
\symb[description={Dimension of the Ambient Manifold}]{dimamb}{\ensuremath{n}}
\symb[description={Submanifold}]{subm}{\ensuremath{M}}
\symb[description={Dimension of the Submanifold}]{dimsubm}{\ensuremath{d}}
\symb[description={Metric of the Ambient Manifold}]{metamb}{\ensuremath{\tilde{g}}}
\symb[description={Induced metric on the Submanifold}]{metsubm}{\ensuremath{h}}
\symb[description={Riemannian Connection of the Ambient Manifold}]{conamb}{\ensuremath{\tilde{\nabla}}}
\symb[description={Riemannian Connection on the Submanifold}]{consubm}{\ensuremath{\nabla}}
\symb[description={Set of tangent vector fields on the submanifold}]{tansubm}{\ensuremath{\mathcal{T}(\subm)}}
\symb[description={Set of orthogonal vector fields along the submanifold}]{orthsubm}{\ensuremath{\mathcal{N}(\subm)}}
\symb[description={Riemannian Curvature of the Ambient Manifold}]{rieamb}{\ensuremath{\widetilde{\rmt}}}
\symb[description={Riemannian Curvature of the Submanifold}]{riesubm}{\ensuremath{\rmt}}
\symb[description={Ricci Tensor of the Ambient Manifold}]{ricamb}{\ensuremath{\widetilde{\ric}}}
\symb[description={Ricci Tensor of the Submanifold}]{ricsubm}{\ensuremath{\ric}}
\newcommand{\rtil}{\tilde{R}}
\symb[description={Isometry group}]{isom}{\ensuremath{Iso}}

\symb[description={Complex Numbers}]{C}{\ensuremath{\mathbb{C}}}
\symb[description={Real Numbers}]{R}{\ensuremath{\mathbb{R}}}
\symb[description={n dimensional Euclidean space}]{Rnn}{\ensuremath{\mathbb{R}^n}}
\symb[description={closed n dimensional upper half Euclidean space}]{Rnp}{\ensuremath{\mathbb{R}^n_{+}}}
\symb[description={Lie Derivative}]{lie}{\ensuremath{\mathcal{L}}}
\symb[description={inclusion}]{inc}{\ensuremath{\lrcorner \,}}
\symb[description={Order (of Magnitude)}]{ord}{\ensuremath{\mathcal{O}}}
\symb[description={Hessian}]{hes}{\ensuremath{\mathcal{H}}}
\symb[description={Lagrangian}]{lag}{\ensuremath{\mathit{L}}}
\symb[description={Off-Shell Action}]{act}{\ensuremath{\mathit{I}}}
\symb[description={Spacetime Manifold}]{stm}{\amb}
\symb[description={Configuration Space}]{consp}{\ensuremath{\mathcal{F}}}
\symb[description={Space of Vacua}]{minsp}{\ensuremath{\mathcal{E}}}
\symb[description={Field of Minima, an element of \minsp}]{minf}{\ensuremath{\varphi}}

\newcommand{\partang}{\delta_{\parallel} \varphi}
\newcommand{\partan}[2]{\delta_{\parallel} \varphi^{(#1)#2}}
\newcommand{\dpartan}[2]{\delta_{\parallel} \dot{\varphi}^{(#1)#2}}
\newcommand{\orthtang}{\delta_{\perp} \varphi}
\newcommand{\orthtan}[2]{\delta_{\perp} \varphi^{(#1)#2}}
\newcommand{\dorthtan}[2]{\delta_{\perp} \dot{\varphi}^{(#1)#2}}

\newcommand{\ez}{e_{0}}
\newcommand{\ei}{e_{i}}
\newcommand{\ej}{e_{j}}
\newcommand{\ek}{e_{k}}
\newcommand{\emm}{e_{m}}

\newcommand{\wz}{\omega^{0}}
\newcommand{\wi}{\omega^{i}}
\newcommand{\wj}{\omega^{j}}
\newcommand{\conambz}{\conamb_{\ez}}
\newcommand{\conambi}{\conamb_{\ei}}
\newcommand{\conambj}{\conamb_{\ej}}
\newcommand{\conambk}{\conamb_{\ek}}
\newcommand{\consubmi}{\consubm_{\ei}}
\newcommand{\consubmj}{\consubm_{\ej}}

\symb[description={Acceleration}]{acc}{\ensuremath{\kappa}}
\symb[description={Two Connection}]{twoconn}{\ensuremath{\mathscr{f}}}

\symb[description={Ambient tensor with zero orthogonal component}]{ta}{\ensuremath{\mathbf{t}}}
\symb[description={Ambient tensor with zero tangential component}]{no}{\ensuremath{\mathbf{n}}}
\symb[description={k-forms that are square integrable with respect to $\ket \ket,\bra \bra$}]{Lkform}{\ensuremath{H^1\Omega^k}}
\symb[description={Harmonic k- fields}]{harmk}{\ensuremath{\mathcal{H}^k}}

\newcommand{\alm}{\mathfrak{m}}
\newcommand{\alk}{\mathfrak{k}}

\symb[description={Hamiltonian Constraint}]{ham}{\ensuremath{\mathrm{H}}}
\symb[description={Momentum Constraint}]{momi}{\ensuremath{\mathrm{P}^i}}
\newcommand{\mom}{\mathrm{P}}

\symb[description={Metric on the Space of Vacua for metrics}]{metonvacua}{\ensuremath{\mathrm{g}}}

\symb[description={Spatial Metric}]{spmet}{\ensuremath{h}}
\symb[description={A Spatial Metric that is in Vacua}]{vacspmet}{\ensuremath{\mathscr{h}}}
\symb[description={Reference Metric that has a zero Ricci Curvature}]{refmet}{\ensuremath{\bar{h}}}
\symb[description={Extrinsic Curvature of a Spatial Slice in the Spacetime Manifold}]{ext}{\ensuremath{K}}
\symb[description={Covariant Derivative of the Spatial Slice}]{covs}{\ensuremath{D}}
\symb[description={Potential}]{pot}{\ensuremath{V}}
\symb[description={Curvature of Spatial Slice}]{rthree}{\ensuremath{R}}
\symb[description={Spatial Manifold}]{sm}{\ensuremath{M}}
\symb[description={Boundary of the Spatial Manifold}]{bnd}{\ensuremath{\pr \sm}}
\symb[description={Metric on Boundary}]{metbnd}{\ensuremath{k}}
\symb[description={Riemannian Connection on for metric on the boundary}]{conbnd}{\ensuremath{D}}
\symb[description={Extrinsic Curvature of Boundary}]{extbnd}{\ensuremath{\mathscr{K}}}
\symb[description={Two Connection of boundary}]{twoconnbnd}{\ensuremath{\mathscr{Y}}}
\symb[description={Acceleration for boundary}]{accbnd}{\ensuremath{\mathscr{a}}}
\symb[description={Orthogonal Derivative on boundary}]{ortder}{\ensuremath{\conbnd^{\perp}}}

\newcommand{\chip}{\chi^{\perp}}
\newcommand{\xip}{\xi^{\perp}}
\newcommand{\chit}{\chi^{\parallel}}
\newcommand{\xit}{\xi^{\parallel}}
\newcommand{\chia}{\chi^a}
\newcommand{\chib}{\chi^b}

\newcommand{\ea}{e_{a}}
\newcommand{\eb}{e_{b}}

\def\uk{{\underline{k}}}

\def\um{{\underline{m}}}

\def\tx{{\tilde{x}}}

\newcommand{\sqg}{\sqrt{\gamma}}
\newcommand{\lpa}{\lp}
\newcommand{\rpa}{\rp}

\newcommand{\bdiffc}{\bdiff_C}

\newcommand{\abdiff}{\mathfrak{bdiff}}
\newcommand{\abdiffz}{\mathfrak{bdiff}_0}
\symb[description={Group of Diffeomorphisms of the Manifold $M$}]{diffm}{\ensuremath{\diff(M)}}
\symb[description={Group of Boundary preserving Diffeomorphisms of the Manifold with boundary $M$}]{bdiffm}{\ensuremath{\bdiff(M)}}
\symb[description={Space of vacua}]{truevac}{\ensuremath{\mathcal{V}}}

\symb[description={Radial Vector Spherical Harmonic}]{vy}{\ensuremath{\vec{Y}_{lm}}}
\symb[description={Longitudinal Vector Spherical Harmonic}]{vpsi}{\ensuremath{\vec{\Psi}_{lm}}}
\symb[description={Transverse Vector Spherical Harmonic}]{vphi}{\ensuremath{\vec{\Phi}_{lm}}}
\def\vn{\vec{\nabla}}

\def\xr{\xi^\perp_{lm}}
\def\xd{\xi^{L}_{lm}}
\def\xt{\xi^{T}_{lm}}

%% file: myacronyms.tex
\newcommand{\acryn}[3]{%
  \newacronym{#1}{#2}{#3}%
  \expandafter\newcommand\csname #1\endcsname{\acs{#1}}}%

\acryn{wt}{WT}{Ward-Takahashi}
\acryn{gr}{GR}{General Relativity}
\acryn{cft}{CFT}{Conformal Field Theory}
\acryn{bms}{BMS}{Bondi-Metzner-Sachs}
\acryn{frw}{FLRW}{Friedmann–Lemaître–Robertson–Walker}
\acryn{cmb}{CMB}{Cosmic Microwave Background}
\acryn{ym}{YM}{Yang-Mills}
\acryn{rhs}{R.H.S.}{Right Hand Side (of the equation)}
\acryn{lhs}{L.H.S.}{Left Hand Side (of the equation)}
\acryn{gnc}{GNC}{Gaussian Normal Coordinates}
\acryn{adm}{ADM}{Arnowitt-Deser-Misner} %(action)
\acryn{eh}{EH}{Einstein-Hilbert} %(action)
\acryn{hmf}{HMF}{Hodge-Morrey-Friedrichs} %decomposition
\acryn{wdw}{WdW}{Wheeler-deWitt} %metric

%% file: titlepage.tex
\begin{titlepage}
\begin{center}
\vspace*{1cm}
 
\LARGE
ADIABATIC SOLUTIONS IN GENERAL RELATIVITY \\ AND BOUNDARY SYMMETRIES
 
\large

\vspace{2.5cm}

\normalsize
A THESIS PRESENTED FOR THE FULFILLMENT \\ 
OF THE REQUIREMENTS FOR \\
THE DEGREE OF DOCTOR OF PHILOSOPHY

\vspace{1cm}
 
 \Large
       BY EMİNE ŞEYMA KUTLUK\\
       \vspace{0.3cm}
       \normalsize
       TO

\vspace{1cm}
 
       \includegraphics[width=0.6\textwidth]{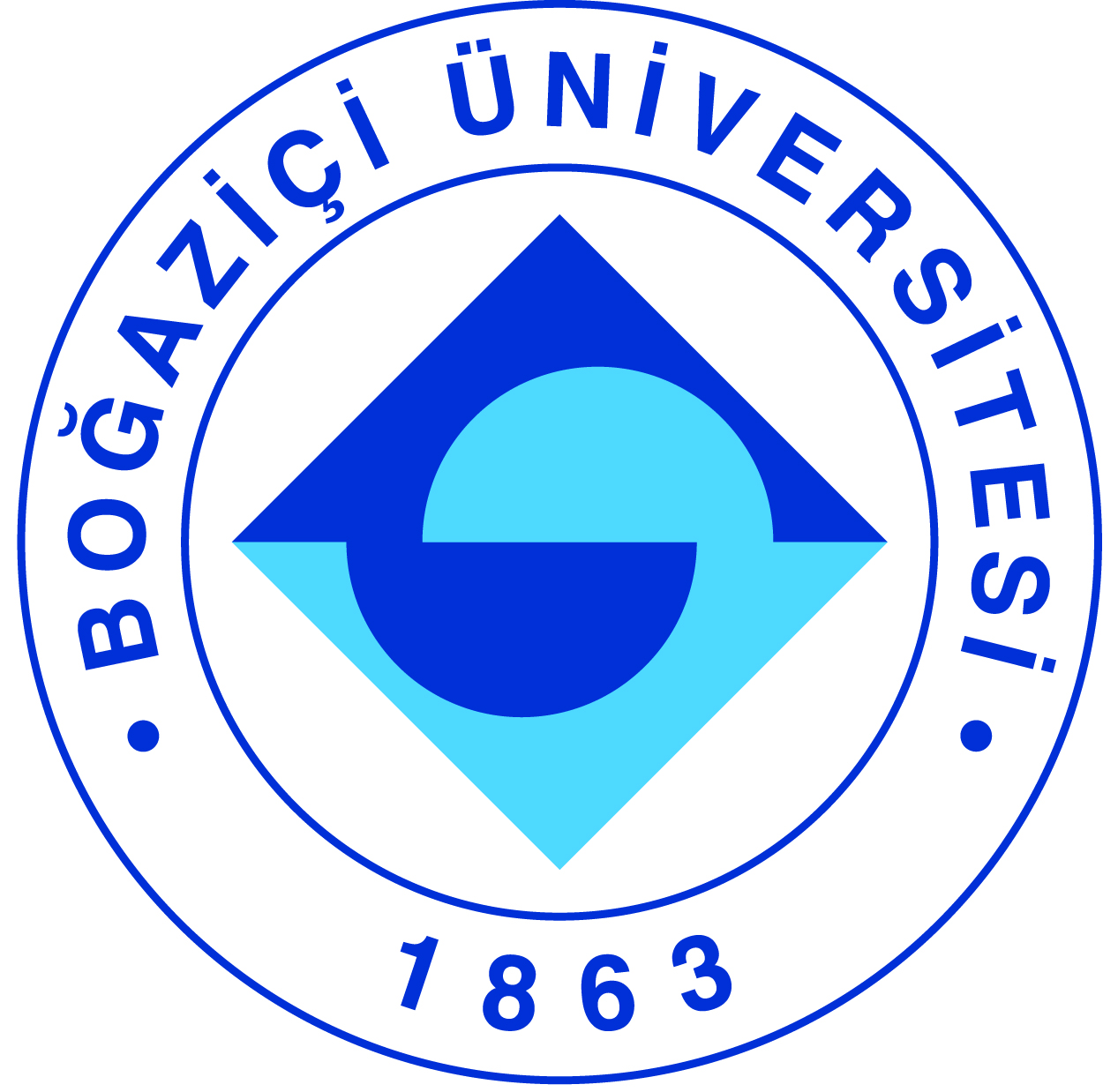}       
       
 \vspace{0.5cm}

       PHYSICS DEPARTMENT\\ 
 
  		\vspace{1cm}
     
       07/01/2020

   \end{center}
\end{titlepage}

%% file: preamble.tex
\chapter*{}
\vspace*{\fill}
\begin{flushright}

\textit{ Man is hidden behind his words,\\
his tongue is a curtain over the door of his soul. \\
When a gust of wind lifts the curtain\\
the secret of the interior is exposed,\\
you can see if there is gold or snakes,\\
pearls or scorpions hidden inside.\\
Thoughtless speech spills easily out of man\\
while the wise ones keep silent.\\
Fault eyes see the moon double\\
and that gazing in perplexity is like a question;\\
once you connect with Divine Light\\
the question and the answer become one.\\
\hfill\\
Meulana Jalaluddin Rumi\\
(translated by Maryam Mafi and Azima Melita Kolin \\
Rumi's Little Book of Life)}
\end{flushright}
\vspace{5cm}

%% file: acknowledgements.tex
\chapter*{Acknowledgements}
This has been a long and bumpy road for me, that had crossed many others'. I would like to start by thanking my family, especially my mother and father, who had been long waiting to see the completion of this study and bear its psychological burdens with me. I would like to thank all of my former teachers throughout my educational life that built the road to this point, for their devotion, care and belief in their students.

I would like to thank my former PhD advisor Ali Kaya. I have learned not only quantum cosmology which was my the starting point of PhD, but much of the physics I know from him. I appreciate his student-friendliness, clarity and consistency as a physicist.

Most importantly for this thesis, I would like to thank my current PhD advisor, Dieter Van den Bleeken, for taking me as his student and suggesting this problem, the backbone of this work mostly belongs to him.
His enthusiasm combined with his gentlemanliness were very valuable to me and elevated my PhD experience. I would also take this chance to thank him for all the departmental activity that much benefited the students of theoretical physics.

I would like to thank many other faculty members of the theoretical physics community in Turkey I have learned much from: Sadık Değer, Bayram Tekin, Levent Akant, Can Kozçaz, Ilmar Gahramanov, Teoman Turgut, Metin Arık and others I might have forgotten. I would like to thank my remaining jury members, Mehmet Özkan and especially İbrahim Semiz, for his many corrections and suggestions on the draft of this thesis.

I would like to thank my friends at the office of High Energy Theory group at Boğaziçi University: Zainab Nazari, Birses Debir, Delaram Mirfendereski, Cihan Pazarbaşı, Canberk Şanlı, Emine Ertuğrul, Şelale Şahin, Sema Seymen, Narçiçeği Kıran, Kerem Kurşun, Metin Güner and my friends from my previous research group Merve Tarman Algan and Merve Uzun, I have greatly enjoyed and learned from their company. I would like to thank former postdocs of our group: Ilies Messamah, George Moutsopoulos, Nicolo Petri. I also would like to thank Dr. Ali Seraj which we collaborated with for the work that will be topic of this thesis, for many useful discussions.

I would like to thank my master's institute Columbia University, for their alumni benefits, especially the access to online sources of the library, and Feza Gürsey institute of Boğaziçi University where I have performed some part of this study in. I would like to thank the secretariat of Boğaziçi University Physics department, Mehtap Tınaz and Yılmaz Kurşun for their technical help. During the work of this thesis I was supported by TUBITAK initially with 2211-A Fellowship and later with grant 117F376.

%% file: abstract.tex
\begingroup
\renewcommand{\cleardoublepage}{}
\renewcommand{\clearpage}{}
\chapter*{Preface to This Version}
This version of the thesis is prepared for submission to arxiv repository, and differs with the version submitted to the Bogazici University Institute of Graduate Studies in Science and Engineering \cite{tez-yok} only in the printing style and correction of minor grammatical and punctuation errors. Purpose of this submission is to make an easier-to-read version of the thesis available. References to the text itself can be reached by clicking on the relevant link-shown with gray color-in suitable pdf readers: not only equations and sections are cross-referenced this way, but also certain technical terms. Errata and other comments can be communicated via the email \includegraphics[height=0.8em]{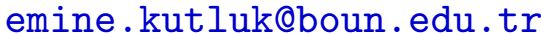}.

\chapter*{Abstract}
\endgroup
We investigate adiabatic solutions to general relativity for a spacetime with spatial slices with boundary, by Manton approximation. This approximation tells us for a theory with a Lagrangian in the natural form, a motion that is described as a slow motion on the space of vacua-static solutions that minimize the energy- is a good approximate solution. To apply this to the case of general relativity we first bring it to the natural form by splitting space and time and choosing Gaussian normal coordinates, where a spacetime is described by the metric on its spatial slices. Then following Manton we propose slow solutions such that each slice is a slowly changing diffeomorphism of a reference slice, and thus each solution is described by a vector field on the spatial slice. 

These solutions will have the property that the action will become a functional of the vector fields on the boundaries of the spatial slices. Moreover using the Hodge-Morrey-Friedrichs decomposition we will show that the constraints of general relativity will identify a unique solution for a given boundary value. Then we comment on the structure of the space of vacua which we show to be a (pseudo)-Riemannian homogeneous space. We illustrate our procedure for a specific reference slice we choose: the 3d Euclidean round ball.

%% file: introduction.tex
%Symmetries and Conserved quantites-QFT:Ward identities / Soft theorems: Ward identities for asym symmetries.......: Explain these in words-and as non technically as possible

%This will be a chapter where I give the general picture with much intuitional insights but without proving my statements. Here I should motivate the reader to read the thesis. I might rewrite this after I finish writing the remaining chapters because I would have more complete overall picture after that. Actually I might have additional results etc if I have time.

%Main messages I want to give: *Symmetries helps us to make physical predictions without knowing the theory exactly. This is the case for exact symmetries in mechanics, and very important in conformal field theory and also cosmology. CFT: this is the core of the bootstrap approach. Cosmology: we have many theories and need to eliminate them by observations / also GR is very nonlinear, if we do cosmological perturbation theory many complicated higher order terms appear.
%* Asymp -> soft theorem 

\chapter{Introduction: Symmetries There, Softness Here}

A physicist's aim is to describe the fundamental principles governing the involuntary events in the universe. A symmetry is yet another higher principle that governs these laws \cite{Wigner956}. The power of this higher principle has increasingly manifested itself as fundamental interactions are better understood i.e. as the theory of quantum fields, conformal fields and the like were developed.

Existence of a symmetry principle allows one to extract certain information about the physical events from the theory without knowing the details of it. For example, if one knows that the motion of a particle is governed by a theory that doesn't depend on one of the spatial coordinates; then one can say, without having further information about the theory, that the particle will have a constant speed along that coordinate. This reflects a fact known as Noether's theorem: for every symmetry there exists a conserved quantity. In quantum field theory Noether's theorem takes the form of Ward-Takahashi (\wt) identities, which gives certain conservation like constraints on the expectation values. 

\wt~identities have been used in the literature many times to extract information about the expectation values. The most prominent example of this is in conformal field theory. There it has been shown that conformal symmetries-especially in 2D where they are an infinite tower- fix the form of all expectation values to a great extent \cite{Blumenhagen:2009zz}.

However not every theory has a set of exact symmetries. Still one can talk about the asymptotic symmetries the theory has. For these, one can also write \wt~identities, which results in what are called soft theorems. While these \wt~identites were studied by Strominger \cite{stromingersoftgraviton,stromingersoftphoton} to derive soft photon and soft graviton theorems; independently in the area of quantum cosmology a similiar type of theorems were derived using the ``adiabatic mode" condition on ``large gauge symmetries". The relationship between these two lines of research has been a question since then. There exists certain elements that appear in one of the analyses while it does not appear in the other. While a discussion of fall-off conditions is crucial in the asymptotic symmetry analysis, it was not considered at all in the original cosmology research, except the attempts that came after Strominger's work  \cite{Kehagias:2016zry,Ferreira:2016hee,flatland}. Same is true for adiabatic modes in the opposite direction.

In this thesis we detail the work by Van den Bleeken, Seraj and myself in \cite{grmanton}, which was a continuation of a previous work by Bleeken and Seraj in \cite{ymmanton}. With motivation from the concept of adiabatic modes in cosmology, we take a third look at the question of asymptotic symmetries and soft theorems. We apply the Manton approximation \cite{manton,WeinbergMono} to general relativity to derive solutions that are slow in time\que{But we never use slowness actually!!}. Manton approximation gives these solutions as geodesics on the space of vacua of the theory.  These works are purely classical and instead of the IR/infinite volume limit where the study of asymptotic symmetries and soft theorems takes place, they study spacetimes with bounded spatial volume. The study of an infinite volume limit and precise connection to the asymptotic symmetries and soft theorems are important questions left for future work.

We will review the aspects of the thesis at the end of this introduction, but let us first discuss the underlying research mentioned above that motivated us, in more detail.

\section{Asymptotic Symmetries and Soft Theorems}

These type of symmetries were first discussed in the context of gravitational waves \cite{sachs,bms}, where it was concluded that, to produce gravitational waves a spacetime should satisfy certain asymptotic conditions at future null infinity. Coordinate transformations that leave these conditions invariant were shown to form a group, called the Bondi-Metzner-Sachs (BMS) group. This group is interestingly larger than the Poincare group, it additionally includes an infinite set of symmetries that are called supertranslations. 

One can define different types of asymptotic symmetries: asymptotic symmetries of spacetimes that are asymptotically flat at spatial infinity is a different group known as the \textit{Spi} group \cite{ashtekarhansen}, whereas definitions of asymptotic symmetries of the spacetimes that are asymptotically de Sitter were also made \cite{ashtekarDS, anninosDS}.

One advantage of the definition of asymptotic flatness is that, through it one can define the charges of general relativity. As it is well known, see e.g.\cite{wald}, since there is no prescribed background for \gr, the concept of total energy is non-trivial there. However for asymptotically flat spacetimes, such energy can be defined since these spacetimes can be thought of as isolated systems.

Roughly around the same time with these studies, the soft photon and graviton theorems were proposed by Weinberg \cite{weinbergsoft}, restricting the form that the quantum expectation values with a soft photon/graviton should take. To be more specific, for example the soft photon theorem says that
\begin{equation}
\mathcal{M}^{\mu}(p^{in}_1,...,p^{out}_1,...,q) \xrightarrow[q \ra 0]{} \mathcal{M}(p^{in}_1,...,p^{out}_1,...) \lp \sum_n \frac{e_n {p_n^{out}}^\mu}{p^{out}_n \cdot q - i \epsilon} - \frac{e_n {p_n^{in}}^\mu}{p^{in}_n \cdot q + i \epsilon} \rp \sgc
\end{equation} 
where the $\mathcal{M}$' s represent the S-matrix for the process with related incoming and outgoing particles without the energy-momentum conservation delta function and $e_n$ is the electrodynamical charge of the nth particle. Contracting with the soft momentum gives
\begin{equation}
q_\mu \mathcal{M}^{\mu}(p^{in}_1,...,p^{out}_1,...,q) \xrightarrow[q \ra 0]{} \mathcal{M}(p^{in}_1,...,p^{out}_1,...) \lp Q_{out}^{total}- Q_{in}^{total} \rp \sgd
\end{equation}
For Lorentz invariance to hold, the left hand side should go to zero, see \cite{weinbergQFT} pg.537 for the reason behind this. Because of this the right hand side gives the charge conservation. This calculation is then interpreted as showing that the Lorentz invariance necessitates the electrodynamic charge conservation. A similar story can be repeated for the graviton, which shows that ``Lorentz invariance requires that low energy massless particles of spin two couple in the same way to all forms of energy and momentum" \cite{weinbergQFT}, which essentially is Einstein's equivalence principle.

The soft-limit formula written above is shown, much later by Strominger to be the \wt~identity for asymptotic symmetries. More specifically in \cite{stromingersoftphoton}, an infinite number of symmetries at null infinity, just like supertranslations, of electrodynamics was found; and Weinberg's soft photon theorem was shown to be the \wt~identity related to this. Similarly in \cite{stromingersoftgraviton} the Weinberg soft graviton was shown to be the \wt~identity for supertranslations.

\section{Cosmological Consistency Relations}

Cosmic Microwave Background (\cmb) observations are arguably the most important physical data about our universe yet: they provide information about the early phases of the cosmological history, and on top of that, as these observations are understood to be the relics of the expectation values of quantized gravitational fields, they can be considered as data for a quantum theory of gravity \cite{woodard}. Because of this it is crucial to get the most out of these observations as a test of our theories.

Any perturbative quantum field theory is expected to give a Gaussian spectrum for the correlation functions, also called N-point functions, and this is indeed the case for the \cmb. The next step is to check our theories for the corrections to this Gaussianity -which was measured to be small- by means of three-point functions. In \cite{maldacena}, it was proposed that a three point function should satisfy
\begin{equation}\label{eq:consis}
\underset{q \rightarrow 0}{lim} \frac{1}{\left< \zeta_q \zeta_{-q} \right>}\left< \zeta_q \zeta_{k_1} \zeta_{k_2} \right> = -  \vec{k_1}.\frac{\partial}{\partial \vec{k_1}} \left< \zeta_{k_1} \zeta_{k_2} \right> 
\end{equation}
at the soft limit, for a single-field inflationary model. This type of relations are called the cosmological consistency relations, see also \cite{creminelli-1} for some generalization of this to a larger class of theories. Consistency relations are particularly important since, because of the highly non-linear nature of gravity theories, it is hard to calculate expectation values with high number of fields. These relations help us to alleviate our lack of knowledge of the theory by our knowledge of the symmetries.

The relation we wrote down above was later understood to be a part of the \wt~identities related to the ``adiabatic modes" subclass of ``residual gauge symmetries" of the cosmological spacetimes studied in a specific gauge \cite{creminelli-2,hui-1}. Adiabatic modes are solutions to Einstein equations linearized around a cosmological background, such that the physical degrees of freedom of the metric remain constant in time in the limit of small momenta/large wavelengths-wavelenghts that are of cosmological interest. In cosmology it is important to show that adiabatic modes exist for an arbitrary matter content of the universe, because this tells us that the large wavelength modes ``freeze out" through the cosmic evolution following the inflation. Only if this is true we can have inflationary predictions without going into the details of the following cosmic history, often modeled as a period called reheating, of which not much is known. The existence of adiabatic modes for an arbitrary content of matter, was proven by Weinberg in \cite{Weinberg:2003sw}, under mild assumptions. In his proof he uses a trick to extend pure gauge solutions into physical solutions by imposing the constraint equations of general relativity . The work of \cite{creminelli-2,hui-1} shows that the relation in equation \eqref{eq:consis} and an infinite set of other relations of the same form are \wt~identities for adiabatic modes. Extending Weinberg's trick, they obtain an infinite set of adiabatic modes and write down \wt~identities for these solutions. Here we note that our starting point will be very much in the same spirit with Weinberg's work, and because of this we will be reviewing this argument in the Section \ref{ch:admodes}.

This realization of consistency relations as \wt~identities is very similar to Strominger's realization of Weinberg's soft theorems as \wt~identities. A connection between these two has been made in \cite{mirbabayi}, where an analogue of Weinberg's adiabatic mode, identified as a residual symmetry of the U(1) theory, gives Weinberg's soft photon theorem and some peculiar technicalities of Strominger's work \cite{stromingersoftphoton} were shown to follow from this construction. Similarly one expects to see the cosmological consistency relations obtained by Weinberg's adiabatic mode argument to be related to the asymptotic symmetries of the cosmological space-times, see e.g. \cite{Kehagias:2016zry,Ferreira:2016hee} for attempts at this, and especially also \cite{flatland}. Although one can be encouraged by the aforementioned U(1) gauge theory study to say this generalization should be straightforward, symmetries and constraints of gravity have somewhat different structure \cite{LeeWald}; and as it can be seen from the soft graviton case of Strominger \cite{stromingersoftgraviton}, a more detailed description of charges related to these symmetries \cite{barnich} had to be used. Moreover, a mathematically more rigorous study of asymptotic symmetries of the flat spacetime was also needed. 
 
\section{Review of this Thesis}

With this background in our mind, what we explore in this thesis is the Manton/adiabatic approximation for general relativity (\gr) on a manifold with a spatial boundary. The Manton approximation tells that a proposed solution described as a slow motion on the configuration space of vacuum solutions to a theory is a good approximation. Similarity of this to Weinberg's proof of the existence of adiabatic modes lies in the fact that a space of vacuum solutions of \gr~is generated by gauge transformations. However, as we will show, where he starts from a space independent solution and then to get the physical modes assumes a small space dependency, we start from a time independent solution and then assume small time dependency. In our analysis constraints play an important role in determining the structure of the space of vacua. 

In the following we start by going over Weinberg's proof for the existence of adiabatic modes in Section \ref{ch:admodes}. Then in Section \ref{sec:setup} we develop the Manton approximation procedure for a generic Lagrangian in the natural form, and show that if the motion is 1) adiabatic, 2) on the space of vacua at zeroth order, then at first order there is no motion off the space of vacua, and the motion is described by the geodesics on the space of vacua with respect to a metric specified by the Lagrangian. We give a simple example to illustrate the procedure, and then switch to summarizing the work \cite{ymmanton}, the Manton approximation for the case of \ym~on a Minkowski space with spatial slices with boundary in Section \ref{sec:ymills}. This section will show us how an adiabatic mode like argument can be used to apply the Manton approximation to a gauge theory. Moreover the study of adiabatic solutions will give us information about the shape of the space of vacua, which will be a Riemannian homogeneous space as we will discuss.

To apply the procedure to the case of \gr, several mathematical tools will be needed. To bring the Einstein-Hilbert (\eh) action to the natural form, one needs to foliate the spacetime. While we discuss the theory of foliations on a differentiable manifold in Appendix \ref{sec:appfol}, we use this theory for (pseudo)-Riemannian manifolds in Section \ref{sec:Fol} to have a decomposition of quantities of interest into directions orthogonal and tangent to the foliation. Since we will take each leaf of foliation of the spacetime to be a spatial manifold with boundary just as the case of \ym, in Section \ref{sec:manbnd} we review the theory of manifolds with boundary. Then in Section \ref{sec:hmf}, we discuss the generalization of the Hodge decomposition to manifolds with boundary: the Hodge-Morrey-Friedrichs (\hmf) decomposition. We will use this theorem later to solve the momentum constraint of general relativity for our adiabatic solutions.

As discussed in the \ym~case where it turned out to be a Riemannian homogeneous space, the structure of the space of vacua will be investigated by adiabatic solutions. For this reason in Section \ref{sec:homsp} we explore the general setting where such manifolds exist: Homogeneous spaces with a geometry. These will be spaces that have a transitive action of a group on them and a metric, and if the metric is invariant under the group action the space will be called a Riemannian homogeneous space. For our case the action will be the action of diffeomorphisms, and metric will be the metric induced on the space of vacua by the Lagrangian.

Equipped with the theory of foliations, in Chapter \ref{ch:natgr} we take on the mission to bring \gr~into to the natural form. First we show how foliations bring the \eh~action into the Arnowitt-Deser-Misner (\adm) form. Then we discuss the constraints of general relativity, since these can be easily seen in the \adm~form; and some peculiarities involving them. Then we discuss the Gaussian Normal Coordinates (\gnc), the coordinates in which \gr~ is in the natural form: kinetic and potential energy separated. Then we discuss the remaining gauge transformations of this gauge choice, as via a Weinberg like argument, these are supposed to be forming our adiabatic solutions. We will see not only that a part of these transformations is subtle and has commutation relations that are not in compliance with the rules of a Lie algebra, reflecting the commutation relations of the constraints of \gr; but also that this part has a conceptual difference compared to other parts.

Choosing a simpler subset of these transformations that do form a Lie algebra, in Chapter \ref{ch:adgr} finally we perform the Manton approximation for \gr. After explicitly constructing the specific form of the proposed solutions in Section \ref{sec:adgr}, using the constraint equations for these solutions and orthogonal decomposition of quantities on the spatial manifold we simplify the action in Section \ref{sec:constaction}. Indeed it turns out as an action on the boundary of the spatial manifold, as we prove in Section \ref{sec:boundarydata} by using the \hmf~decomposition. In Section \ref{sec:decconst}, in an attempt to explicitly solve the constraints for the solutions we propose, we study them for our type of solutions by decomposing them orthogonally on the spatial manifold and then simplifying them under some assumptions.

Having studied the spacetime properties of solutions we have defined, we move onto a study of the space of vacua. First we compare our metric on the the space of vacua to the Wheeler-deWitt (\wdw) metric, and discuss its signature. In the next section we explicitly show that our space of vacua is a homogeneous space under the action of boundary diffeomorphisms with an isotropy group that is equivalent to the group of isometries of the spatial slices. Moreover we argue that the metric we have defined is invariant under this action, so that the space of vacua is indeed a (pseudo)-Riemannian homogeneous space. 

In Chapter \ref{ch:ball}, we study a specific example. We consider solutions on 4d spacetime with slices that are 3d round balls with flat metric inside. We employ vector spherical harmonics as a tool to study these, as these are well adapted to a round boundary and also help us to better visualize the solutions. We find explicit solutions to the momentum constraint, and write down the action as a summation over the spherical harmonic coefficients of the solutions. We complete the chapter by discussing the Hamiltonian constraint and the homogeneous space structure. 

We finalize the thesis with Chapter \ref{ch:discussion}, where we summarize the results and discuss open issues and future directions. Our accompanying appendix is designed as a summary of relevant mathematical concepts. I hope it to bring together many definitions and theorems used often in theoretical physics under a unified language. It should be consulted whenever the need occurs while reading the text, but can also be studied quite independently for general purposes.\\

\begin{remark}
In addition to the paper \cite{grmanton} which will be the subject of this thesis, during my doctoral studies I have also contributed to the work in \cite{kaya1,kaya2,kaya3}, that deals with the correlation functions during the reheating period mentioned above.
\end{remark}

%% file: adiabaticmodes.tex
\chapter{Motivation and Setup}

\section{Adiabatic Modes in Cosmology}\label{ch:admodes}
The universe is thought and partially observed to be homogeneous and isotropic at large scales. From Cosmic Microwave Background (\cmb) observations, we see there also exist very small inhomogeneity and anisotropy \cite{baumann}. The origin of these perturbations are thought to be the quantum perturbations during the inflationary era. However many eras of cosmic history have passed from the time of inflation till today where these observations are made. If one is to test a theory of inflation by these observations, one should find a way to connect the quantities in the inflationary period to those of today.

Even though later eras of radiation, matter and dark energy domination that were proposed are relatively established, much less is known about what follows the inflation. What Weinberg showed was whatever the constituents of the universe, there always exist solutions to perturbed Einstein equations such that the physical degrees of freedom are constant for large wavelengths \cite{Weinberg:2003sw}, see also \cite{baumann} which we will also benefit from for the following arguments.

To show this statement one starts with a general metric perturbed around the spatially flat Friedmann–Lemaître–Robertson–Walker  (\frw) background:
\begin{equation}
ds^2= - (1+2 \Delta N) dt^2 + 2 N_i dx^i dt + \left( a^2 \delta_{ij} + a^2 \Delta h_{ij} \right) dx^i dx^j \, .
\end{equation}
Here, $\Delta N,N_i,\Delta h_{ij}$ are to be perturbations. It is conventional to split these into components that are in the irreducible representations of the $SO(2)$ little group:
\begin{align}
N_i &= \partial_i \psi + N_{i}^{T} \, , \\
\Delta h_{ij} &=  \left(  2 \zeta \, \delta_{ij} + \left( \partial_i \partial_j - \frac{\delta_{ij}}{3} \partial^2 \right) \gamma + { \partial_{ (i }{\gamma  }_{ j) } }^{T}  + \gamma_{ij}^{TT} \right) \, . 
\end{align}
where
\begin{align}
\partial_i {\gamma  }_{ i }^{T}=0 \, ,  \quad \gamma_{ii}^{TT}=0 \, , \quad \partial_i \gamma_{ij}^{TT} =0 \sgd
\end{align}

$\Delta N, \psi, \zeta, \gamma$ are scalars (i.e has helicity 0) under $SO(2)$ transformations, while $N_i^T,\gamma_j^T$ are vectors and $\gamma_{ij}^{TT}$ is tensor. Because of the $SO(2)$ symmetry of the background different helicity components become decoupled in the Einstein equations \cite{baumann}. Similarly we write down the energy-momentum tensor as
\begin{align}
T^0_0&=-\lp \bar{\rho} + \Delta \rho \rp \cgap T^0_i= \lp \bar{\rho} + \bar{p} \rp \lp \pr_i v + v_i^{T} \rp \cgap \\
T^i_j&= \delta_{ij} \lp \bar{p} + \Delta p \rp + \left( \partial_i \partial_j - \frac{\delta_{ij}}{3} \partial^2 \right) \sigma + { \partial_{ (i }{\sigma  }_{ j) } }^{T}  + \sigma_{ij}^{TT} \cgap
\end{align}
where $\sigma$ quantities carry the same properties as $\gamma$s; and an arbitrary infinitesimal coordinate transformation as
\begin{equation}
\xi_{\mu}=g_{\mu \nu} \xi^{\mu} = \lp \xi_0, \pr_i \xi + \xi_i^T \rp \sgd
\end{equation}

Because of decoupling of the scalar, vector, tensor modes; in the following we only consider the scalar modes. Under an infinitesimal coordinate transformation they transform as
\begin{alignat}{6}
& \Delta N && \rightarrow && \delta N - \dot{\xi_0}\sgc &\quad \quad  & \zeta && \rightarrow && \zeta +  \left( -H \xi_0 + \frac{1}{3} \frac{\partial^2 \xi}{a^2}  \right) \sgc \nonumber \\
& \gamma && \rightarrow && \gamma + 2 \frac{\xi}{a^2} \sgc &\quad \quad & \psi && \rightarrow && \psi + \left( \dot{\xi} + \xi_0 - 2H \xi \right) \sgc \nonumber \\
& \Delta \rho && \ra && \Delta \rho  - \dot{\bar{\rho}} {\xi}_0  \sgc & \quad \quad &  \Delta p && \ra && \Delta p - \dot{\bar{p}} {\xi}_0  \sgc \nonumber \\
& \sigma && \ra && \sigma \sgc & \quad \quad &  v && \ra && v + {\xi}_0 \sgd
\end{alignat}

Using this freedom one can set two scalars of the metric to zero. If one chooses $\gamma=0, \psi=0$ this is called the Newtonian gauge. With this choice there is no gauge freedom left, except for the case of spatial homogeneity of the metric. For this case the remaining gauge freedom is
\begin{equation}
\xi_0= -\epsilon(t) \cgap \frac{\xi}{a^2}= f(t)+ c_i x^i + c x^2 \cgap
\end{equation} 
where $c,c_i$ are constants.

Now $g=\lie_{\xi} \bar{g}$ , $T=\lie_{\xi} \bar{T}$ are solutions to Einstein equations,where $\bar{g},\bar{T}$ are background solutions. Thus we say
\begin{align}\label{eq:ad sln}
\Delta N= \dot{\epsilon} \,,\,\, \zeta= H \epsilon -2 c \,,\,\, \gamma=0 \,,\,\, \psi=0 \,,\,\, \Delta \rho = \dot{\bar{\rho}} \epsilon \,,\,\, \Delta p = \dot{\bar{p}} \epsilon \,,\,\, \sigma=0 \,,\,\, v= - \epsilon
\end{align}
is a solution to Einstein equations. This is gauge equivalent to a trivial solution, but we use the following trick to get a physically nontrivial solution: Introduce a small space dependency into the gauge parameter $\epsilon$. Einstein equations will be no longer trivially satisfied, instead one should solve $\epsilon$ for them.  Now these solutions will be non trivial solutions, as introducing the small space dependency took us out of the complete homogeneity case, and for us now the gauge is completely fixed. 

So let us try to solve Einstein equations with \eqref{eq:ad sln} where now we also introduce a small spatial dependency into $\epsilon$. We will first write down the linearized equations for the scalar modes in Planck units:
\begin{align}
\Delta p - \Delta \rho + \frac{2 \pr^2 \sigma}{3} &= H \Delta \dot{N} + \lp 6H^2 + 2 \dot{H} \rp \Delta N + \frac{\pr^2 \tilde{\zeta}}{a^2} - \ddot{\tilde{\zeta}} - 6 H \dot{\tilde{\zeta}} - H \frac{ \pr^2 \tilde{\psi}}{2 a^2} \cgap \\
\Delta \rho + 3 \Delta p &= \frac{\pr^2 \Delta N}{a^2} + 3 H \Delta \dot{N} + 6 ( \dot{H} + H^2 ) \Delta N + 3 \ddot{\tilde{\zeta}} + 6 H \dot{\tilde{\zeta}} + \frac{\pr^2 \dot{\tilde{\psi} }}{2a^2}  \cgap \\
( \bar{\rho}+ \bar{p} ) \pr_i v &= -  H \pr_i \Delta N + \pr_i \tilde{\zeta} \cgap \\
4 \pr_i \pr_j \sigma &=  \pr_i \pr_j \lp  2 \Delta N + 2 \tilde{\zeta} - \frac{1}{a} \frac{d ( a \tilde{\psi} ) }{dt}  \rp \cgap 
\end{align}
where we introduced
\begin{equation}
\tilde{\zeta}= \zeta - \frac{1}{6} \partial^2 \gamma \cgap \quad \tilde{\psi}= a^2 \dot{\gamma}-2 \psi \cgap
\end{equation}
for the simplicity of the notation. If one plugs in \eqref{eq:ad sln} into these equations, they will be trivially satisfied. However if we introduce a small space dependency into $\epsilon$, the equations become
\begin{align}
\frac{\pr^2 g(x) }{a^2} H \epsilon &= 0 \cgap \\
\frac{\pr^2 g(x)}{a^2} H \dot{\epsilon} &= 0 \cgap \\
\lp \bar{\rho} + \bar{p} + \dot{H} \rp \epsilon \, \pr_i g(x)  &= 0 \cgap \\ 
\lp \dot{\epsilon} + H \epsilon \rp \pr_i \pr_j g(x) &=0 \sgd
\end{align}

We wrote the small spatial dependency in $\epsilon$ explicitly by taking $\epsilon \ra \epsilon(t) g(x)$. The 3rd equation is automatically satisfied by virtue of background equations. Now let us impose
\begin{equation}\label{eq:eps eqn}
\dot{\epsilon} + H \epsilon=0 \sgd
\end{equation}
then the last equation will be exactly satisfied. We can now take $g$ arbitrarily small so that it satisfies the remaining Einstein equations- but also vanishing at infinity.  Since there is no gauge freedom left in this case we conclude that we have found a solution that is not gauge equivalent to a zero solution.

In the next section we will see that the Manton approximation will be much in the same spirit, with a small difference: We will start with a pure gauge solution that has only space dependency, then introducing small time dependency and imposing the constraint equation we will get some physical solutions. We will see for this solution the theory will reduce to a boundary theory.\que{For further investigation: Does Weinberg adiabatic modes make the solution a boundary theory?}

%% file: mantonapproximation.tex
\section{Manton Approximation for a System with Natural Lagrangian}\label{sec:setup}

The approximation that given a system with a Lagrangian that can be decomposed into kinetic and potential parts whose potential has a continuous minima, dynamics will be such that the motion off the continuous minima is ignorable if the system has started with a small velocity around the minima is called the Manton approximation, see \cite{manton,WeinbergMono,Stuart2007}. 

Now we move on to proving the statement above. Assume we have field theory defined on some spacetime $\stm$, with fields $\phi^I$, with the action
\begin{equation}
\act= \int  \lp \frac{1}{2} g_{I J} \lp \phi^K(t,x) \rp \dot{\phi}^{I}(t,x)\dot{\phi}^{J}(t,x)- V \lp \phi^K(t,x) \rp \rp d^n x \cgap
\end{equation}
where $(t,x) \in \stm$ and $n$ is the dimension of $\stm$. Its equations of motion are
\begin{equation} \label{eq:geod on config}
g_{I K} \lp \ddot{\phi}^{I} +  \Gamma^{I}_{M J} \dot{\phi}^{M} \dot{\phi}^{J} \rp = - \pr_{K} V \sgd
\end{equation}

We will now adopt a geometric picture as follows: Let us assume that the set of all allowed field configurations $\phi^I$ that depend only on space coordinates $x$ form a smooth manifold $\consp$. Here ``allowed" means satisfying a given set of spatial boundary conditions. Dynamical fields that describe the time evolution of the system will formally be \hyperlink{flow}{flows} on $\consp$. 
\begin{figure}[!ht]
\centering
\includegraphics[width=\linewidth]{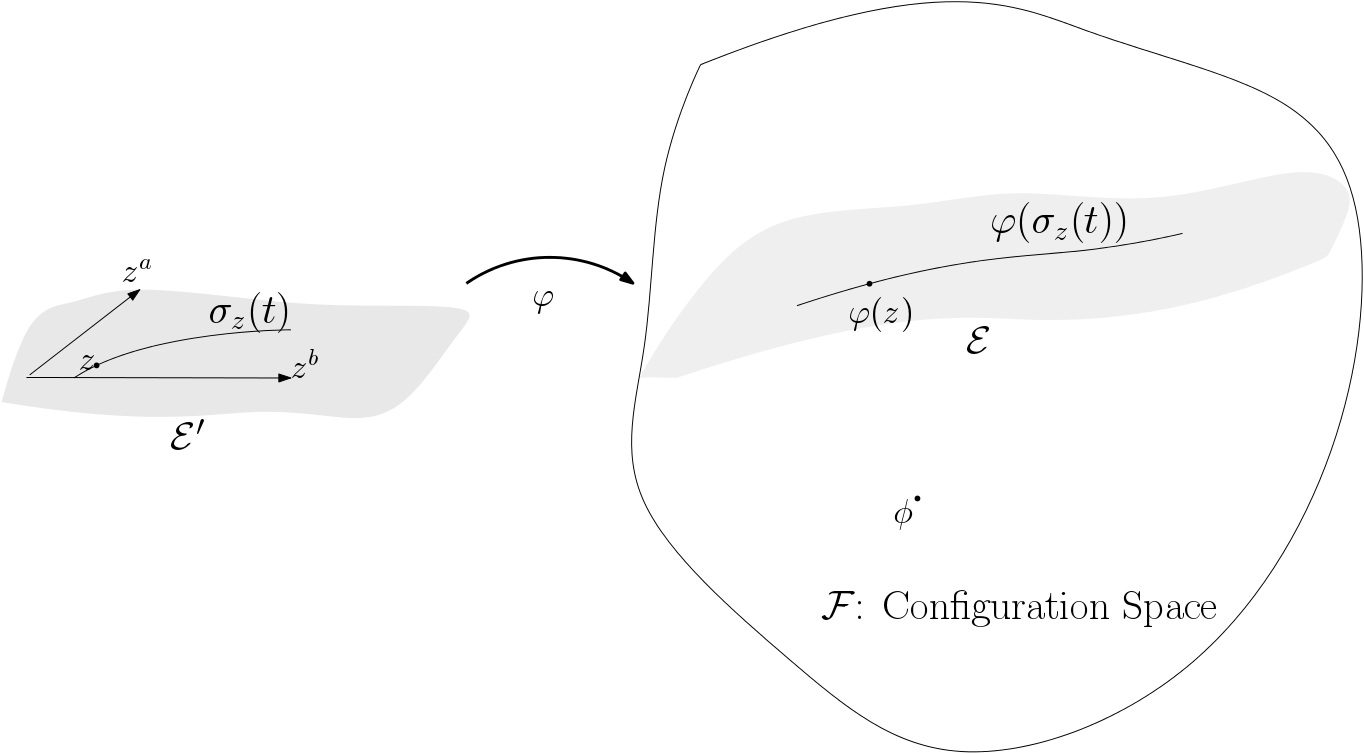}
\caption{Configuration Space and Minima Space}
\vskip\baselineskip
\label{fig:mantonapp}
\centering
\end{figure}

We also assume that the potential $V$ has set of minima $\minsp$ that forms a submanifold of $\consp$. Let this set be an embedding of a manifold $\minsp'$ into $\consp$, with map $\minf : \minsp' \ra \consp$ so that $\minf(\minsp')=\minsp$, then we have not only
\que{After this there is the question: what are coordinates on $\consp$-Witten paper might help}
\begin{equation} \label{eq:1stconst}
\left. \pr_I V \right|_{\minsp}=0 \cg
\end{equation}
but also since the above holds along $\minsp$,
\begin{equation} \label{eq:2ndconst}
\pr_z \pr_I V \lp \minf(z) \rp = 0 \quad \ra \quad \pr_I \pr_J V \lp \minf(z) \rp \cdot \frac{\pr \minf^J(z)}{\pr z^a} =0 \sgd
\end{equation}
Here $z^a$ are coordinates on $\minsp$. Since we have a metric $g_{IJ}$ defined on $\consp$, now we can split the tangent bundle of $\consp$, $T\consp$ into two orthogonal parts $T\consp=T\minsp+T\minsp^{\perp}$. Note that $ \partang= \frac{\pr \minf^J(z)}{\pr z^a} \frac{\pr}{\pr \phi^J} $ span $T\minsp$ and a tangent vector $\orthtang$ is element of $T\minsp^{\perp}$ if and only if $g_{IJ} \orthtang^I  \partang^J=0$. Also note \eqref{eq:2ndconst} means parallel tangent vectors are zero eigenvectors of $\hes_{IJ}=\left. \pr_I \pr_J V \right|_{\minsp}$, the Hessian matrix of the potential V on $\minsp$.

Our approximation will be that the solution to our theory is a adiabatic motion around the minima. For this we conjecture a solution: 
\begin{equation} \label{eq:field exp}
\phi^I(x,t)= \minf^I(\sigma(z,t)) + \epsilon \lp \partan1{I} + \orthtan1{I} \rp + \epsilon^2 \lp \partan2 {I} + \orthtan2{I} \rp + \ord(\epsilon^3)\sgd 
\end{equation}
where $\sigma$ is a flow on $\minsp$.
To account for the adiabaticity, we let each time derivative bring in a perturbation parameter $\epsilon$, i.e. 
\begin{equation}\label{eq:field der exp}
\dot{\phi}^I(x,t)= \epsilon \, \dot{\minf}^I(\sigma(z,t)) + \epsilon^2 \lp \dpartan1{I} + \dorthtan1{I} \rp + \epsilon^3 \lp \dpartan2{I} + \dorthtan2{I} \rp + \ord(\epsilon^4) \sgd
\end{equation}
Plugging these into the equations of motion \eqref{eq:geod on config} we get
\begin{align}\label{eq:geod exp}
g_{I K} ( \minf + \epsilon \, \delta \phi ) & \lp \epsilon^2 \, \ddot{\minf}^{I} + \epsilon^2 \, \Gamma^{I}_{M J} ( \minf + \epsilon \, \delta \phi ) \dot{\minf}^{M} \dot{\minf}^{J} \rp + \mathcal{O}(\epsilon^3)  \nonumber \\ &= - \left. \pr_{K} V \right|_{\minf} - \epsilon \, \hes_{IJ}  \lp \partan1{I} + \orthtan1{I} \rp + \mathcal{O}(\epsilon^2) \cgap
\end{align}
where in writing we omit the dependency of $\minf$ on $\sigma$ and collect all the perturbations into the term $\delta \phi$ wherever needed. First two term on the right hand side drops due to the conditions \eqref{eq:1stconst}, \eqref{eq:2ndconst}. Thus there will be no new equation in $\ord(\epsilon^0)$. 

At order $\ord(\epsilon)$ we have
\begin{equation}
\hes_{IJ} \cdot \orthtan1{I}= 0 \sgd
\end{equation}
Contracting this with $\partang^I$ gives an identity via \eqref{eq:2ndconst}, whereas contracting with $\orthtang^I$ will give
\begin{equation}
\hes_{IJ} \, \orthtan1{I} \, \orthtan1{J}=0 \sgd
\end{equation}
Since on $\minsp$ space V is minimal, its Hessian should be positive semidefinite there, but $\orthtan1{I}$ is orthogonal to the zero direction and thus the equation above implies
\begin{equation} \label{eq:1st orth pert}
\boxed{ \orthtan1{I}=0 \sgd }
\end{equation}

Moving on to the $\ord(\epsilon^2)$ we first take a look at the right hand side. For that we expand \eqref{eq:geod exp} to $\epsilon^2$ order. Taking into account \eqref{eq:2ndconst} and $\orthtan1{I}=0$ we get
\begin{equation} \label{eq:2ndorderRHS}
\ord(\epsilon^2) \, \mbox{of R.H.S of \eqref{eq:geod exp}}= \epsilon^2 \hes_{IK} \orthtan2{I} + \frac{\epsilon^2}{2} \left. \pr_J \pr_I \pr_K V \right|_{\minsp} \partan1{J} \partan1{I} \sgd
\end{equation}
To simplify this we take the z derivative of \eqref{eq:2ndconst} to get
\begin{equation}
\left. \pr_K \pr_I \pr_J V \right|_{\minsp} \frac{\pr \minf^J(z)}{\pr z^a} \frac{\pr \minf^K(z)}{\pr z^b} + \left. \pr_I \pr_J V \right|_{\minsp} \frac{\pr^2 \minf^J(z)}{\pr z^a \pr z^b}= 0 \sgd
\end{equation}
contracting this with $\partang^K= \frac{\pr \minf^K(z)}{\pr z^a}$ gives
\begin{equation}
\left. \pr_K \pr_I \pr_J V \right|_{\minsp} \partang^K \partang^I \partang^J =0 \sgd
\end{equation}

So contracting \eqref{eq:2ndorderRHS} with a parallel tangent vector $\partan1{}$ one gets identically zero. Thus taking the $\ord(\epsilon^2)$ of the left hand side of \eqref{eq:geod exp} and contracting it with $\partan1{K}$ one gets
\begin{equation}
g_{IK}(\minf) \lp  \ddot{\minf}^{I} + \Gamma^{I}_{M J} ( \minf ) \dot{\minf}^{M} \dot{\minf}^{J} \rp \partan1{K} =0 \sgd
\end{equation}
Now we note
\begin{align}
\dot{\minf}^{I}\lp \sigma \rp &= \pr_a \minf^I \dot{\sigma}^a \cg \\
\ddot{\minf}^{I} \lp \sigma \rp &= \pr_a \pr_b \minf^I \dot{\sigma}^a \dot{\sigma}^b + \pr_a \minf^I \ddot{\sigma}^a \sgd
\end{align} 
Plugging these in we have
\begin{equation}
g_{IK} \lp \pr_a \pr_b \minf^I \dot{\sigma}^a \dot{\sigma}^b + \pr_a \minf^I \ddot{\sigma}^a + \Gamma^{I}_{M J} ( \minf ) \pr_a \minf^M \dot{\sigma}^a \pr_b \minf^J \dot{\sigma}^b \rp \pr_c \minf^K=0 \cg
\end{equation}
which becomes
\begin{equation} \label{eq:geod on vacua}
\boxed{ \tilde{g}_{cd} \lp \ddot{\sigma}^d + \tilde{\Gamma}^d_{ab} \dot{\sigma}^a \dot{\sigma}^b \rp =0} \cg
\end{equation}
where $\tilde{g},\tilde{\Gamma}$ are induced metric on $\minsp$ and its corresponding Christoffel symbol respectively. Thus we see the adiabatic approximation \eqref{eq:field exp} and \eqref{eq:field der exp} gives a solution to the system such that at the first order in perturbation, $\orthtan1{}=0$ i.e. there exists no motion in directions orthogonal to the manifold $\minsp$ of the minima, and the motion along $\minsp$ is described by the geodesic equation on it with respect to the metric $\tilde{g}$ induced by $g$.

Now we illustrate this approximation via a simple example.
\begin{example}
We consider a point particle in three dimensional flat space, with the following Lagrangian:
\begin{equation}
\lag= \frac{1}{2} m \dot{X}^I(t) \dot{X}^I(t) - \frac{1}{4} \lambda \lp X^I X^I - R^2 \rp^2 \sgd
\end{equation}

In the language of our above prescription fields that depend on spacetime now becomes positions $X^I$ that depend on time and the metric is $g_{IJ}=m \delta_{IJ}$. The potential is more simply expressed however, if one goes to the spherical coordinates such that
\begin{equation}
L= \frac{m}{2} \dot{\rho}^2 + \frac{m}{2} \rho^2 \lp \dot{\theta}^2 + sin^2(\theta) \dot{\phi}^2 \rp - \frac{1}{4} \lambda \left(\rho^2 - R^2 \right)^2 \sgc
\end{equation}
where now $g_{IJ}=m \, \mbox{diag}(1, \rho^2, \rho^2 sin^2(\theta))$ and $X^I=(\rho,\theta,\phi)$. The configuration space $\consp$ is the three dimensional flat space itself and $\minsp=S^2_R$, round sphere with radius $R$. Note that $z^a=\lbr \theta, \phi \rbr$ parametrize $\minsp$.

We now go onto write down the exact equations of motion, and then will apply the approximation \eqref{eq:field exp}, \eqref{eq:field der exp}; so that we will see a confirmation of our conclusions \eqref{eq:1st orth pert}, \eqref{eq:geod on vacua}. Equations of motions read 
\begin{align}
& \rho\left(\ddot z^a+ \tilde{\Gamma}_{b c }^{a} \dot{z}^{b} \dot{z}^{c} \right)+ \dot{\rho}\dot{ z}^{a} = 0\,,\\
& \ddot \rho+\frac{4}{m}(\rho^2-R^2)\rho - \rho \tilde{g}_{a b}\dot{z}^{a} \dot{z}^{b}=0 \, ,
\end{align}
where $\tilde{g},\tilde{\Gamma}$ are induced metric and corresponding Levi-Civita symbol on $\minsp$. Now we write down the approximation of \eqref{eq:field exp}:
\begin{align}
\rho &= R + \epsilon \rho^{(1)}(t) + \epsilon^2 \rho^{(2)}(t)+ \cdots \cgap \\
z^{a} &= \bar{z}^{a}(t) + \epsilon z^{(1) a}(t) + \epsilon^2 z^{(2) a}(t) + \cdots \sgd
\end{align}
Note that approximation of being around the vacua is captured in $\rho^{(0)}=R$. By also employing  \eqref{eq:field der exp} together with the expansion above in the exact equations of motion we get for the first order in $\epsilon$:
\begin{equation}
\rho^{(1)}=0 \sgc
\end{equation}
as we have conjectured. For the second order:
\begin{align}
& \frac{8}{m} R \rho^{(2)}- R g_{a b} \dot{\bar{z}}^{a} \dot{\bar{z}}^{b} = 0 \cgap \\
& \ddot{\bar{z}}^{a} + \tilde{\Gamma}_{b c}^{a} \dot{\bar{z}}^{b} \dot{\bar{z}}^{c} =0 \sgd
\end{align}
Here first equation is the orthogonal part of the second order equation which we had omitted in our proof above, and the second equation is the geodesic equation on $\minsp$ again as was our conclusion above.
\end{example}

\section{Manton Approximation for Yang Mills}\label{sec:ymills}
The purpose of this thesis will be to apply the Manton approximation described above to the theory of general relativity. In the accompanying paper this was performed for Yang-Mills theory \cite{ymmanton}. Here we summarize this paper, to illustrate certain points we will be making in a different setting.

Yang-Mills theory on the flat spacetime can be considered to be defined by Lie algebra valued one forms $A$ on the spacetime with Minkowski metric. The action for the theory is
\begin{equation}
S_{YM}=- \frac{1}{2} \int d^4x \tr \lp F_{\mu \nu} F^{\mu \nu} \rp \sgc
\end{equation}
where
\begin{equation}
F=dA
\end{equation}
is a two form, and $\tr$ is over the generators of the related Lie algebra. To bring this Lagrangian into the natural form one makes the gauge choice
\begin{equation}
A_0=0 \sgd
\end{equation}

As $A_0$ is in the Lagrange multiplier form in the action, with this choice the constraint coming from this Lagrange multiplier should be externally imposed in addition to the reduced action, i.e. in this gauge the theory is described by the action
\begin{equation}
S= K - V = \frac{1}{2} \int d^4x \tr \lp \dot{A}_i \dot{A}^i \rp - \frac{1}{2} \int d^4x \tr \lp F_{ij} F^{ij} \rp
\end{equation}
together with the constraint
\begin{equation}
D_i \dot{A}^i= \pr_i \dot{A}^i + \lb A_i, \dot{A}^i \rb =0 \sgc
\end{equation}
where $D_i$ is the Yang-Mills covariant derivative, and the bracket is the Lie bracket of the associated Lie algebra. Thus the configuration space $\consp$ will be described by fields $A_i(x)$ with spatial dependency. On the right hand side the matrix representation is implied and multiplications are matrix multiplications. On a manifold with trivial homology the vacuum solutions, the solutions with absolute minimum energy, are pure gauge solutions. Let us denote the gauge transformation of the fields $A$ as $g \cdot A$ where
\begin{equation}
g \cdot A= g A g^{-1} + gdg^{-1} \sgc
\end{equation}
where $g$ is the element of the set of maps from Lie group to the spacetime manifold, i.e. Lie group elements with spacetime dependency. Then the space of minima fields $\minsp$ is fully described by
\begin{equation}
\mathcal{A}= g_z \cdot A_o \sgc
\end{equation}
where $A_o$ is a given reference point on $\minsp$, which can be chosen as $A_o=0$ and $z$ parametrizes the set of all transformations that can be made. Note that since we have restricted our selves to the gauge $A_0=0$, the remaining gauge transformations will be group elements with only spatial dependency, i.e.
\begin{equation}\label{eq:vac}
\mathcal{A}= g_z(x) \cdot A_o \sgd
\end{equation}

Now to study the adiabatic solutions we consider flows on $\minsp$. This will amount to introducing a time dependency into the solutions above:
\begin{equation}\label{eq:ad}
\mathcal{A}= g_{z(t)}(x) \cdot A_o \sgd
\end{equation}
Note that introducing this time dependency will take us out of the physical equivalence class of the pure gauge solutions, i.e. solutions above are not gauge equivalent to $A_o=0$ solution. We will stress and explain this point more in the chapters to come. Now we attempt to solve the Yang-Mills equations for this ansatz. First we consider the constraint equation. For this one needs to calculate
\begin{equation}
\dot{\mathcal{A}}_i= -D_i \gamma_z \gwg \gamma_z=\dot{g}_z g^{-1}
\end{equation}
where $\gamma_z$ is an element of the Lie algebra, and $D_i$ is the covariant derivative with the connection $\mathcal{A}$. Using this the constraint equation becomes
\begin{equation}
D_i D^i \gamma_z=0 \sgd
\end{equation}
This equation might be complicated due to the fact that $\mathcal{A}$'s are also present in the covariant derivative $D_i$. However since at each instant $t$, $A_i$ is gauge equivalent to $A_o$ at an instant this equation can be written as
\begin{equation}\label{eq:constr}
D_o^2 \sigma_z=0 \gwg \sigma_z=g^{-1} dg \sgd
\end{equation}
Using the form of the solution we propose, after some manipulations and constraint equation in the action, and also the fact that the potential is zero on vacuum, action becomes
\begin{equation}\label{eq:ymaction}
S= \int dt \int_{\pr M} \tr \lp \sigma_z D^{\perp} \sigma_z \rp
\end{equation} 
where $D^{\perp}\sigma_z = \left. n^i D_i \sigma_z \right|_{\pr M}$.

\paragraph{Homogeneous Space Structure:} Now we move onto exploring the structure of the space of vacua. Note that solutions to the constraint equation \eqref{eq:constr} are uniquely determined by their values on the boundary, for $\sigma$ that vanish on the boundary they are zero. Thus the solutions are identified by the Lie algebra quotient
\begin{equation}
\mathfrak{s}= \frac{\mathfrak{g}}{\mathfrak{g}_0} \sgc
\end{equation}
where $\mathfrak{g}$ is the set of Lie algebra valued functions with spatial dependency and $\mathfrak{g}_0$ is the subset of it that vanishes on the spatial boundary $\pr M$. One can show that $\mathfrak{g}_0$ is an ideal in $\mathfrak{g}$, thus the quotient will be a Lie algebra itself. The corresponding Lie group is
\begin{equation}
\mathcal{S}= \frac{\mathcal{G}}{\mathcal{G}_0} \sgc
\end{equation}
where $\mathcal{G}_0$ will be the set of transformations that are identity on the boundary and it is a normal subgroup of $\mathcal{G}$. However $\mathcal{S}$ actually does not identify the space of vacua accurately, there will be members of its algebra $\mathfrak{s}$ that are in the kernel of $D^{\perp}$, making the action vanish. In the paper \cite{ymmanton} kernel of $D^{\perp}$ was shown to be equal to the kernel of the operator $D_o$, thus they are indeed the isotropy subgroup of the set of transformations acting on the gauge fields $A$. The corresponding group elements $\mathcal{K}$ are subgroup of $\mathcal{S}$, but not a normal subgroup. As a result the final space of vacua
\begin{equation}
\mathcal{V}= \mathcal{K} \backslash \lp \mathcal{G} / \mathcal{G}_0 \rp
\end{equation}
will not be a Lie group, but rather a ``homogeneous space". Moreover since the action \eqref{eq:ymaction} imposes a metric on this space via
\begin{equation}
S= \int dt \frac{1}{2} \metonvacua( \delta_{\gamma_z} A, \delta_{\gamma_z} A) \sgc
\end{equation}
where $\metonvacua$ is the metric on the space of vacua given by
\begin{equation}
\metonvacua( \delta_{\gamma^1_z} A , \delta_{\gamma^2_z} A) = \int_{\pr M} \tr \lp \sigma^1_z \ortder \sigma^2_z \rp \sgc
\end{equation}
that is ``invariant under left translations of the Lie group $\mathcal{S}$", the space of vacua will be a ``Riemannian homogeneous space". The concepts inside the quotation marks will be explained in detail in the following chapters of the thesis.

%% file: mathgr.tex
\chapter{Mathematics for General Relativity}
In this chapter we explore the mathematical apparatus necessary to implement the program we have prescribed, to general relativity. As we have seen previously, Manton approximation heavily relies on a natural form of the Lagrangian where the kinetic and potential parts are separated, and thus the time coordinate is split from the others. The Einstein-Hilbert action is given by the curvature of the spacetime and hence this split is not visible there. In the context of general relativity splitting of time from space will mean foliating the spacetime into spacelike surfaces, and expressing all objects on the spacetime, i.e. tensors, in terms of their pullbacks to the spacelike surfaces and the remaining components. 

Splitting of spacetime into space and time was famously performed by Arnowitt-Deser-Misner \cite{adm}, and became a standard treatment in many of general relativity books, see e.g. \cite{mtw,wald} and also \cite{gour}, a nice review for the purposes of numerical relativity. However many of the treatments were performed in a coordinate basis, which automatically gives a foliation and satisfies extra conditions other than those required by the foliation. Instead of this approach in Section \ref{sec:Fol}, we will start from a general non-coordinate basis to get the most general expressions for fundamental quantities such as Ricci tensor and Ricci scalar in therms of their split parts. The theory of foliations was discussed in Appendix \ref{sec:appfol} where the requirements for a differentiable manifold to be foliated into well defined leaves are discussed explicitly.

Moreover we will perform this operation of foliation for a general (pseudo)-Riemannian manifold, taking care of the Lorentzian and Euclidean signatures simultaneously. The reason for this will be as follows: after foliating the spacetime manifold, our construction will become a theory on a spatial manifold. We will take this manifold to be a manifold with boundary. For this reduced theory we will furthermore split the quantities in terms of quantities orthogonal and tangent to boundary, and more specifically we will be imposing certain conditions on our quantities on the boundary. For ease of operation we will take this spatial manifold to be also foliated such that the boundary is an \hyperlink{integman}{integral manifold} of the spatial manifold. All of the tools developed in Section \ref{sec:Fol} will be also used for this second foliation, however the specifics will be performed in the next chapter, not here. 

On this spatial manifold we will need more machinery. First in Section \ref{sec:manbnd}, we will define what a manifold with boundary is, and some useful relations by following mostly \cite{Gunter}. As our theory reduces to a theory on the spatial manifold, the constraints of general relativity will become differential equations on a spatial manifold; moreover the momentum constraint becomes an exterior calculus equation for our solutions. To attack these differential equations we will use the Hodge decomposition. However the Hodge decomposition on a manifold with boundary is different than its usual treatment in manifolds without a boundary, e.g. while there is no harmonic form on a manifold without boundary with trivial homology, there will be an infinite number of them on a manifold with boundary. Hodge decomposition on a manifold with boundary is called Hodge-Morrey-Friedrichs (\hmf) decomposition, and we will treat it in Section \ref{sec:hmf}.

As we have discussed in Section \ref{sec:ymills}, our analysis will give us information about the space of vacua which will turn out to be a Riemannian homogeneous space. In Section \ref{sec:homsp} we review homogenous spaces with geometry where we follow \cite{arvan}. As a natural place to start, first the subject of Lie groups with a geometry is reviewed, and what peculiarities this geometry has if the metric has some invariance properties under the group actions. Then this study is extended to homogeneous spaces in the following subsection, where we define what a Riemannian homogeneous space is and conclude by a nice proposition on equivalence of various constructions on these manifolds we will be using later on.

\section{Foliating a Riemannian Manifold}\label{sec:Fol}
Let $(\amb,\metamb)$ be a (pseudo) Riemannian manifold of $\dimamb$ dimensions with the corresponding Riemannian connection $\conamb$. For generality we will consider both Euclidean and Lorentzian signature metrics. We would like to foliate this manifold into proper hypersurfaces i.e. we assume there exists an \hyperlink{integrabledistribution}{integrable distribution}. By \hyperlink{frob}{Frobenius Theorem} this means it is enough to assume the existence of an involutive distribution. Let the set of vector fields $\{ \ei \}$ , $i=1,...,\dimsubm$, be a basis for the distribution we are considering, where $\dimsubm=\dimamb-1$ is the dimension of the distribution and the integral manifolds. As explained in the appendix, Section \ref{sec:appfol}, a distribution can equivalently be defined by defining forms, and for our case since we are considering a distribution of codimension-1, we need a single defining one form which we call $\wz$.

Assumption of involutivity is equivalent to saying
\begin{equation}
\wz \lp \lb \ei,\ej \rb \rp=0 
\end{equation}
or equivalently
\begin{equation}
\wz \wedge d \wz=0 \sgd
\end{equation}

Let us complete the basis $\lbr \ei \rbr$ with the dual vector field to $\wz$, called $\ez$. Note that here we consider a basis that is not necessarily coordinate, thus basis vectors fields do not commute. We also make the unconventional choice for dual vectors that
\begin{equation}
\wz(\ez)= \pm 1 \sgc
\end{equation}
where we use the convention that upper sign corresponds to metrics with Euclidean signature and lower sign corresponds to metrics with Lorentzian signature, a convention to be used throughout the document. The reason for the unconventional minus choice for the second case comes from the desire to make two choices simultaneously: we want $\ez$ to be equivalent to the vector obtained by raising the index of $\wz$ and we want $\ez$ to have a negative norm for the Lorentzian case: 
\begin{equation}\label{eq:nrm1}
\metamb(\ez,\ez)= \pm 1 \sgd
\end{equation}
Note that this choice also produces orthogonality between the normal and tangential basis vector fields, i.e.
\begin{equation}\label{eq:nrm2}
\metamb(\ez,\ei)=0 \sgc
\end{equation}
since $\wz(\ei)=0$, and involutivity condition becomes
\begin{equation}
\metamb(\ez, \lb \ei, \ej \rb)=0 \sgd
\end{equation}

Whenever necessary we will also have duals to $ \ei $, called $\wi$, such that
\begin{equation}
\wi(\ej)=\delta^i_j \sgd
\end{equation}
Note, unlike $\wz$ we do not choose $\wi$ to be the raised $\ei$, so there will be no normalization condition on $\lbr \ei \rbr$ and we will call
\begin{equation}
\metamb( \ei, \ej) = \metsubm_{ij} \sgd
\end{equation}

To summarize our conditions, we say we consider a manifold with metric $\metamb$, with a basis $\{ \ez, \ei \}$ such that
\begin{equation}\label{eq:cond on metric}
\metamb(\ez,\ez)= \pm 1 \cg \metamb(\ez,\ei)=0 \cg \metamb(\ez, \lb \ei, \ej \rb)=0 \sgd
\end{equation}

By these choices we have fixed our foliation, now we can move onto splitting the objects on our manifold, i.e. tensor fields in accordance with the foliation. Consider e.g. a vector field, we can decompose it into
\begin{equation}
V=V^{\mu} e_{\mu}=V^0 \ez+V^i \ei \sgd
\end{equation}
Thus a vector field can be decomposed into a scalar $V^0$ and a hyperspace vector. Similarly a type (2,0) tensor can be decomposed into
\begin{equation}
T= T^{00} \ez \otimes \ez + T^{0i} \ez \otimes \ei + T^{i0} \ei \otimes \ez + T^{ij} \ei \otimes \ej \sgc
\end{equation}
i.e. into a scalar,vector and a tensor on hypersurfaces. Applying this to metric and inverse metric,they can be written as
\begin{align}
\metamb &= \pm \wz \otimes \wz + \metsubm_{ij} \wi \otimes \wj \sgc \\
\metamb^{-1} &= \pm \ez \otimes \ez + \metsubm^{ij} \ei \otimes \ej \sgd
\end{align}
where $\metsubm^{ij}$ is the inverse of the matrix $\metsubm_{ij}$. 

After tensors themselves, we also use their covariant derivatives, and also should decompose covariant derivatives of tensors. The most useful way to do this is writing out the covariant derivatives of the basis vectors with respect to each other. Note that since our conditions \eqref{eq:cond on metric} hold through out the manifold, one can show that the following also do by using metric compatibility and equation \eqref{eq:Cov distr over tensor}:
\begin{align}
& \metamb(\conambz \ez,\ez) =0 \cg  & \metamb(\conambz \ez,\ei) + \metamb(\ez,\conambz \ei) =0 \cg\\
& \metamb(\conambi \ez,\ez) =0 \cg  & \metamb(\conambj \ez,\ei) + \metamb(\ez,\conambj \ei) =0 \sgd
\end{align}
The covariant derivative of the involutivity condition turns out not to be giving an independent constraint. Using these constraints we identify the covariant derivatives of the basis vectors by the free quantities of $\acc^i,{\twoconn_i}^j,\ext_i^j$ and $\consubm$- the Riemannian connection of the induced metric on each integral manifold- as follows:
\begin{align}
\conambz \ez & = \acc^i \ei \sgc\\
\conambi \ez &= \ext_i^j \ej \sgc \label{eq:mywein}\\ 
\conambz \ei &= \mp \acc_i \ez + {\twoconn_i}^j \ej \sgc \\
\conambj \ei &= \consubmj \ei \mp \ext_{ji} \ez \sgc \label{eq:mygauss}
\end{align}
where $\ext_{ji}= \ext_j^k \metsubm_{ki},\acc_i= \acc^j \metsubm_{ji}$. Note that \eqref{eq:mywein} is the Weingarten equation \eqref{eq:wein}, and \eqref{eq:mygauss} is the Gauss equation \eqref{eq:gauss}. $\ext_{ij}$ is known as the extrinsic curvature and its related to \hyperlink{second fundamental form}{second fundamental form} via
\begin{equation}
\ext_{ij}= - \wz( \Pi(\ei,\ej) ) \sgd
\end{equation}
The condition of involutivity of the distribution makes the extrinsic curvature symmetric:
\begin{equation}
\metamb(\ez, \lb \ei, \ej \rb)=0 \ra \ext_{ij}=\ext_{ji} \sgd
\end{equation}
Except the symmetry of the extrinsic curvature, no other condition is imposed on the quantities we have defined. We also remind here that, even if a connection is torsionless, the corresponding Christoffel symbols are symmetric only in a coordinate basis, but not necessarily so in a non-coordinate basis, i.e.
\begin{equation}
\consubmi \ej - \consubmj \ei = \lb \ei,\ej \rb = 2 \Gamma_{[ij]}^k e_k \sgd
\end{equation}

It is useful also to write down the following commutator:
\begin{equation}
\lb \ez,\ei \rb= \conambz \ei - \conambi \ez= \mp \acc_i \ez + \lp {\twoconn_i}^j-{\ext_i}^j \rp \ej \sgd
\end{equation}
The quantity $\acc^i= \wi(\conambz \ez)$ is called acceleration. Note that for the case where $\ez$ is $\conamb$ geodesic in the ambient manifold the acceleration is zero. \expl{What is $\twoconn$} 

Note that in the way we introduced the quantities $\acc^i,\twoconn^i_j, \ext^i_j, \Gamma^k_{ij}$ are simply numbers. However one can define hypersurface tensor fields
\begin{equation}
\acc=\acc^i e_i \cg \ext= \ext_{ij} \wi \otimes \wj \cg \twoconn= {\twoconn_i}^j \wi  \otimes \ej \cg \Gamma= \Gamma^k_{ij} \ek \otimes \wi \otimes \wj \sgd
\end{equation}
One might be surprised that we call $\Gamma$ a tensor field, but viewed this way it will be. The point is that if one has two sets of vector fields $\{ \ei \}$ and $\{ \tilde{e}_i \}$ and makes the definitions
\begin{align}
\tilde{\Gamma}&= \tilde{\Gamma}^k_{ij} \tilde{e}_k \otimes \tilde{\omega}^i \otimes \tilde{\omega}^j \gwg \tilde{\Gamma}^k_{ij}=\tilde{\omega}^k\lp \conambi \ej \rp \sgc \\
\Gamma &= {\Gamma}^k_{ij} {e}_k \otimes {\omega}^i \otimes {\omega}^j \gwg {\Gamma}^k_{ij}={\omega}^k\lp \conambi \ej \rp \sgc
\end{align}
then these tensor fields will not be equal. Similar argument is true for other quantities we defined. As long as we fix the foliation this will not be an important issue however. For the rest of the document we will call these tensors sometimes as \hypertarget{slicing tensors}{slicing tensors}.

Now we move onto splitting the Ricci curvature tensor of this manifold into normal and tangential parts. Remember that the Ricci curvature is expressed as
\begin{equation}
\ricamb= \rtil_{\mu \nu} \omega^{\mu} \otimes \omega^{\nu} = {\rtil_{\sigma \mu \nu}}^{\sigma}= \metamb \lp \rieamb(e_{\alpha},e_{\mu})e_{\nu},e_{\beta} \rp {\metamb}^{\alpha \beta} \sgd
\end{equation}
First calculate
\begin{align}
\rtil_{00} &=\metamb \lp \conambi \conambz \ez - \conambz \conambi \ez- \conamb_{\lb \ei, \ez \rb} \ez, \ej  \rp \metsubm^{ij} \\
&= \metamb \lp \conambi \conambz \ez - \conambz \conambi \ez, \ej  \rp \metsubm^{ij} \mp \acc_i \acc^i - \lp \ext_i^j-\twoconn_i^j \rp \ext_j^i\sgc
\end{align}
so that 
\begin{equation}
\boxed{\rtil_{00}= - \ez(\ext) - \ext^{ij} \ext_{ij} + (\consubmi \acc)^i \mp \acc^i \acc_i  \sgc }
\end{equation}
where we use the notation $ (\consubmi \acc)^i= \wi \lp \consubmi \acc \rp $, to distinguish it from the covariant derivative acting on a labeled scalar i.e. $ \consubmi \acc^i = e_i (\acc^i)$. Moving on to $\rtil_{0i}$ we have
\begin{equation}
\rtil_{0i}={\rtil_{j0i}}^j= \metsubm^{jk} \rtil_{ikj0} = \pm \metsubm^{jk} \metamb \lp \ez,\rieamb(\ei,\ek)\ej \rp
\end{equation}
where we have used the symmetries of the Riemannian tensor. Then
\begin{align}
\rtil_{0i} &= \pm \metamb \lp \conambi \conambk \ej - \conambk \conambi \ej- \conamb_{\lb \ei, \ek \rb} \ej, \ez  \rp \metsubm^{jk}  \\
&=\pm \metsubm^{jk} \lp \metamb \lp \lp \mp \Gamma^m_{kj} \ext_{im} \ez \mp \ei\lp \ext_{mj} \rp \ez \rp - \lp i \leftrightarrow k \rp \rp - \metamb \lp \ez, \mp 2 \Gamma^m_{[ik]} \ext_{mj} \ez \rp \rp 
\end{align}
so that one gets
\begin{equation}
\boxed{ \rtil_{0i}= \pm \lp \consubm^j \ext_{ji} - \consubm_i K \rp \sgd }
\end{equation}
\noindent
Similarly,
\begin{align}
\rtil_{ij} &= \pm \metamb\lp \rieamb(\ez,\ei)\ej,\ez \rp + \metamb \lp \rieamb(\ek,\ei) \ej, \emm \rp \metsubm^{km} \equiv \rtil_{ij}^{(1)} + \rtil_{ij}^{(2)}  \sgd
\end{align}
Focusing on the first part,
\begin{align}
\rtil_{ij}^{(1)}&= \pm \metamb \lp \conambz \conambi \ej - \conambi \conambz \ej- \conamb_{\lb \ez, \ei \rb} \ej, \ez \rp \\
&= \mp \Gamma^k_{ij} \acc_k \pm \twoconn_j^k \ext_{ik} \mp \conambz(\ext_{ij}) \pm \conambi (\acc_j) \mp \metamb \lp \conamb_{\lb \ez,\ei \rb} \ej, \ez \rp \\
&=  \mp \Gamma^k_{ij} \acc_k \pm \twoconn_j^k \ext_{ik} \mp \conambz(\ext_{ij}) \pm \conambi (\acc_j) - \acc_i \acc_j \pm \lp \twoconn_i^k - \ext_i^k \rp \ext_{kj} \\
&=  \pm (\consubm_i \acc)_j - \acc_i \acc_j \mp \lp \ez(\ext_{ij}) + \ext_i^k \ext_{kj} \rp \pm 2 f_i^k \ext_{kj} \sgd
\end{align}
For the second part we simply use the Gauss formula:
\begin{align}
\rtil_{ij}^{(2)} &= R_{ij} -\metamb \lp \ez \ext( \ek, \emm ) ,  \ez  \ext(\ei,\ej) \rp \metsubm^{km} + \metamb \lp \ext( \ek, \ej ) \ez , \ext(\ei,\emm) \ez \rp \metsubm^{km}  \\
&= R_{ij} \mp \ext \ext_{ij} \pm \ext_j^m \ext_{ij} \sgd
\end{align}
Combining the two we get
\begin{equation}
\boxed{ \rtil_{ij}= R_{ij} \mp \ext \ext_{ij} \pm 2 {\twoconn_i}^k \ext_{kj} \mp \ez( \ext_{ij} ) \pm (\consubm_i \acc)_j - \acc_i \acc_j \sgd}
\end{equation}
The Ricci Scalar of the ambient manifold is
\begin{equation}
\rtil = \pm \rtil_{00} + \rtil_{ij} \metsubm^{ij} \sgc
\end{equation}
so that
\begin{equation}\label{eq:ricsc}
\boxed{ \rtil = R \mp \ext^2 \mp \ext^{ij} \ext_{ij} \mp 2 \ez(\ext) \pm 2 \consubm_i \acc^i - 2 \acc_i \acc^i \sgd }
\end{equation}

\section{Manifold with a Boundary}\label{sec:manbnd}
As we will see in the following chapters our adiabatic solutions will be described by the data on the boundary of the spatial manifold. Because of this we now develop the language necessary to study these manifolds, using some elements from the book of G.Schwarz \cite{Gunter}, thus a good deal of the following will be some varied versions of definitions and theorems there. Smooth manifolds are defined as spaces that are locally diffeomorphic to the Euclidean space of the same dimension. This definition will fail for a manifold with boundary, so in place of the full Euclidean space one puts the requirement of being locally diffeomorphic to the closed upper half Euclidean space $\Rnp$, defined as
\begin{equation}
\Rnp= \lbr (x^1, \cdots, x^n) \in \Rnn \, | \, x^n \ge 0 \rbr \sgd
\end{equation}
So the formal definition is as follows:
\begin{defn}\cite{LeeISM}
An n-dimensional smooth manifold with boundary is a second-countable
Hausdorff space in which every point has a neighborhood diffeomorphic either to an open subset of $\Rnn$ or to a relatively open subset of $\Rnp$.
\end{defn}

For our definition it is crucial to indicate that, differentiability on $\Rnp$ around the boundary is \uline{defined} by extensions of functions to be differentiated on the boundary to $\Rnn$ such that the extension agrees with the function on the boundary. Because of this definition of differentiability, we see that the tangent space of a manifold with boundary is the full n-dimensional vector space isomorphic to tangent spaces at the interior. Note that the tangent space of $\sm$ at the boundary, $\left. TM \right|_{\bnd} $ is this space we mentioned and includes vectors that are not ``tangent" to $\bnd$, in contrast to $T \bnd$ which only includes vectors tangent to $\bnd$. We refer to members of $\left. TM \right|_{\bnd} $ as \hyperlink{vector field along}{vector fields along} $\bnd$, while members of $ T\bnd $ are vector fields on $\bnd$. 

Now let $n$ be the unique vector field on $\sm$ that is on the boundary orthogonal to the boundary, i.e.
\begin{equation}
\left. \metsubm(n,n) \right|_{\bnd} = 1 \sgd
\end{equation}
In fact we will let this to be true off $\bnd$ as well, actually we will make the assumption that $\bnd$ is a part of a Frobenious foliation, so that we use the formulae of the preceding section. In a neighbourhood of the boundary we will decompose any tensor into components in the direction of $n$ and ones that are orthogonal to that. For the case of a vector field on $\sm$ this is realized as:
\begin{equation}
\xi(x)= \xi^{\perp}(x) n(x) + \xi^{\parallel}(x) \sgd
\end{equation}
We will use a special notation for restriction of vector fields to the boundary. We define
\begin{equation}
\ta \xi \equiv \left. \xi^{\parallel} \right|_{\bnd} \cg \no \xi \equiv \left. \xi^{\perp} n \right|_{\bnd} \sgd
\end{equation}
Generalization to covector fields and tensors is done similarly. Now we show two important results following from the definition:
\begin{prop}
Let $\omega$ be a p-form in $\sm$, then
\begin{equation}
*(\no \omega)= \ta (* \omega) \cg *(\ta \omega)= \no (* \omega)
\end{equation}
where $*$ is Hodge dual in $\sm$ with respect to $\metsubm$ and equations are understood to be holding on $\bnd$.
\end{prop}
\noindent
We refer to \cite{Gunter} for the proof, and content ourselves with an example. 
\begin{example}
Consider the three dimensional ball with its boundary $\bar{B}_R^3$ with the Euclidean metric. Using the spherical coordinates let
\begin{equation}
v= v_r dr + v_{\theta} d\theta + v_{\phi} d\phi \sgd
\end{equation}
Note that the vector field orthogonal to the boundary is
\begin{equation}
n=\frac{\pr}{\pr r} \sgd
\end{equation}
So one has
\begin{align}
\ta v = v_{\theta}(R,\theta,\phi) d \theta + v_{\phi} (R,\theta,\phi) d \phi \cg
\no v = v_r(R,\theta,\phi) dr \sgd
\end{align}
Also note
\begin{equation}
* v= \frac{v_{\phi}}{\sin \theta} dr \wedge d\theta - \sin \theta \, v_{\theta} \, dr \wedge d\phi + r^2 \sin \theta \, v_r \, d\theta \wedge d\phi \sgd
\end{equation}
Using these equations one can see
\begin{equation}
* (\ta v) = \frac{v_{\phi}}{\sin \theta} dr \wedge d\theta - \sin \theta \, v_{\theta} \, dr \wedge d\phi = \no (* v) \sgd
\end{equation}
Note the difference between this and \uline{boundary hodge dual} of the pullback of a one-form on $\bnd$
\begin{equation}
*_{\bnd} \lp j^* v \rp = \frac{v_{\theta}}{\sin \theta} d\theta - \sin \theta \, v_{\theta} \, d\phi
\end{equation}
where $j: \bnd \ra \sm$ is the inclusion map.
\end{example}

A second significant result is the commutation of $\ta$ with exterior derivative $d$ and $\no$ with $\dad$: 
\begin{prop}\label{prop:der_tan_comm}
Let $\omega$ be a p-form and $j: \bnd \ra \sm$ inclusion map. Then
\begin{equation}
j^*(\ta (d \omega))=d(j^*(\ta \omega)) \cg j^*(*(\no \dad \omega))=(-1)^{(p+1)(d-p+2)} d(j^*(* \no \omega)) \sgd
\end{equation}
\end{prop}
A difference between $d \ta \omega$ and $d j^* \omega$ is again significant here. Let $\omega$ be a two-form on a three-dimensional manifold with boundary. Then while $d j^* \omega$ is necessarily zero, $d \ta \omega$ need not be.

Another consequence of considering manifolds with boundary is related to inner products of k-forms. Remember-or see appendix \ref{ch:apprie}- that a natural inner product for k-forms is defined via the Hodge star operation as
\begin{equation}
\dket \omega,\eta \dbra = \int_{M} \omega \wedge * \eta
\end{equation}
where $\omega,\eta$ are k-forms. The following shows that $d$ and $\dad$ are not adjoint with respect to this inner product if $M$ has a boundary.
\begin{prop}[Green's Formula]
Let $\omega,\eta$ be (k-1) and k-forms respectively that are square integrable with respect to the inner product $\dket, \dbra$ defined above. Then
\begin{equation}
\dket d\omega, \eta \dbra = \dket \omega, \dad \eta \dbra + \int_{\pr M} \ta \omega \wedge * \no \eta
\end{equation}
\end{prop}

\section{Hodge-Morrey-Friedrichs Decomposition}\label{sec:hmf}
The usual Hodge decomposition is done on a compact manifold. Here we generalize this to a compact manifold with a boundary. The following definitions are important:
\begin{defn}[Harmonic, Dirichlet, Neumann fields]
Let $\Lkform(M)$ be space of square integrable k-forms on $M$, then we define
\begin{align}
H^1\Omega^k_D(M) &= \lbr \omega \in \Lkform(M) \, | \, \ta \omega=0 \rbr \sgc \\
H^1\Omega^k_N(M)&= \lbr \omega \in \Lkform(M) \, | \, \no \omega=0 \rbr \sgc \\
\harmk(M) &= \lbr \omega \in \Lkform(M) \, | \, d\omega=0 \,\, \& \,\, \dad \omega=0 \rbr \sgd
\end{align}
An element of $\harmk(M)$ is called a harmonic field. Elements of its subspaces
\begin{equation}
\harmk_D(M)=H^1\Omega^k_D(M) \cap \harmk(M) \quad \mbox{and} \quad \harmk_N(M)=H^1\Omega^k_N(M) \cap \harmk(M)
\end{equation}
are called Dirichlet and Neumann fields respectively.
\end{defn}
On a compact manifold without boundary a harmonic field and a harmonic form are equivalent but not so on a manifold with boundary. To illustrate this one defines the Dirichlet integral:
\begin{equation}
\dir(\omega, \eta) = \dket d\omega, d\eta \dbra + \dket \dad \omega , \dad \eta \dbra \sgd
\end{equation}
This integral is positive semi-definite inner product, more specifically both  $\dket d\lambda, d\lambda \dbra$ and $\dket \dad \lambda , \dad \lambda \dbra$ are both positive semi-definite. If $\dir( \lambda,\lambda)=0$ this implies $d\lambda=0$ and $\dad \lambda=0$. However one can show
\begin{equation}
\dir(\lambda, \lambda)= \dket \Delta \lambda , \lambda \dbra + \int_{\bnd} \ta \lambda \wedge *\no d \lambda - \int_{\bnd} t \dad \lambda \wedge * \no \lambda \sgc
\end{equation}  
so that only for the case there is no boundary $\Delta \lambda=0$ implies $\dir(\lambda,\lambda)=0$ and thus $d\lambda=0$ and $\dad \lambda=0$.
\begin{thm}[Hodge-Morrey Decomposition]
A k-form that is square integrable on a compact manifold $M$ with boundary can be uniquely decomposed as
\begin{equation}
\omega=d\alpha + \dad \beta + \kappa \sgc
\end{equation}
where $\alpha \in H^1\Omega^{k-1}_D(M)$,$\beta \in H^1\Omega^{k+1}_N(M)$ and $\kappa \in \harmk(M)$.
\end{thm}
Note that each piece in the decomposition above is orthogonal to each other with respect to the $\dket,\dbra$ inner product. One can further decompose the harmonic part:
\begin{thm}[Friedrichs decomposition]
A harmonic k-field that is square integrable on a compact manifold $M$ with boundary can be uniquely decomposed as:
\begin{enumerate}
\item $\kappa= \kappa_N + d\lambda$ where $\kappa_N \in \harmk_N(M)$ and $\dad d \lambda=0$,
\item $\kappa= \kappa_D + \dad \lambda$ where $\kappa_D \in \harmk_D(M)$ and $d \dad \lambda=0$.
\end{enumerate}
\end{thm}
\noindent
Using these two theorems one arrives at a third result:
\begin{cor}
A k-form that is square integrable on a compact manifold $M$ with boundary can be uniquely decomposed as:
\begin{enumerate}
\item $\omega= d \phi + \psi$ where $\dad \psi=0$ and $\no \psi=0$,
\item $\omega= \dad \rho + \sigma$ where $d \sigma=0$ and $\ta \sigma=0$.
\end{enumerate}
\end{cor}
\noindent
Using this decomposition one can show a version of Hodge theorem so that
\begin{equation}
\harmk_N(M) \cong H_p(M) \gag \harmk_D(M) \cong H_{d-p}(M) \sgc
\end{equation} 
Thus for a manifold with trivial homology, e.g. a ball $\harmk_N(M)=\harmk_D(M)=\lbr 0 \rbr$. 

With the use of \hmf~decomposition Dirichlet and Neumann problems can be solved. We will express them for the cases of trivial homology where there are no Dirichlet and Neumann fields.
\begin{thm}[Dirichlet (Neumann) problem] \label{thm:dir}
On a manifold with boundary and with trivial homology, given $\chi,\rho$ and $\psi (\phi)$, solution to the boundary problem
\begin{equation}
d\omega= \chi \cg \dad \omega = \rho \cg \ta \omega = \ta \psi \, \lp \no \omega = \no \phi \rp
\end{equation}
exists if and only if
\begin{equation}
d\chi=0 \cg \dad \rho=0 \cg \ta \chi = \ta d\psi \, \lp \no \rho = \no \dad \phi \rp \sgc
\end{equation}
and unique.
\end{thm}

\section{Homogeneous Spaces with a Geometry}\label{sec:homsp}
A Lie group is a smooth manifold, thus one can define a Riemannian metric on it. Metrics that connect to the group structure will be useful to study, and will appear in our construction. In the following we collect some definitions and results from the book of Arvanitogeorgos \cite{arvan} we will be using later on when we discuss the structure of the space of vacua. For a discussion of Lie groups see Chapter \ref{ch:applie} in the appendix.

\subsection{Geometry on a Lie Group}
\begin{defn}
A Riemannian metric $g$ on a Lie group $\G$ is called left invariant if
\begin{equation}
g(x)(u,v) = g \lp L_a x \rp \lp (L_a)_* u,(L_a)_* v \rp
\end{equation}
for all $x,a \in \G$, i.e. if $L_a$ is an isometry of the metric $\forall a \in \G$.
\end{defn} 

\begin{prop}
There is a one-to-one correspondence between left-invariant metrics on a Lie group $\G$ and scalar products on its Lie algebra $\g$.
\end{prop}

\begin{proof}
Let $g$ be a left invariant metric and $X,Y \in \g$ i.e. $X,Y$ are left-invariant vector fields. Let $e$ be the unit element of $\G$.\expl{After this I switch to $dL$ for a pushforward, otherwise i need to write things like ${{L_a}_*}_e$.} Then
\begin{equation}
g_{L_a(x)} \lp X_a, Y_a \rp= g_{L_a(x)} \lp {dL_a}_e X_e,{dL_a}_e Y_e \rp =g_e(X_e,Y_e) \sgc
\end{equation}
i.e. there is a one-to-one correspondence between metric evaluated at any point on left-invariant vector fields and a scalar product on $\g$. Since left-invariant vector fields constitute a global frame, this completes the proof. \expl{explain more?}
\end{proof}

A metric on $\G$ that is both left-invariant and right-invariant is called bi-invariant and there is a one-to-one correspondence between bi-invariant metrics on $\G$ and $\Ad$- invariant scalar products on $\g$, see subsection \ref{subsec:adrep} for defintion of adjoint representation and $\Ad$-invariance.

If $\G$ is a Lie group that is compact and semi-simple, the \hyperlink{killing form}{Killing form} provides a bi-invariant Riemann metric.

If $\G$ is a Lie group with a bi-invariant metric $g$, $\nabla$ covariant connection with respect to $g$ and $X,Y,Z \in \g$ then:
\begin{enumerate}
\item $\nabla_X Y = \frac{1}{2} \lb X,Y \rb$;
\item Geodesics of $\G$ starting at $e$, the identity element, are the one-parameter subgroups of $\G$;
\item $R(X,Y)Z= \frac{1}{4} \lb \lb X,Y \rb Z \rb$;
\item If $\G$ is also semi-simple and compact, then $Ric(X,Y)=-\frac{1}{4} B(X,Y)$, where B is the Killing form.
\end{enumerate}

\subsection{Geometry of Homogeneous Spaces}
In the previous subsection we have covered Lie Groups with a metric on, and have seen how a left invariant or bi-invariant metric reflects the group structure. Now we will consider homogeneous spaces, which we have defined in Section \ref{sec:apphom}, and the special case where one can define metrics that has invariance properties on them. 
\begin{defn}
A Riemannian homogeneous space is a Riemannian manifold $(M,g)$ on which its isometry group $\iso(M,g)$ acts transitively.
\end{defn}
\hyperlink{isot}{Isotropy group} of a given point on this manifold is a compact subgroup of $\iso(M,g)$. $\iso(M,g)$ is compact if and only if $M$ is compact.
\begin{defn}
Let $M$ be a homogeneous space with the action of a Lie group $\G$, and let $K$ be the isotropy subgroup. Recall, or see Section \ref{sec:apphom}, that $M$ is diffeomorphic to $\G/K$. $M$ is called a reductive homogeneous space if there exists a subspace $\alm$ of $\g$ such that $\g=\alm \oplus \alk$ and $\Ad(k) \alm \subset \alm$ for all $k \in K$.
\end{defn} 

If $M$ is the reductive homogeneous space described above, then
\begin{equation}
\alm \cong T_o \lp \G/K \rp .
\end{equation}
A distinguishing property of reductive spaces manifests itself in isotropy representations: Even though homogeneous spaces are not Lie groups, one can define a representation on them by virtue of the group action. For the case of reductive homogeneous spaces, this representation will be decomposable. 
\begin{defn}
Let $\tau_a : \G/K \ra \G/K$ be a diffeomorphism such that
\begin{equation}
\tau_a \lp \lb g \rb \rp = \lb a.g \rb \sgc
\end{equation}
where $.$ is the group multiplication of $\G$ and $\lb g \rb$ means the equivalence class of $g$ in $\G/K$. Then isotropy representation of $\G/K$ is the map $\Ad^{\G/K}: K \ra \GL\lp T_o \G/K \rp$ such that
\begin{equation}
\Ad^{\G/K}(k)=\lp d\tau_k \rp_o \sgd
\end{equation}
\end{defn}
For the case of reductive homogeneous space one has $\alm=T_o \G/K$ and the following proposition.
\begin{prop}
Let $M \cong \G/K$ be a reductive homogeneous space and $k \in K$,$X \in \alk$,$Y \in \alm$. Then
\begin{equation}
\Ad^{\G}(k)\lp X + Y \rp = \Ad^K(k) X + \Ad^{\G/K}(k)Y \sgd
\end{equation}
\end{prop}
Let $M \cong \G/K$ be a homogeneous space. A metric $g$ on $M$ is called $\G$-invariant if $\tau_a$ defined above is an isometry for it for all $a \in \G$.
\begin{prop}\label{prop:RieHom}
Let $M \cong \G/K$ be a homogeneous space. Then there is a one-to-one correspondence between
\begin{enumerate}
\item $\G$ invariant Riemannian metrics $g$ on $\G/K$.
\item $\Ad^{\G/K}$-invariant scalar products on $\alm$.
\end{enumerate}
If in addition $K$ is compact and $\alm=\alk^{\perp}$ with respect to the negative of the Killing form $B$ of $\G$, following is also equivalent to above: 
\begin{enumerate}[resume]
\item $\Ad^{\G/K}$-equivariant and $B$-symmetric operators $A: \alm \ra \alm$ such that $\ket X,Y \bra = B \lp A X, Y \rp$.
\end{enumerate}
\end{prop}

%% file: naturalformofgr.tex
\chapter{Natural Form of General Relativity}\label{ch:natgr}

Having developed the tools to apply the Manton approximation to \gr, we move onto using these tools. We first study general relativity on a foliated spacetime, and recover the Arnowitt-Deser-Misner (\adm) action in Section \ref{sec:splitting} and write down the slicing tensors we have defined in a coordinate basis. In Section \ref{sec:grconstr} we study the \adm~action as a constrained system, and discuss the constraints of \gr. As we will have seen in Section  \ref{sec:splitting} the splitting of space and time will not be enough to bring \gr~into the natural form, since we will still have space and time mixing. To get to the natural form we will choose Gaussian Normal Coordinates in Section \ref{sec:gnc}. In the next chapter we will see that the space of vacua will be set of diffeomorphisms of a reference solution: the remaining gauge transformations of the \gnc. However the remaining gauge transformations of \gnc~will have a complicated structure that we explore in Section \ref{sec:boosts}.

\section{Time and Space Split form of General Relativity} \label{sec:splitting}

General relativity on a spacetime that is a pseudo-Riemannian manifold $(\amb,\metamb)$ of dimension $\dimamb$, is most frequently described by the Einstein-Hilbert action given by
\begin{equation}
S_{EH}= \int_{\amb} \vol_{\metamb} \, \tilde{R}
\end{equation}
where $\vol_{\metamb}$ is the \hyperlink{rievol}{Riemannian volume form} for metric $\metamb$, see Appendix \ref{ch:apprie} for its definition, and $\tilde{R}$ Ricci scalar. Equations of motion belonging to this action give vacuum Einstein equations. Now we perform the foliation we have learned from Section \ref{sec:Fol} for our spacetime. 

By using equation \eqref{eq:ricsc} with the Lorentzian signature we write down the Ricci Scalar of the spacetime as
\begin{equation}
\tilde{R} = R + \ext^2 + \ext^{ij} \ext_{ij} + 2 \ez(\ext)- 2 (\consubm_i \acc)^i - 2 \acc_i \acc^i \sgd
\end{equation}
where $\dimsubm=\dimamb-1$ is the dimension of the spatial slices. Some of these terms are divergences and will yield a boundary term and not effect the equations of motion. To extract out those terms, we should know how they look like in the split form:
\begin{align}
\conamb_{\mu} v^{\mu} &= \lp \conamb_{e_{\mu}} v \rp^{\mu}  \nonumber \\
&= \eta^{\mu}_{\alpha} \omega^{\alpha} \lp \conamb_{e_{\mu}} v \rp \nonumber  \\
&= - \wz( \conambz v ) + \wi ( \conambi v) \sgd
\end{align}
Using this one can show
\begin{equation}
\conamb_{\mu} ( \ext \ez )^{\mu}= \ez(\ext) + \ext^2
\end{equation}
and
\begin{equation}
\conamb_{\mu} ( \acc^j e_j)^{\mu} = \pm \lp \pm (\consubm_i \acc)^i - \acc_i \acc^i \rp \sgd
\end{equation}
Suppressing these boundary terms the Einstein Hilbert action takes the form
\begin{equation}
S_{EH}= \int_{\amb} \lp R - \ext^2 + \ext^{ij} \ext_{ij} \rp \wz \wedge \epsilon_{\metsubm}  + S_{\pr \amb} \sgc
\end{equation}
where $\metsubm$ is the induced metric on the spatial slice and we used equation \eqref{eq:appvol}. This action is known in the literature as the \adm~action \cite{adm}. Now we rewrite our basis in terms of coordinate basis to get the conventional ADM picture. For this we let (\cite{mtw})
\begin{align}
\ez &= \frac{1}{N} \lp \frac{\pr}{\pr t} - N^i \frac{\pr}{\pr x^i} \rp, \\
\ei &= \frac{\pr}{\pr x^i}, \\
\wz &= - N dt, \\
\wi &= dx^i + N^i dt \sgc
\end{align} 
so that the metric becomes
\begin{equation}
\metamb = (N^k N_k-N^2) dt^2 + 2 N_i dx^i dt + h_{ij} dx^i dx^j \sgd
\end{equation}
From the commutation relations of the basis vector fields, we read off the \hyperlink{slicing tensors}{slicing tensors}:
\begin{align}
\lb \ei, \ej \rb &= 0 \sgc \\
\lb \ez, \ei \rb &= - \frac{\pr_i N}{N} \ez + \frac{\pr_i N^j}{N} \ej \sgc
\end{align}
so that
\begin{equation}
\Gamma^k_{[ij]}=0 \cg \acc_i = -\frac{\pr_i N}{N} \cg \twoconn_i^j - K_i^j= \frac{\pr_i N^j}{N} \sgd
\end{equation}
Noting that 
\begin{align}
2 \twoconn_{(ij)} &= \metamb( \conambz \ei, \ej ) + \metamb( \conambz \ej, \ei) \\
&= \ez(\metsubm_{ij}) \sgc
\end{align}
one gets for the extrinsic curvature
\begin{equation}\label{eq:extincoo}
\ext_{ij}= \frac{1}{2 N} \dot{\metsubm}_{ij} - \frac{1}{N} \consubm_{(i} N_{j)}
\end{equation}
or more compactly written
\begin{equation}
\ext_{ij}= \frac{1}{2} \lp \lie_{\frac{1}{N} \lp \frac{\pr}{\pr t} - N^i \frac{\pr}{\pr x^i} \rp} \metamb \rp_{ij}=\frac{1}{2} \lp \lie_{\ez} \metamb \rp_{ij} \sgd
\end{equation}
From this expression it is clear that the extrinsic curvature of the spatial integral manifolds measures how much the metric on them changes as you move in the direction orthogonal to them. To have more sense of what is going on, we illustrate these with an example.
\begin{example}
Consider the spatially flat \frw~space. The metric is:
\begin{equation}
ds^2= -dt^2 + a^2(t) (dx^i)^2 \sgd
\end{equation}
Let
\begin{equation}
\ez= \frac{\pr}{\pr t} \cg \ei= \frac{\pr}{\pr x^i} \sgc
\end{equation}
then
\begin{equation}
\lb \ez, \ei \rb=0 \sgc
\end{equation}
so that
\begin{equation}
\acc_i=0 \cg \twoconn_i^j=\ext_i^j \sgd
\end{equation}
Thus
\begin{equation}
K_{ij}= \frac{1}{2} \ez( a^2 \delta_{ij} ) = \dot{a} a \delta_{ij} \sgd
\end{equation}

\end{example}

\section{General Relativity as a Constrained System}\label{sec:grconstr}
In the previous section we have seen \eh~action in the time splitted form in a coordinate basis reduces to the \adm~action
\begin{equation}
S_{ADM}= \int_{\amb} \lp R(\metsubm) - K^2 + K^{ij} K_{ij} \rp \wz \wedge \epsilon_{\metsubm} \sgc 
\end{equation}
plus some surface terms that do not affect the equations of motion, where
\begin{equation}
\ext_{ij}= \frac{1}{2 N} \dot{\metsubm}_{ij} - \frac{1}{N} \consubm_{(i} N_{j)} \sgd
\end{equation}
One can think of this shift from \eh~action to \adm~action as changing variables of the theory from $\metamb$ to the new set $\metsubm,N,N_i$. Writing $\wz$ explicitly, and letting the spacetime manifold to be $\amb= \R \times \subm$ the action in terms these quantities is
\begin{equation}
S_{ADM}= \int dt \int_{\amb} N  \lp  R(\metsubm) - K^2 + K^{ij} K_{ij} \rp \epsilon_{\metsubm} \sgd
\end{equation}
The quantities $N,N_i$ are known in the literature as lapse and the shift. In the coordinate basis $\{ t,x^i \}$ we have introduced above, lapse is simply a measure of how much further the next spatial slice is, and shift measures how much spatially a coordinate $x^i$ is shifted in the next slice \cite{mtw}.

One must note that the lapse and shift in the above Lagrangian are not dynamical variables, no time derivative of them enters the action. Because of this property, their variations will not give us dynamical equations but rather constraint equations. One of them is directly obtained by variation of $N_i$ :
\begin{equation}
\momi = \consubm_j \ext^{ji} - \consubm^i \ext=0
\end{equation}
and actually corresponds to the Einstein equation $\tilde{R}_{0i}=0$. This equation is called the momentum constraint. The other is obtained via the variation of the shift $N$:
\begin{equation}
\ham= ^{(d)}R + \ext^2 - \ext^{ij} \ext_{ij} = 0 \sgc
\end{equation}
and corresponds to the equation $+ \tilde{R}_{00} + \metsubm^{ij} \tilde{R}_{ij}=0$ coming from Einstein equations. This constraint is called the Hamiltonian constraint. The remaining equation is the dynamical equation involving the double time derivative of the $\metsubm$, which we will not express here.

As discussed in Section \ref{sec:appconstr} the set of constraints are closely related to the set of gauge symmetries. In the Hamiltonian formalism the Hamiltonian vector fields corresponding to the constraints are generators of the gauge transformations. Since one has \expl{is this true for gr?}
\begin{equation}
\lb X_f, X_g \rb= - X_{\{f,g\}} \sgc
\end{equation}
the algebra of generators of gauge transformations on the phase space will be directly related to the Poisson bracket algebra of the constraints. Since for our case the homogeneous space structure of the solution space is directly related to algebra of gauge transformations, these will be important for us. For this reason let us note down the Poisson bracket algebra of the constraints of General Relativity \cite{ishamkuchar1,wald}. For the proper formulation we need to define \expl{why?-see Witten paper-how to define phase space etc on field thy}
\begin{align}
C_{\chi} &=\int_{\subm} \chi_i(x) \lp \consubm_j \ext^{ji} - \consubm^i \ext \rp \vol_{\metsubm} \sgc\\
C_f &= \int_{\subm} f(x) \lp R + \ext^2 - \ext^{ij} \ext_{ij} \rp \vol_{\metsubm} \sgc
\end{align}
where $\chi$ is an arbitrary spatial vector field and $f(x)$ is an arbitrary function. To calculate the Poisson bracket we should pick a phase space for the theory. The dynamical variable for the theory is $\metsubm$, and the corresponding canonical momentum for it is
\begin{equation}
\pi^{ij}= \sqrt{h} \lp \ext^{ij} - \metsubm^{ij} \ext \rp \sgd
\end{equation}
The corresponding symplectic form in Darboux coordinates will be
\begin{equation}
\omega = \int \delta \pi^{ij} \wedge \delta h_{ij} \sgc
\end{equation}
where $\delta$ and $\wedge$ denote the exterior derivative and the exterior product on the phase space. The Poisson brackets of the constraints then are calculated as
 \begin{align}
\lbr C_{\chi_1},C_{\chi_2} \rbr &= C_{\lb \chi_1, \chi_2 \rb} \sgc \\
\lbr C_{\chi}, C_{f}  \rbr &= C_{\chi(f)} \sgc \\
\lbr C_f,C_g \rbr &= C_{\eta} \sgc
\end{align}
where $\eta$ is also a spatial vector field such that
\begin{equation}
\eta^i=\metsubm^{ij} \lp f \pr_j g - g \pr_j f \rp \sgd
\end{equation}
This set of relations are known as Dirac relations. Since there is field dependency-i.e. they depend on the $\metsubm$, the set of constraints do not form a Lie algebra, e.g. they do not satisfy the Jacobi identity \cite{ishamkuchar1,weinstein}. This property is considered to have important implications, e.g. if the constraints would form a Lie algebra the the general relativity could have been quantized. For a discussion of this point see e.g. \cite{LeeWald,ishamkuchar1,ishamkuchar2}.

\section{General Relativity in Gaussian Normal Coordinates}\label{sec:gnc}
\noindent
The \adm~form of the action 
\begin{equation}
S_{ADM}= \int dt \int_{\amb} N  \lp R(\metsubm) - K^2 + K^{ij} K_{ij} \rp \epsilon_{\metsubm} 
\end{equation}
illustrates the time splitting, however it is not still in the natural form: there are terms that are first order in time derivatives. However as one can easily see from the expression for the extrinsic curvature \eqref{eq:extincoo}, it can be brought into the natural form by choosing coordinates such that $N=1$ and $N_i=0$. This choice is known as Gaussian Normal Coordinate (\gnc) \cite{wald}. In our notation this is equivalent to choosing $[\ez,\ei]=0$, i.e. choosing $\ez,\ei$ to be coordinates. Note that this choice implies
\begin{equation}
\acc^i=0 \ra \conambz \ez=0 \sgc
\end{equation}
i.e. $\ez$ is Riemannian geodesic in the ambient manifold and
\begin{equation}
f_i^j=K_i^j \sgd
\end{equation}
Note that since this choice will eliminate the lapse and shift from the action, one needs to impose now the constraints by hand-they will not come from the variation of the new action. With this in mind, in the new coordinate frame our theory is defined by the action in the natural form:\\
\hfill\\
\fbox{\begin{minipage}{\linewidth}
\vspace{3mm}
\begin{equation}
S=\int dt \lp \frac{1}{2} \metonvacua_{\metsubm}(\dot{\metsubm}_{ij},\dot{\metsubm}_{kl}) - V(\metsubm) \rp
\end{equation}
where
\begin{equation}
\metonvacua_{\metsubm}(\delta_1 h,\delta_2 h) = \frac{1}{2}\int_{\subm} d^dx \sqrt{\det \metsubm}  \left( h^{ik} h^{jl} -  h^{ij} h^{kl} \right) \delta_1 h_{ij} \delta_2 h_{kl} \sgc
\end{equation}
and
\begin{equation}
V(\metsubm) =-\int_{\subm} d^dx \sqrt{\det \metsubm} \, {R(\metsubm)} \sgc
\end{equation}
\end{minipage}}\\ 
\hfill\\
together with the constraints imposed by hand:
\begin{align}
\momi &= \consubm_j \ext^{ji} - \consubm^i \ext=0 \sgc \\
\ham &= ^{(d)}R + \ext^2 - \ext^{ij} \ext_{ij} = 0 \sgc
\end{align}
where
\begin{equation}
\ext_{ij}= \frac{\dot{h}_{ij}}{2} \sgd
\end{equation}
The remaining gauge transformations after this choice are interesting and can be found by requiring that the change of metric under them does not violate the Gaussian Normal condition:
\begin{align*}
\delta_{\xi} \metamb_{00} = \lie_{\xi} \metamb_{00}=0 &\ra \conamb_0 \xi_0 =0 \sgc\\
  \lie_{\xi} \metamb_{0i}=0 & \ra \conamb_0 \xi_i + \conamb_i \xi_0 = 0 \sgc
\end{align*}
where we have used the infinitesimal form of the transformation equations so that they can be explicitly solved for the diffeomorphisms $\xi$. Doing this results in one set that is the 3-spatial diffeomorphisms $\chi(x)$ -argument depends only on space- and another set given by the \uline{field-dependent} diffeomorphisms
\begin{equation}\label{Normal tfm}
\xi_f[h]= \lp f(x), \partial_j f(x) \int_{t_0}^{t} \metsubm^{ji}(x,t') dt' \rp = \lp f(x), \zeta^i_f[h]\rp \sgd
\end{equation}
We will call this second set local boosts. Can we get the geometric insight behind these transformations? First the geometric insight behind the Gaussian Normal Coordinates \cite{wald}: take a spatial hypersurface, define $Q$ along this hypersurface, the unique orthogonal and normalized vector field. 
\begin{figure}[!hb]
  \begin{center}
    \includegraphics[width=\textwidth]{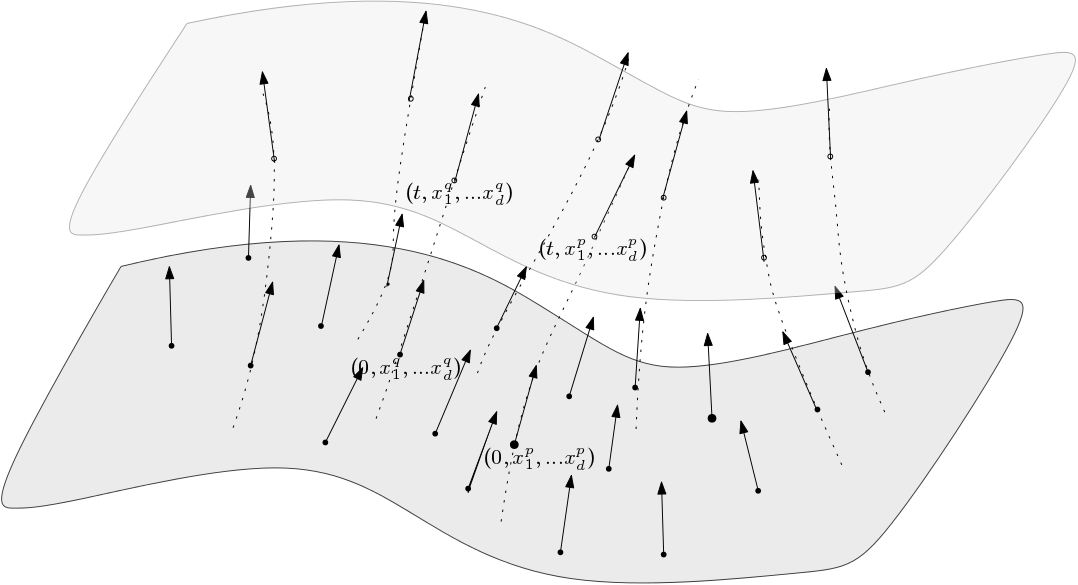}
  \end{center}
  \caption{Gaussian Normal Coordinates}
\end{figure}
Construct the \uline{the unique geodesic} going through each point on the hypersurface such that at that point the tangent to the geodesic is $Q$ evaluated at that point.(See Appendix \ref{ch:apprie} for the discussion of uniqueness) \\
Now choose some spatial coordinates $(x^1,...,x^D)$ for each point on the hypersurface.\\
Label all geodesics with parameter t.\\
Let $(x^1,...,x^D,t)$ to be the point on the geodesic that started from $(x^1,...,x^D)$ and on the geodesic ``t"-away from it.\\
Now, the remaining freedom to choose coordinates are:
\begin{itemize}
\item Arbitrary relabeling of $(x^1,...,x^D)$,
\item Choice of initial hypersurface.
\end{itemize} 
The second freedom in the choice corresponds to local boosts and might also be considered as redefining the orthogonal vector field $Q$: for this to work one should make an accompanying spatial transformation so that vector field remains orthogonal to the hypersurface. Reconsider the transformation \eqref{Normal tfm} in the light of this: change in the time component is arbitrary (except that it should be independent of time to make new orthogonal vector unit norm), but that is the all freedom you have, change in other coordinates are uniquely determined such that your new hypersurface is orthogonal to your new $Q$.
To get a better intuition of GNC and local boosts, and to support to fact that the role of local boosts is to change the choice of initial hypersurface in way to keep properties of the \gnc, let us consider the following simple examples.
\begin{example}
Consider 2-dimensional flat Euclidean space such that the metric is
\begin{equation*}
ds^2= dx^2 + dy^2 \sgd
\end{equation*}
These coordinates are Gaussian normal. Now consider the coordinate transformations that preserves the Gaussian normal condition. They should satisfy
\begin{align*}
\partial_x \xi_x =0 \quad \partial_x \xi_y + \partial_y \xi_x =0 \sgd
\end{align*}
One set of solution to this gives local boosts,
\begin{equation*}
\xi= \lp f(y), - f'(y) x \rp \sgd
\end{equation*}
Let us make some simple choices for the function $f$: when $f=\mbox{constant}$ , we get translations in $x$ direction, which simply shifts the initial hypersurface in the $x$ direction, whereas for $f=y$ we get rotations.
Now I would like to show that the coordinates after making this transformation are orthogonal coordinates:
\begin{equation*}
x'=x+f(y) \quad y'=y-f'x \sgd
\end{equation*}
Slope of the constant $x'$ and $y'$ lines are respectively
\begin{align*}
\frac{dx}{dy}=-f' \gag \frac{dx}{dy}= \frac{1}{f'}- \frac{f''}{f'} x \sgd
\end{align*}
Thus to first order in $f$, and this is the order we should keep ourselves in since we are considering infinitesimal coordinate transformations, they are orthogonal.
\end{example}
\begin{example}
Consider the flat 3-dimensional Euclidean surface in spherical coordinates:
\begin{equation}
ds^2=dr^2+r^2 (d\theta^2+ sin^2\theta d\phi^2) \sgc
\end{equation}
so that
\begin{equation}
Q= \frac{\partial}{\partial r} \cg K_{ij}=2r \, \diag(0,1,sin^2\theta) \cg K=\frac{4}{r} \sgd
\end{equation}
Part of remaining gauge transformation will be again local boosts:
\begin{align*}
r' &=r+f(\theta,\phi) \sgc \\
\theta' &= \theta + \frac{\partial_{\theta}f}{r} \sgc \\
\phi' &= \phi + \frac{\partial_{\phi}f}{r \sin^2\theta } \sgd
\end{align*}
Try the case $f= \epsilon \, \sin\theta$. As we have claimed before to keep the properties of \gnc~new coordinate $r'$ should be orthogonal to others, e.g. $r'=\mbox{constant}$ and $\theta'=\mbox{constant}$ lines are orthogonal. If you check one of them satisfies $ \frac{dy}{dx} = - \frac{x}{y} \left( 1 - \frac{\epsilon}{y} \right) $ and other $ \frac{dy}{dx}= \frac{y}{x} \left( 1 + \frac{\epsilon}{y} \right) $ so they are orthogonal at the order $\epsilon$.
\end{example}

\section{Local Boosts and Dirac Algebra}\label{sec:boosts}
%%%following i take from my dynamical embedding notes%%%
In the next chapter we will apply Manton approximation to general relativity in the natural form, in a spirit similar to that of Weinberg's proof for adiabatic modes: We will consider solutions that are diffeomorphisms of a vacuum solution, then adding a parameter dependence we will argue that these are no longer equivalent to the vacuum solution and physically new. As we have shown above in the natural form, the remaining gauge transformations are composed of spatial diffeomorphisms and local boosts. The space of solutions we will consider will be gauge transformations of a reference metric under these transformations and the tangent space of the space of solutions will be given by an infinitesimal transformation of the reference point under these gauge transformations.  However since local boosts are field dependent transformations, variations of fields under them and especially algebra of these variations are highly non-trivial. As we signaled in Section \ref{sec:grconstr} this algebra will turn out to be the Dirac relations, the algebra of constraints of \gr:
\que{what does dirac relations tells us about space of vacua? why did we do this calculation?what is the difference between this and $[\delta_1 h, \delta_2 h]$ Is tangent vector $\delta h$ or $\delta$ itself?}
\begin{align}
\lb \delta_{\chi_1} , \delta_{\chi_2} \rb h = \delta_{[\chi_1,\chi_2]} h \sgc \label{rln1} \\ 
\lb \delta_{\chi}, \delta_f \rb h = \delta_{\lie_{\chi}f} h \sgc \label{rln2} \\ 
\lb \delta_{f_1}, \delta_{f_2} \rb h = \delta_{\eta(f_1,f_2,h(t_0))} h \sgc \label{rln3}
\end{align}
where
\begin{equation}
\eta(f_1,f_2,h(t_0))^i = h^{ij}(x,t_0) \lp f_2(x)  \partial_j f_1(x) - f_1(x) \partial_j f_2(x) \rp
\end{equation} 
and $\delta_f h$ means the variation of $h$ \textit{as a tensor in 4-dim} under local boosts \eqref{Normal tfm} and $\delta_{\chi} h$ is just usual spatial variation:
\begin{align}
\delta_f h &= \lie_{\xi_f[h]} h = f(x) \dot{h} + \lie_{\zeta_f[h]} h \sgc \\
\delta_{\chi} h &= \lie_{\chi}h \sgd
\end{align}
Note that this would mean that the ``tangent space" of the space of vacua generated by these transformations would not be a proper tangent space, since the tangent space of a differentiable manifold is expected to be a Lie algebra with respect to the commutator bracket. Let us now move onto showing how this set of commutators come about to illustrate the hardship coming from field dependency of the gauge transformations. First relation \eqref{rln1} is trivial by the fact
\begin{equation}
\lie_{[V,W]}h= \lie_V \lie_W h - \lie_W \lie_V h \sgd
\end{equation}
For the field dependent variations one also needs to take into account the fact that generators of the transformations are changing. Thus for the second relation \eqref{rln2} one should consider the following picture:
\begin{equation}
h \xrightarrow{\delta_{\chi}} h+ \lie_{\chi} h \xrightarrow{\delta_f} \lp h + \lie_{\chi}h \rp + \lie_{\xi_f[h+ \lie_{\chi}h]} \lp h + \lie_{\chi} h \rp \sgd
\end{equation}
Note that in the second step it is important to remember the local boost is acting on the varied metric of the first step, thus its dependency is on this varied metric. To get the commutation relation also consider the variation in the reverse order: 
\begin{equation}
h \xrightarrow{\delta_f} h+ \lie_{\xi[h]} h \xrightarrow{\delta_{\chi}}  h + \lie_{\chi}h + \lie_{\xi[h]} h + \lie_{\chi} \lie_{\xi[h]} h
\end{equation}
where local boosts now are functions of the original metric. $\lb \delta_f, \delta_{\chi} \rb h $ is the difference between last two expressions. Here once again remember what $\xi[h]$ is:
\begin{equation}
\xi[h]= \lp f(x), \partial_j f(x) \int_{t_0}^{t} h^{ji}(x,t') dt' \rp = \lp f(x), \zeta^i_f[h]\rp \sgc
\end{equation}
then
\begin{equation} \label{tilde zeta dfn}
\xi[h+\lie_{\chi}h]=\xi[h]- \lp 0, \partial_j f(x)  \int_{t_0}^{t} h^{jl}(x,t') h^{ik}(x,t') \lp \lie_{\chi}h \rp_{lk}(x,t') dt' \rp \sgc
\end{equation}
where we define
\begin{equation}
\tilde{\zeta}_f \lb \lie_{\chi}h \rb=\lp 0, \partial_j f(x)  \int_{t_0}^{t} h^{jl}(x,t') h^{ik}(x,t') \lp \lie_{\chi}h \rp_{lk}(x,t') dt' \rp \sgd
\end{equation}
We put a tilde over the zeta to emphasize that we are not actually using inverse of the $\lie_{\chi}h$ in the argument for the local boost. Thus
\begin{equation}
\lb \delta_f, \delta_{\chi} \rb h= - \lie_{\tilde{\zeta}_f \lb \lie_{\chi}h \rb} h + \lie_{[\xi[h],\chi]} h \sgd
\end{equation}
Then by using
\begin{equation}
\tilde{\zeta}_f \lb \lie_{\chi}h \rb^i - \lb \zeta_f[h], \chi \rb^i = \zeta^i_{\lie_{\chi} f} [h]
\end{equation}
one gets
\begin{equation}
\boxed{\lb \delta_f, \delta_{\chi} \rb h = - \lp \lie_{\chi}f \rp \dot{h}- \lie_{\zeta_{\lie_{\chi} f} [h]} h = - \delta_{\lie_{\chi}f} h} \sgd
\end{equation}
For the third relation note
\begin{equation}
h \xrightarrow{\delta_{f_1}} h+ \lie_{\xi_1[h]} h \xrightarrow{\delta_{f_2}} h + \lie_{\xi_1[h]}h + \lie_{\xi_2[h+ \lie_{\xi_1}h]} h + \lie_{\xi_2[h]} \lie_{\xi_1[h]} h \sgc
\end{equation}
where
\begin{equation}
\xi_2 \lb h+ \lie_{\xi_1}[h] \rb = \xi_2[h]+ \lp 0, - \partial_j f_2(x) \int h^{jk} h^{il} \lp \lie_{\xi_1[h]}h \rp_{kl} \rp
\end{equation}
and
\begin{equation}
- \partial_j f_2(x) \int h^{jk} h^{il} \lp \lie_{\xi_1[h]}h \rp_{kl}= f_1(x) \dot{\zeta}^i_2[h] - f_1(x) \dot{\zeta}^i_2[h(t_0)]- \tilde{\zeta}_2 \lb \lie_{\zeta_1[h]} h \rb \sgc
\end{equation}
where $\tilde{\zeta}$ is defined by \eqref{tilde zeta dfn}. Using this one gets
\begin{equation}
\xi_2 \lb h + \lie_{\xi_1[h]}[h] \rb= \xi_2[h] + f_1(x) \dot{\zeta}_2[h]- f_1(x) \dot{\zeta}^i_2[h(t_0)] - \tilde{\zeta}_2 \lb \lie_{\zeta_1[h]} h \rb \sgd
\end{equation}
Then
\begin{equation}
\lb \delta_{f_1}, \delta_{f_2} \rb h = \lie_{ \lp f_1 \dot{\zeta}_2 - f_1 \dot{\zeta}_2(t_0)-\tilde{\zeta}_2 \lp \zeta_1 \rp \rp} h - \lie_{ \lp f_2 \dot{\zeta}_1 - f_2 \dot{\zeta}_1(t_0)-\tilde{\zeta}_1 \lp \zeta_2 \rp \rp} h + \lie_{\lb \xi_2[h], \xi_1[h] \rb} h \sgc
\end{equation}
where I define
\begin{equation}
\tilde{\zeta}_2 \lb \lie_{\zeta_1[h]} h \rb = \tilde{\zeta}_2 \lp \zeta_1 \rp  \sgd
\end{equation}
Now open up $\lie_{\lb \xi_2[h], \xi_1[h] \rb} h$ and note
\begin{align}
\lie_{f_2 \dot{\zeta}_1 - f_1 \dot{\zeta}_2} h &= f_2  \lie_{\dot{\zeta}_1} h - f_1 \lie_{\dot{\zeta}_2} h \label{frln} \sgc\\
\lie_{\zeta_2} f_1 - \lie_{\zeta_1} f_2 &= 0 \sgc
\end{align}
where for the derivation of first equation one must see
\begin{equation}
\lie_{fX}h= f \lie_X h + (X \inc h) \otimes df + df \otimes ( X \inc h ) \sgc
\end{equation}
where $h$ is a symmetric rank-2 tensor. Using these one gets
\begin{equation}
\lb \delta_{f_1}, \delta_{f_2} \rb h = \lie_{ \lb \zeta_2, \zeta_1 \rb - \tilde{\zeta}_2 (\zeta_1) +\tilde{\zeta}_1 (\zeta_2) } h - \lie_{\lp f_1 \dot{\zeta}_2(t_0) - f_2 \dot{\zeta}_1(t_0) \rp } h \sgd
\end{equation}
The first term is zero since
\begin{align}
\tilde{\zeta}_2(\zeta_1) &= \partial_j f_2 \int h^{jl} h^{mi} \lp \lie_{\zeta_1[h]}h \rp_{ml} \nonumber \\
&= \int \dot{\zeta}^k_2 h^{mi} \lp \zeta^l_1 \partial_l h_{mk} + h_{ml} \partial_k \zeta^l_1 + h_{kl} \partial_m \zeta^l_1 \rp \nonumber \\
&=\int \partial_l \lp \dot{\zeta}^k_2 h_{mk} \rp h^{mi} \zeta^l_1- \zeta^l_1 h^{mi} h_{mk} \partial_l \dot{\zeta}^k_2 + \dot{\zeta}^k_2 \partial_k \zeta^i_1 + \dot{\zeta}^k_2 h_{kl} h^{mi} \partial_m \zeta^l_1 \nonumber \\
&= \int \dot{\zeta}^k_2 \partial_k \zeta^i_1 - \zeta^l \partial_l \dot{\zeta}^i_2 + \int \partial_l \partial_m f_2 h^{mi} \zeta^l_1 + \partial_l f_2 h^{mi} \partial_m \zeta^l_1 \nonumber \\
&= \int \dot{\zeta}^k_2 \partial_k \zeta^i_1 - \zeta^l \partial_l \dot{\zeta}^i_2 + \int h^{mi} \partial_m \lp \zeta^l_1 \partial_l f_2 \rp \sgc
\end{align}
so that
\begin{align}
\tilde{\zeta}_2(\zeta_1)-\tilde{\zeta}_1(\zeta_2)= \int  \lb \dot{\zeta}_2, \zeta_1 \rb -  \lb \zeta_1, \dot{\zeta}_1 \rb + \int h^{mi} \partial_m \lp \zeta^l_1 \partial_l f_2 - \zeta^l_2 \partial_l f_1\rp \sgd
\end{align}
The last term is zero and
\begin{equation}
\tilde{\zeta}_2(\zeta_1)- \tilde{\zeta}_1(\zeta_2)= \int \frac{d}{dt} \lb \zeta_2, \zeta_1 \rb = \lb \zeta_2, \zeta_1 \rb \sgc
\end{equation}
since $\zeta_2(t_0)=\zeta_1(t_0)=0$. So one finally gets 
\begin{equation}
\boxed{ \lb \delta_{f_1}, \delta_{f_2} \rb h =  \lie_{\eta(f_1,f_2,h(t_0))} h \sgc}
\end{equation}
where
\begin{equation}
\eta(f_1,f_2,h(t_0))^i= -\lp f_1 \dot{\zeta}_2(t_0) - f_2 \dot{\zeta}_1(t_0) \rp^i = h^{ij}(x,t_0) \lp f_2(x)  \partial_j f_1(x) - f_1(x)  \partial_j f_2(x)\rp \sgd
\end{equation}
This relations should be true for any tensor on the space-time that does not have a component in the time direction and satisfies \eqref{frln}. 

This completes the proof. We have seen that the algebra of variations of the spatial metric $h$ under the remaining gauge transformations give the Dirac algebra. However the fact that the constraint algebra is Dirac suggests that you shouldn't need to make a gauge choice to get this algebra. In fact one can show that the variations under the full diffeomorphism group written in a comoving basis also gives these commutation relations. This shouldn't be surprising though, choosing Gaussian Normal Coordinates should be somehow equivalent to choosing a comoving basis. One slightly disturbing fact is in the first case you have the full diffeomorphism group whereas in the second case you have some subspace of that group. One might say that remaining gauge transformations of \gnc~form a nice subspace of diffeomorphisms such that it preserves the Dirac algebra.

%% file: adiabaticgr.tex
\chapter{Adiabatic Solutions of General Relativity} \label{ch:adgr}
In the last chapter we have brought general relativity into a form suitable for the Manton approximation. In this chapter we apply the Manton approximation. We start by writing down the form of the solution suggested by Manton approximation in Section \ref{sec:adgr}, discussing what kind of spacetime it is proposing. In Section \ref{sec:constaction} we write down the constraints and action for these set of solutions and we simplify the action by using orthogonal decomposition on the spatial metric. We see that it becomes a theory on the boundary, however still involves information from the bulk, since there will be orthogonal derivatives of the tangent vectors. In Section \ref{sec:boundarydata} however, we show for a tangent vector on the boundary its orthogonal derivative is uniquely given by the momentum constraint. In Section \ref{sec:decconst}, we attempt at solving constraints explicitly for various types of cases by using again the orthogonal decomposition on the spatial manifold.

\section{Manton Approximation for General Relativity} \label{sec:adgr}
In the previous chapter we have concluded that one can bring the \eh~action into the natural form by choosing \gnc . In these coordinates we had
\begin{equation}
S_{EH}=\int dt \lp \frac{1}{2} \metonvacua_{\metsubm}(\dot{\metsubm}_{ij},\dot{\metsubm}_{kl}) - V(\metsubm) \rp
\end{equation}
where
\begin{equation}\label{eq:wdw}
\metonvacua_h(\delta_1 h,\delta_2 h) = \frac{1}{2}\int_{\subm} d^dx \sqrt{\det h}  \left( h^{ik} h^{jl} -  h^{ij} h^{kl} \right) \delta_1 h_{ij} \delta_2 h_{kl} \sgc
\end{equation}
and
\begin{equation}
V(h) =-\int_M d^dx \sqrt{\det h} \, {R(h)} \sgd
\end{equation} 
We see with this choice our theory is fully described by the spatial metric $\metsubm$. Note that we wrote down the kinetic term in the way we did to emphasize that this will define a metric, $\metonvacua$, on the space of vacua. Now we move onto applying the Manton approximation. For this, we again consider the configuration space of spatial metrics that depend only on space, $\consp$. After this we need to identify the space of vacua $\minsp$. We define it as the set of solutions -to $\spmet$- that extremize the potential energy which means for our case,we are looking for solutions such that
\begin{equation}
\delta \pot = \int \sqrt{\det \spmet} ( \rthree_{ij} - \frac{1}{2} \rthree \spmet_{ij} ) \delta \spmet^{ij} =0 \sgc
\end{equation}
i.e.
\begin{equation}\label{eq:VacEins}
\rthree_{ij}(\vacspmet)=0 \sgd
\end{equation}
where we denote elements of $\minsp$ with $\vacspmet$. Note that unlike the ``usual" cases, extremum here will not imply absolute minimum since the potential is not bounded from below. However note that, via the equations of motion, this condition is necessary to have time independent solutions.

The vacua $\minsp$ of the configuration space for this case will be solutions to above. For a four dimensional spacetime these will be the flat Euclidean metric and its gauge transformations since solutions to \eqref{eq:VacEins} are three dimensional Ricci flat spaces and they are purely gauge since in 3d there is no physical degree of freedom. For dimensions bigger than 4, there might be metrics inequivalent to flat metric with zero Ricci tensor, and there we will further say that we focus our attention only to the part of vacua generated from the flat metric. Note that the gauge transformations of the spatial metric will be spatial diffeomorphisms. But what have happened to local boosts? Analysis of the last chapter has told us that under transformations by local boosts our theory of the spacetime should be invariant. However if we take a solution to \eqref{eq:VacEins} and transform it under the local boost it will not be solving the equation again! Where did we lose the local boosts? Of course when imposed a condition on our spacetime that is not ``covariant": extremization of the potential energy. The job of the boosts will be transferring ``some energy" from potential to kinetic energy, while keeping the ``total energy" constant. 

Because of this conceptual difference, it looks like it is not easy to put local boosts into the Manton approximation picture. Here \uline{we will only consider the set of vacua generated by the purely spatial coordinate transformations}, i.e. we will take $\minsp$ to be composed of the metrics written as
\begin{equation}\label{eq:flatspt}
\vacspmet_{ij}(x) = \frac{\partial \tx^\uk(x,z)}{\partial x^i} \frac{\partial \tx^\um(x,z)}{\partial x^j} \refmet_{\uk \um} \sgc
\end{equation}
where $z=\lbr z^a \rbr$ labels all families of spatial gauge transformations, and $\refmet$ is a reference spatial metric that has a vanishing Ricci curvature. Manton approximation then amounts to considering solutions of the form
\begin{equation}\label{eq:nonflatsp}
\vacspmet_{ij}(x,t)= \frac{\partial \tx^\uk(x,z(t))}{\partial x^i} \frac{\partial \tx^\um(x,z(t))}{\partial x^j} \refmet_{\uk \um} \sgd
\end{equation}
Note that with this choice of the spatial metric our spacetime manifold is no longer equivalent to a flat manifold, since spatial coordinate transformations with time dependency were not a part of remaining gauge transformations of \gnc. Spacetimes that first and second spatial metric we wrote down define can be visualized as follows: For the first case in \eqref{eq:flatspt}, we only perform a spatial coordinate transformation of a flat metric in \gnc. In other words each spatial slice we have in the transformed spacetime is obtained by transforming one flat spatial slice with the same coordinate transformation, see Figure \ref{fig:flatspt}.
\begin{figure}[!ht]
\centering
\includegraphics[width=0.9\linewidth]{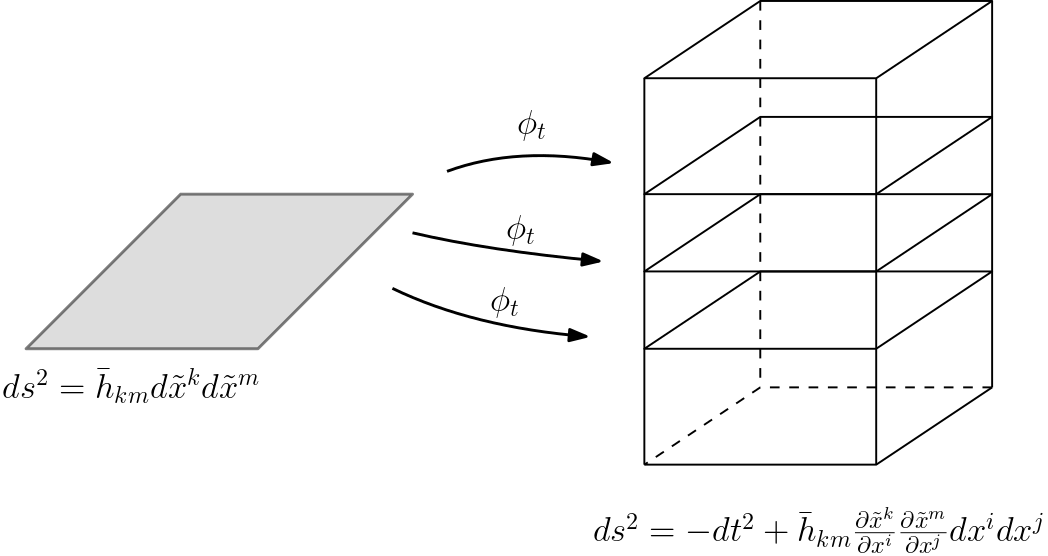}
\caption{Flat spacetime obtained by a unique spatial coordinate transformation}
\label{fig:flatspt}
\centering
\end{figure}
For the second case each spatial slice is obtained by a different spatial coordinate transformation, thus it is will not be diffeomorphic to the flat spacetime, see Figure \ref{fig:nonflatspt}.
\begin{figure}[!ht]
\centering
\includegraphics[width=0.8\linewidth]{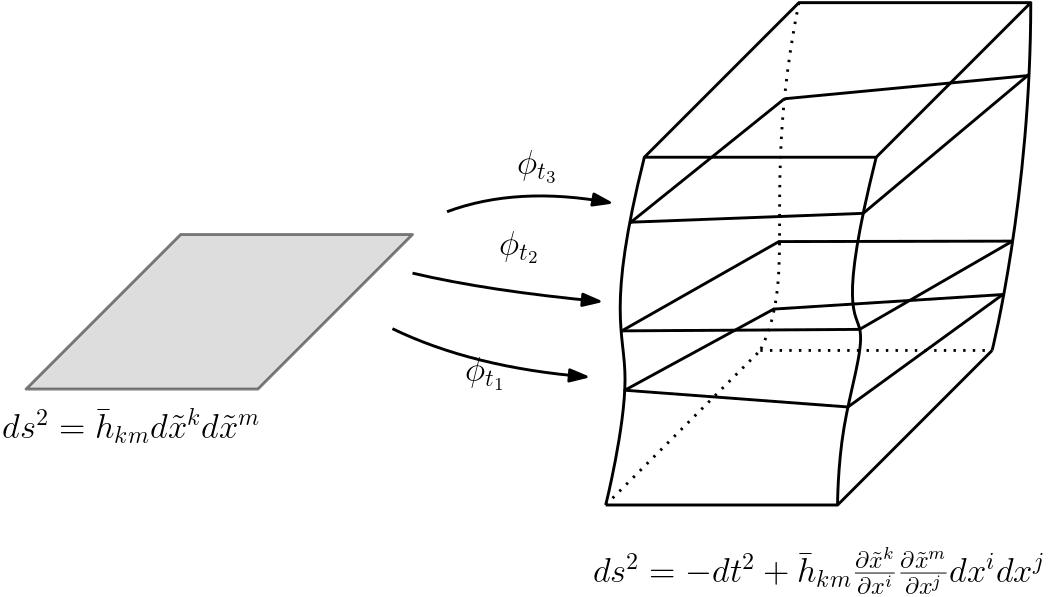}
\caption{Non-flat spacetime obtained by choosing a different spatial coordinate transformation of flat spatial hypersurface for each spatial slice}
\label{fig:nonflatspt}
\centering
\end{figure}
One can rephrase our proposed solutions more geometrically as follows:
\begin{equation} \label{eq:Manapp}
\vacspmet(t)= \sigma_t^* \refmet \sgc
\end{equation}
i.e. take a spatial slice with a reference metric. Consider a flow $\sigma(x,t): \subm \times \R \ra \subm$, where $\subm$ is to be taken as the spatial slice now. Our construction tells us to construct the spacetime manifold $\stm$ as a manifold whose slice at $t$ is given by the pullback of $\refmet$ by $\sigma_t$. Note this spacetime will correspond to an integral curve on the space of vacua $\minsp$, see Figure \ref{fig:adsl}. Our aim will be to identify $\minsp$ via solving for all possible flows on it such that they describe solutions to Einstein's equations (in vacuum).\\
\begin{figure}[!ht]
\centering
\includegraphics[width=\linewidth]{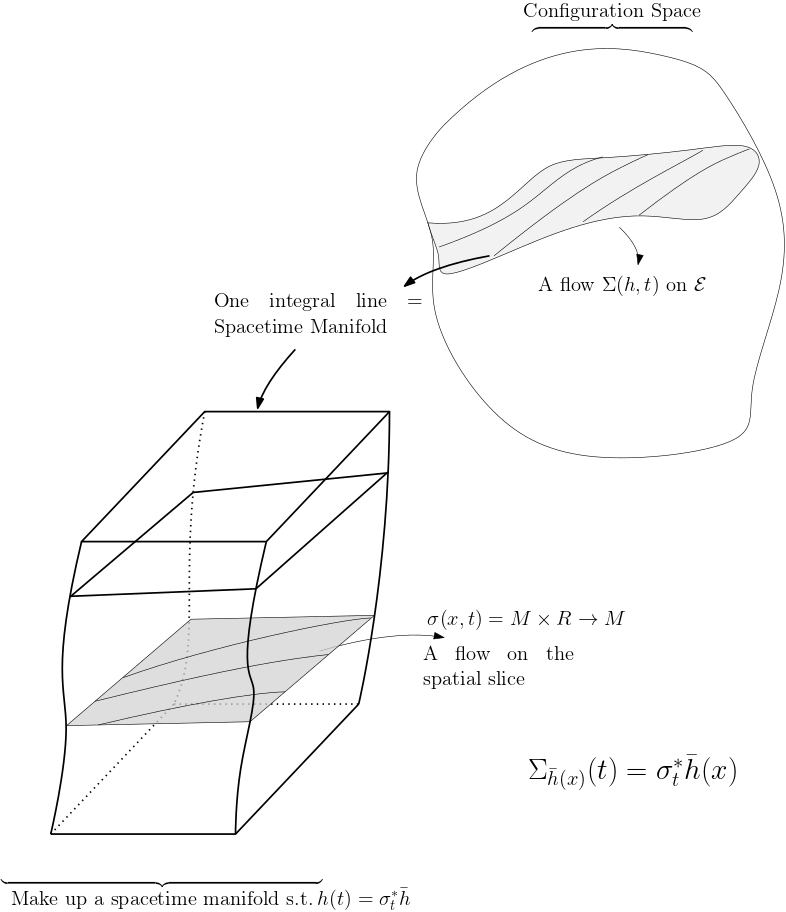}
\caption{An adiabatic spacetime as a flow on $\minsp$}
\label{fig:adsl}
\centering
\end{figure}
After writing down the desired metric we move onto studying the constraint equations and the Lagrangian for this subset of spacetimes. For this we will need to evaluate
\begin{equation}
K_{ij}= \frac{1}{2} \dot{\vacspmet}_{ij} = \frac{1}{2} \frac{d}{dt} \sigma_t^* \bar{h} = \frac{1}{2} \lie_{\chi} \vacspmet
\end{equation}
where
\begin{equation}
\chi= \sigma'_t \sgc
\end{equation}
i.e. $\chi$ is the vector field tangent to flow $\sigma$. So we write
\begin{equation}
\ext_{ij}= \consubm_{(i} \chi_{j)}
\end{equation}
where $\consubm$ is the Riemannian connection of $\vacspmet$. Thus $\chi$ spatial vector fields will fully describe our solutions, and we will need to solve our (vacuum) Einstein equations for $\chi$.
To illustrate the type of solutions we are considering and that they are not physically equivalent to the reference spacetime let us consider the following example.

\begin{example}[spatially flat \frw~from gauge transformations]
Let us consider the flat reference metric in four dimensions and let the transformation $x \ra \tilde{x}$ be given by
\begin{equation}
\tilde{x}^i= a(t) x^i
\end{equation}
at each slice so that 
\begin{equation}
d\tilde{x}^i= a(t) dx^i \sgd
\end{equation}
For these transformations the Manton ansatz will be a spacetime metric:
\begin{equation}
ds^2= -dt^2+ a^2(t) (dx^2+dy^2+dz^2)
\end{equation}
This illustrates very clearly that the Manton ansatz gives us physically new solutions, for this case the spatially flat \frw~spacetime. (This a solution to Einstein equations with matter, but our point is still true.) Note that this would be a coordinate transformation only if we took
\begin{equation}
(t,x,y,z)=(\tilde{t}, e^{-H \tilde{t}} \tilde{x},e^{-H \tilde{t}} \tilde{y},e^{-H \tilde{t}} \tilde{z} )\sgc
\end{equation}
so that
\begin{align}
d\tilde{t} &=dt \sgc \\
dx &= e^{-H \tilde{t} } d\tilde{x} - H e^{-H \tilde{t} } \tilde{x} d\tilde{t} \sgd
\end{align}
\end{example}
\section{Constraints and Action}\label{sec:constaction}

For the solutions we constructed above, our system will be described by the constraint equations and the action in \gnc~coordinates as follows. Momentum constraint becomes
\begin{equation}
\mom^j=\consubm_i \left( \consubm^{(i} \chi^{j)} - h^{ij} \consubm_k \chi^{k} \right)=0 \sgd
\end{equation}
Here note that this equation can also be written as
\begin{equation}
 \consubm^i \consubm_{[i} \chi_{j]} + R_{ij} \chi^i =0 \sgc
\end{equation}
where using the fact $\ric(\vacspmet)=0$ becomes
\begin{equation}\label{eq:momconst}
 \mom_j(\chi) \equiv \consubm^i \consubm_{[i} \chi_{j]} =0 \sgd
\end{equation}
Similarly we see Hamiltonian equation becomes
\begin{equation}\label{eq:hamconst}
\ham(\chi) \equiv (\consubm_i \chi^i)^2-\consubm_{(i} \chi_{j)} \consubm^i\chi^j=0 \sgd
\end{equation}

Now we move onto evaluating the action. Note that as we are on $\minsp$, potential will be zero. Action is
\begin{equation*}
S_{EH}=\int dt  \frac{1}{2} \metonvacua_{\vacspmet}(\delta_{\chi} \vacspmet,\delta_{\chi} \vacspmet) \sgc
\end{equation*}
where
\begin{equation*}
\metonvacua_{\vacspmet}(\delta_{\chi} \vacspmet,\delta_{\xi} \vacspmet) =  2 \int_{\vacspmet} d^dx \sqrt{\det \vacspmet}  \left( \consubm^{(i} \chi^{j)} \consubm_{(i} \xi_{j)} - \consubm \cdot \chi \consubm \cdot \xi \right) \sgd
\end{equation*}
Here note that solutions to the action above will be geodesics on $\minsp$ with respect to the metric $\metonvacua$ and \uline{imposing the Hamiltonian Constraint will mean choosing null $\metonvacua$ geodesics among these solutions.} One can easily see that the above metric can be brought into the form
\begin{equation}
\metonvacua_{\vacspmet}(\delta_{\chi} \vacspmet,\delta_{\xi} \vacspmet) = 2 \int_M d^dx \sqrt{\det \vacspmet} \, \consubm_i \lp \lp \consubm^{(i} \chi^{j)} - \vacspmet^{ij} \consubm \cdot \chi \rp \xi_j \rp - 2 \int d^dx \sqrt{\det \vacspmet} \, \xi^j \mom_j(\xi) \sgd
\end{equation}
Using the fact that we want our solutions to satisfy the momentum constraint, the metric on the space of vacua becomes
\begin{equation}
\metonvacua_{\vacspmet}(\delta_{\chi} \vacspmet,\delta_{\xi} \vacspmet) = 2 \int_{\bnd} d^dx \sqrt{\det \metbnd} \, n_i \lp \lp \consubm^{(i} \chi^{j)} - h^{ij} \consubm \cdot \chi \rp \xi_j \rp \sgc
\end{equation}
where $n$ is the vector field orthogonal to the boundary $\bnd$ of the spatial manifold $\sm$, and $\metbnd$ is the metric induced on $\bnd$. Note that it is not obvious here but $\metonvacua$ is symmetric in its arguments by definition. To further simplify the metric we make a second orthogonal decomposition: split everything into components in the direction of $n$ and orthogonal to it. We simply use the splitting formula we derived in Section \ref{sec:splitting}. i.e. We let
\begin{equation}
\chi = \chip n + \chit \sgc
\end{equation}
where 
\begin{equation}
\metbnd( n , \chit)=0 \sgd
\end{equation}
Now we impose an extra condition on our solutions: we want them to be \hypertarget{bnd prev}{boundary preserving} i.e.
\begin{equation}
\left. \chip \right|_{\bnd}=0 \sgc
\end{equation}
so that the boundary is fixed. This causes the term with divergence to drop. With this split we get 
\begin{equation}
\metonvacua_h(\delta_{\chi} h,\delta_{\xi} h) = \frac{1}{2}\int_{\bnd} d^dx \sqrt{\det \metbnd} \lp \metsubm( \chit , \consubm_n \xi ) + \vacspmet (n , \consubm_{\chit} \xi) \rp + \lp \chi \leftrightarrow \xi \rp \sgd
\end{equation}
Using the fact
\begin{equation}
\left. \xip \right|_{\bnd}=0 \cg \left. \chit(\xip) \right|_{\bnd}=0 \cg h(n,\chit)=0\sgc
\end{equation}
we get
\begin{equation}
\metonvacua_h(\delta_{\chi} h,\delta_{\xi} h) = \frac{1}{2}\int_{\bnd} d^dx \sqrt{\det \metbnd} \lp \vacspmet( \chit , \consubm_n \xit ) + \vacspmet (n , \consubm_{\chit} \xit) \rp + \lp \chi \leftrightarrow \xi \rp \sgd
\end{equation}
Then we define the splitting parameters as before:
\begin{align}
\consubm_n n & = \accbnd^b \eb \sgc\\
\consubm_a n &= \extbnd_a^b \eb \sgc \\
\consubm_n \ea &= - \accbnd_a n + {\twoconnbnd_a}^b \eb  \sgc\\
\consubm_a \eb &= \conbnd_a \eb - \extbnd_{ab} n \sgc
\end{align}
and note
\begin{align}
 \vacspmet( \chit , \consubm_n \xit )= \metbnd_{ab} n({\xit}^a) {\chit}^b + k_{ab} {\twoconnbnd_c}^a {\xit}^c {\chit}^b = \chit_b \ortder {\xit}^b \sgc
\end{align}
where we define the ``orthogonal derivative" as
\begin{equation}
\ortder \zeta^a = n^i \pr_i \zeta^a + {\twoconnbnd_b}^a \zeta^b \sgd
\end{equation}
Also using the definition of extrinsic curvature on the boundary we get
\begin{equation}\label{eq:vacmet1}
\metonvacua_h(\delta_{\chi} h,\delta_{\xi} h) = \int_{\bnd} d^dx \sqrt{\det \metbnd} \lp \chi_a \ortder \xi^a - \extbnd_{ab} \xi^a \chi^b  \rp \sgc
\end{equation}
where now we drop the $\parallel$ demarcation. We see that the action is almost reduced to data on the boundary $\bnd$ of the spatial manifold. The exception is due to the existence of the orthogonal derivative $\ortder$, since it includes information on how vectors tangent to the boundary changes off the boundary.

Before completing this section we note that $\ortder$ is a ``covariant" and metric compatible operator. It satisfies the properties of a connection as given in Definition \ref{def:appconn} in the appendix with proper modification, i.e. $\ortder$ is a map $\vfs(\pr M) \ra \vfs(\pr M)$ such that:
\begin{enumerate}
\item It is linear in its argument over $\R$ i.e. $\ortder( a \zeta + b \xi ) = a \ortder(\zeta) + b \ortder(\xi)$.
\item Satisfies $\ortder(f \zeta)= f \ortder (\zeta) + \ortder(f) \zeta$.
\end{enumerate}
There will be no analogue of the last item as the vector field we are taking the derivative with is a unique normalized vector field, $n$. One can extend the definition of $\ortder$ to arbitrary tensor as
\begin{equation}
\ortder {T^{abc..}}_{ghe..}= n({T^{abc..}}_{ghe..}) + {\twoconnbnd_d}^a {T^{dbc..}}_{ghe..} + ... - {\twoconnbnd_g}^d {T^{abc..}}_{dhe..} - ... \, \sgd
\end{equation}
Under this definition $\ortder$ will be metric compatible i.e.
\begin{equation}
\ortder \metbnd_{ab} =0 \sgd
\end{equation}
\section{Action as Boundary Data} \label{sec:boundarydata}
In this section we show that the action of equation \eqref{eq:vacmet1} is completely made up of the data on the boundary, for bulk vector fields that satisfy the momentum constraint. To prove this we will use the technology of \hmf~decomposition, as momentum constraint can be written in the language of forms:
\begin{equation}
\nabla^i \nabla_{[i} \chi_{j]}=0 \quad \ra \quad \dad d \chi=0 \sgd
\end{equation}
To show that the action is composed of boundary data we claim the following:
\begin{claim}\label{cl:bvp}
Given a boundary one-form $\ta \zeta$ on the boundary of a manifold with boundary and \uline{with trivial homology} there exists a one-form on the manifold that satisfies:
\begin{equation} \label{eq:bvp}
\dad d \chi=0 \cg \no \chi=0 \cg \ta \chi=\ta \zeta \sgc
\end{equation}
and unique up to exact forms that vanish on the boundary.
\end{claim} 
The gap in the claim above is that the undetermined exact part will have a normal derivative. However we will also show that this normal derivative is also given by the boundary data, i.e. we will show that
\begin{equation}\label{eq:ortderexact}
\xi^i_e= h^{ij} \pr_j \lambda \quad \mbox{and} \quad \no \xi^i_e=0 \quad \ra \quad \left. \lp \ortder \xi_e^{\parallel a} \rp \right|_{\pr M}=- \extbnd^{ab} \consubm_{e_b} \lambda \sgd
\end{equation}
First let us first prove this claim. Note that
\begin{align}
\no \xi^i_e=0 \quad \ra \quad \left. \consubm_n \lambda \right|_{\bnd} = 0 \quad \ra \quad \left. \consubm_{e_a} \consubm_n \lambda \right|_{\bnd}=0 \sgd
\end{align}
$\ortder \xi_e^a$ is calculated as
\begin{align}
\ortder \xi_e^a &= \consubm_n \lp k^{ab} \consubm_{e_b} \lambda \rp  + {\twoconnbnd_b}^a \metbnd^{bc} \consubm_{e_c} \lambda \sgc \\
&= \consubm_n k^{ab} \consubm_{e_b} \lambda + k^{ab} \lp \lb \consubm_n , \consubm_{e_b} \rb \lambda + \consubm_{e_b} \consubm_n \lambda \rp + {\twoconnbnd_b}^a \metbnd^{bc} \consubm_{e_c} \lambda \sgd
\end{align}
Using the facts
\begin{align}
\consubm_n k^{ab}= -2 \twoconnbnd^{(ab)} \cg \lb \consubm_n , \consubm_{e_b} \rb= - \accbnd_b n + \lp {\twoconnbnd_b}^c - \extbnd_b^c \rp e_c \sgc
\end{align}
one gets
\begin{equation}
\ortder \xi_e^a = - \extbnd^{ab} \consubm_{e_b} \lambda \sgd
\end{equation}

Now we move onto proving Claim \ref{cl:bvp}:\noclub[4]
\begin{proof}
Let $\chi$ be a one form on $\subm$. By \hmf~decomposition we can write it as
\begin{equation}
\chi= d \phi + \psi \gwg \dad \psi=0 \gag \no \psi=0
\end{equation}
Furthermore by Friedrichs decomposition, since we are on a manifold with trivial homology we know
\begin{equation}
\psi=\dad \beta \sgd
\end{equation}
Our problem then reduces to showing that there exists a $\chi$ unique up to exact forms such that
\begin{equation}
\dad d \dad \beta=0 \cg \no \beta=0 \cg \no d\phi=0 \cg \ta d \phi + \ta \dad \beta= \ta \zeta
\end{equation}
for a given $\ta \zeta$ on the boundary. Now let $\theta= d \dad \beta$. Then $\theta$ should satisfy
\begin{equation}
\dad \theta=0 \cg d \theta=0 \gag \ta \theta= \ta d \zeta \sgc
\end{equation}
so $\theta$ should be a harmonic form in $\subm$ with given tangential boundary condition. Since we have trivial homology and $\harmk_D(M)=0$ this $\theta$ will be unique, the existence is guaranteed by the fact that the boundary form that $\theta$ is supposed to be equal to is a closed one form, see theorem \ref{thm:dir}. Now let $\eta=\dad \beta$, then it should satisfy
\begin{equation}
d \eta= \theta \cg \dad \eta=0 \gag \no \eta=0 \sgd
\end{equation}
Again since there exists no Neumann field, such $\eta$ will be unique and existence is again guaranteed by theorem \ref{thm:dir}. Thus we have proven for a $\chi$ satisfying the conditions \eqref{eq:bvp}, the coexact part $\dad \beta$ is unique. Now we move on to the exact part. $\phi$ should be a function that satisfy 
\begin{equation}
\no d \phi=0 \gag \ta d \phi = \ta \delta \beta - \ta \zeta \sgd
\end{equation}
Such $\phi$ exists by the following theorem:
\begin{thm}[Thm 3.3.3 \cite{Gunter}]
Consider a closed boundary form $\ta \psi$ on the boundary of a homologically trivial $M$. Then there exists a closed form $\omega$ such that $\no \omega=0$ and $\ta \omega=\ta \psi$.
\end{thm}
\noindent
But this is not unique, since one can add any $d\alpha$ such that $\left. d \alpha \right|_{\partial M}=0$ to $d \phi$. 
\end{proof}
Thus we have shown that for a given one form on $\bnd$ there exists a $\chi$ with zero normal component on $\bnd$ that is unique up to an exact one form that vanishes on the boundary. That such one form exists is guaranteed by the above theorem.

\paragraph{Summary:}Let us summarize our conclusions. We are looking for solutions to Einstein equations in \gnc~that are of the form
\begin{equation}
\vacspmet(t)= \sigma_t^* \refmet \sgc
\end{equation}
so that they are parametrized by flows $\sigma_t$, whose tangent vector field we denote by $\chi$. To have boundary unaltered we impose
\begin{equation}
\left. \chi^{\perp} \right|_{\bnd}=0 \sgd
\end{equation}
To solve Einstein equations, $\chi$ should satisfy the momentum constraint
\begin{equation}
D^{i} \partial_{[i} \chi_{j]}=0 \sgc
\end{equation}
and Hamiltonian constraint
\begin{equation}
(D_i\chi^i)^2-D_{(i}\chi_{j)} D^i\chi^j=0 \sgc
\end{equation}
and extremize
\begin{equation}
L = \frac{1}{2}\int_{\bnd} d^dx \sqrt{\det \metbnd} \lp \chi_a \ortder \chi^a - \extbnd_{ab} \chi^a \chi^b  \rp \sgd
\end{equation}
We have shown the existence of $\chi$ with vanishing normal on the boundary that solves the momentum constraint and for a given value at the boundary, uniqueness up to an exact form that vanish on the boundary. We have also shown that the above action on the boundary is only made up of the boundary data. 

Note that thus far we have not commented on the Hamiltonian constraint except that it forces solutions to the action on the boundary to be null geodesics instead of any geodesic. Given a solution to the momentum constraint, since there will be an arbitrary exact part, Hamiltonian equation will become an equation for this exact part; i.e. let the solution to the momentum constraint be
\begin{equation}
\chi= d \phi + \dad \beta \sgc
\end{equation}
where $\beta$ is completely fixed and $d \phi$ is completely arbitrary in bulk but fixed on the boundary given a boundary value for $\chi$. Then Hamiltonian equation can be written as
\begin{equation}
\lp \consubm^2 \phi \rp - \consubm_i \consubm_j \phi \consubm^i \consubm^j \phi - 2 \consubm^i \consubm_k \beta^{kj} \consubm_{(i} \consubm_{j)} \phi= \consubm_{\left( i \right|} \consubm^k \beta_{k \left| j \right)} \consubm^i \consubm_k \beta^{kj} \sgd
\end{equation}
This equation is to solved for $d \phi$ for given $\dad \beta$ and boundary conditions. Note that this equation is second order in Hessian of $\phi$. Such equations in the literature is known as k-Hessian equations, see \cite{Wang2009}- our case falling in the category of 2-Hessians. Leaving a more specific study to future, we conjecture that solutions $d\phi$ to this equation exist and are unique. Note that we have already completely specified the theory on the boundary $\bnd$ by the momentum constraint, thus the use of Hamiltonian constraint is to provide us the full solution $\chi$ in the bulk.

\section{Orthogonal Decomposition of Constraints} \label{sec:decconst}

Note that even though we have shown the existence and uniqueness of solutions to the momentum constraint up to exact forms that vanish on the boundary, we have not built the explicit solutions. Note that in proving the existence and uniqueness in the previous section, we have used the bulk Hodge decomposition but not orthogonal decomposition. In this section we perform this decomposition for the constraint equations in an attempt to solve them. We will start with arbitrary dimension and arbitrary Ricci flat bulk metric and step by step will simplify the equations by making assumptions, until they come to a form we can easily solve. All the other more general cases we cannot solve, we leave as an open problem.
\subsection{Momentum Constraint}
First we perform the orthogonal decomposition of the momentum constraint. One gets
\begin{align}
d\chi &= \left( \accbnd_b \chi^\perp + D^\perp \chi_b - \conbnd_b \chi^\perp + \extbnd_{ab} \chi^a \right) n \wedge e^b + \conbnd_{[a} \chi_{b]} e^a \wedge e^b \\
&= t_b  n \wedge e^b + \conbnd_{[a} \chi_{b]} e^a \wedge e^b
\end{align}
where we define $t_b$ for simplicity. Then
\begin{align}
\dad d \chi & = \frac{1}{2} \left( D^\perp t_b + \extbnd t_b - \extbnd_{ab} t^a \right) e^b - \frac{1}{2} \conbnd_b t^b n + \left( \conbnd^a \conbnd_{[a} \chi_{b]}- \conbnd_{[a} \chi_{b]} \accbnd^a \right) e^b \sgc
\end{align}
so that the momentum constraint in the split form is
\begin{align}
P_b(\chi) &=\frac{1}{2} \left( D^\perp t_b + \extbnd t_b - \extbnd_{ab} t^a \right) + \conbnd^a \conbnd_{[a} \chi_{b]}-  \conbnd_{[a} \chi_{b]} \accbnd^a = 0 \sgc \\
P^{\perp}(\chi) &=\conbnd_b t^b =0 \sgd
\end{align}
One should in addition take into consideration that the bulk Ricci tensor was set to zero. Also performing the orthogonal splitting for this case will give us
\begin{align}
& -n(\extbnd)-\extbnd^{ab} \extbnd_{ab} + (\conbnd_a \accbnd)^a - \accbnd^a \accbnd_a =0  \sgc \\
& \conbnd_b \extbnd^{ba}- \conbnd^a \extbnd =0 \sgc \\
& ^{(d-1)} R_{ab} - \extbnd \extbnd_{ab} + 2 {\twoconnbnd_a}^c \extbnd_{cb} - n(\extbnd_{ab}) + (\conbnd_a \accbnd)_b- \accbnd_a \accbnd_b =0 \sgd
\end{align}
As can be seen above equations get already complicated at this stage. As a first assumption we restrict ourselves to a simplified basis:
\begin{assumption}
Assume $\lb n,\eb \rb=0$ so that
\begin{equation}
\lb n, \eb \rb = - \accbnd_b n + \lp \twoconnbnd^b_a - \extbnd^b_a \rp \eb =0 \quad  \ra \accbnd_b=0 \gag \twoconnbnd^b_a = \extbnd^b_a  \sgd
\end{equation}
\end{assumption}
\noindent
With this assumption we get
\begin{align}
2 P^\perp (\chi) &=  \conbnd^2 \chip  - 2 \conbnd_b \lp \extbnd^b_a \chia \rp - \conbnd_a \lp \dot{\chi}^a \rp =0 \sgc \\
2 P^{a}(\chi) &= \ddot{\chi}^a + \conbnd^2 \chia + \lp \ext \delta^a_b + 2 \extbnd^a_b \rp \lp \dot{\chi}^b - \conbnd^b \chip \rp - \conbnd^a \dot{\chi}^{\perp} + \prescript{(d-1)}{}{R^a_b} \chib - \conbnd^a \conbnd_b \chib =0, 
\end{align}
where we define the dot derivative as
\begin{equation}
n^i \pr_i f = \dot{f} \sgd
\end{equation}
Now remember that our boundary $\bnd$ will be a closed surface and topologically equivalent to $S^{d-1}$, and no harmonic forms live on it. Since we have also chosen a foliation for the spatial manifold $M$, each integral manifold will also be topologically equivalent to $S^{d-1}$. Because of this at each slice one can decompose $\chit$ as
\begin{equation}
\chi^a= \conbnd^a f + \conbnd_c \omega^{ca} \sgd
\end{equation}
Using this decomposition equations take the form
\begin{align}\label{eq:bndhodge}
2P^\perp(\chi) &= \conbnd^2(\chip - \dot{f}) - 2 \extbnd^{b}_{c} \conbnd_{b} \conbnd_{a} \omega^{a c} - \conbnd_{a} \extbnd \conbnd_{c} \omega^{c a} \sgc \\
2 P^{c}(\chi) &= \conbnd^{c} \left( \ddot{f} - \dot{\chi}^\perp \right) - 2 \extbnd^{c b} \conbnd_{b} \left( \dot{f}- \chip \right) + \extbnd\conbnd^{c} \left( \dot{f}-\chip \right) \nonumber  \\ 
&   + \left( 2 \extbnd^{c}_{b} + \extbnd \delta^{c}_{b} \right) \left( \conbnd_{d} \dot{\omega}^{d b} + \conbnd_{d} \extbnd \omega^{d b} \right)   + 2 \conbnd_{a} \extbnd \dot{\omega}^{a c} + \conbnd_{a}\dot{\extbnd} \omega^{a c} + \conbnd_{d} \ddot{\omega}^{d c} \nonumber \\ &+ \conbnd^2 \conbnd_{d} \omega^{d c} + ^{(d-1)}{R}^{c}_{b} \conbnd_{d} \omega^{d b} \sgd
\end{align}
This splitted form suggest $\chi^{\perp}=\dot{f}$, with $\omega=0$ this would correspond to a bulk solution $\chi^i=\nabla^i f$, the arbitrary exact form we already knew would appear. By letting $\chi^{\perp}=\dot{f}$ anyways, one can get two equations for $\omega$, still nontrivial to solve.  To get some more simplification let us consider a special case:
\begin{case}
Let us consider the case where the metric on the boundary has the form:
\begin{equation}
\metbnd_{ab}(r,y)=\gamma(r) \breve{\metbnd}_{ab}(y^a)
\end{equation}
still in arbitrary dimension. For this case the extrinsic curvature and Ricci tensor will be proportional to the boundary metric i.e.:
\begin{equation}
\extbnd_{ab}= \frac{\dot{\gamma}}{\gamma} \metbnd_{ab} \cg ^{(d-1)}R_{ab}= \frac{(s-1)}{4} \frac{\dot{\gamma}^2}{\gamma^2} \metbnd_{ab} \sgd
\end{equation}
where $s=d-1$.Furthermore imposing the flatness equations enforce:
\begin{equation}
\gamma(r)=(r+c)^2
\end{equation} 
for some constant $c$.
\end{case}
\noindent
Under this assumption orthogonal part of the momentum constraint becomes
\begin{equation}
2P^\perp(\chi) =\conbnd^2 \left( \chi^\perp - \dot{f} \right) =0 \sgd
\end{equation}
Since the boundary is topologically trivial without boundary this mean $\chi^\perp=\dot{f}$. Then the tangential momentum constraint becomes
\begin{equation}
2 P^{c}(\chi) = \conbnd_{d} \ddot{\omega}^{\rho c} + \frac{2}{r+c} \left( 1 + \frac{s}{2} \right) \conbnd_{d} \dot{\omega}^{d c} + \conbnd^2 \conbnd_{a} \omega^{a c} + \frac{(s-1)}{(r+c)^2} \conbnd_{a} \omega^{a c} =0 \sgd
\end{equation}
\expl{simplify further by switching $D^2 D$ ?}Note that this assumption has allowed us to decouple the boundary exact and coexact parts, and reduced the problem to a differential equation for the boundary coexact part. 
Existence of the two form $\omega$ here complicates the matters. In place of the assumption above, starting from \eqref{eq:bndhodge}, if one focuses on $d=3$, since a two form will be hodge dual to a one form things will also simplify in a different way. Let us consider this case:
\begin{case}
Consider the case with coordinate basis and $d=3$ without any other assumption on the metric.
\end{case}
For this case since the boundary will be 2 dimensional, one can use the fact that any antisymmetric rank 2 tensor should be proportional to $ \volk^{a b} $,then $\omega_{ab}= \omega(y^a) {\volk}_{ab} $ and e.g.
\begin{equation}
\extbnd^{c}_{b} \volk^{b a} - \extbnd^{a}_{b} \volk^{b c} = \extbnd \volk^{c a} \sgd
\end{equation}
Then
\begin{align}
2P^\perp(\chi) &= \conbnd^2(\chi^\perp - \dot{f}) - 2 \extbnd^{b}_{c} \volk^{d c} \conbnd_{b} \conbnd_{d} \omega - \volk^{c d}  \conbnd_{d} \extbnd \conbnd_{c} \omega \sgc\\
2 P^{c}(\chi) &= \conbnd^{c} \left( \ddot{f} - \dot{\chi}^\perp \right) - 2 \extbnd^{c b} \conbnd_{b} \left( \dot{f}- \chi^\perp \right) + \extbnd \conbnd^{c} \left( \dot{f}-\chi^\perp \right) \nonumber  \\ 
&   + \left( 2 \extbnd^{c}_{b} - \extbnd \delta^{c}_{b} \right)\left(  \conbnd_{a}  \dot{\omega}  - \extbnd \conbnd_{a} \omega \right) \volk^{a b} + \volk^{a c} \conbnd_{a} \left( \ddot{\omega} + \conbnd^2 \omega \right) \sgd
\end{align}
Note that now boundary coexact parts simplify but we lose the boundary exact - coexact decoupling. An example \expl{and the only?} that satisfy both of these assumptions is actually the case of 3 dimensional ball. For this case we again have
\begin{equation}
2P^\perp(\chi) =\conbnd^2 \left( \chi^\perp - \dot{f} \right) =0 \sgc
\end{equation}
again meaning $\chi^\perp=\dot{f}$. Note that there are no other conditions on $\chi^\perp$ and $f$, and they actually correspond to the arbitrary exact part of the solution, $\chi^i= \metsubm^i f$ we have shown to exist. Tangential equation then simplifies significantly:
\begin{equation}
\vol_k^{ab} \conbnd_b \lp \ddot{\omega} + \conbnd^2 \omega \rp =0 \sgd
\end{equation}
This will be the the laplacian in 3-d flat space. Using spherical harmonics one can solve it as
\begin{equation}
\omega= a_{lm} r^{(l+1)} Y_{lm} \sgc
\end{equation}
resulting in a bulk vector
\begin{equation}
\chi^i= \vol_h^{ijk} x_j \pr_k \omega \sgd
\end{equation}
We will discuss this case of the ball in more detail later.

Other special case we consider is 2 dimensional spatial manifold (but no assumptions about the metric on $\bnd$).
\begin{case}
Let us now focus on the case $d=2$.The spatial metric will be of the form
\begin{equation}
ds^2=dr^2 + \gamma(r,\theta) d \theta^2 \sgd
\end{equation} 
\end{case}
\noindent
For this case the constraint equations become
\begin{align}
2 P^\perp(\tau) &= \frac{1}{\sqg} \pr_{\theta} \lpa \frac{\pr_{\theta} \tau^r}{\sqg} \rpa - \frac{1}{\sqg} \pr_r \lpa \frac{\pr_r \gamma}{\sqg} \rpa \tau^r - \frac{1}{\gamma} \pr_{\theta} \pr_r \lpa \gamma \tau^{\theta} \rpa + \frac{\pr_{\theta} \gamma}{2 \gamma^2} \pr_r \lpa \gamma \tau^{\theta} \rpa \sgc \\
2 P^{\theta}(\tau) &= \pr^2_r \tau^{\theta} + \frac{3}{2} \frac{\pr_r \gamma}{\gamma} \pr_r \tau^{\theta} + \frac{1}{2} \frac{\pr_r \gamma}{\gamma^2} \pr_{\theta} \tau^r - \frac{1}{\gamma} \pr_{\theta} \pr_r \tau^r \sgd
\end{align}
Now let us make a specific choice:
\begin{case}
Let $d=2$, $\gamma=r^2$ such that
\begin{equation}
ds^2=dr^2 + r^2 d \theta^2 \sgd
\end{equation}
\end{case}
\noindent
Then the equations simplify into
\begin{align}
2 P^r(\tau) &= \frac{1}{r^2} \pr_{\theta}^2 \tau^r - \pr_r \lpa r^2 \pr_{\theta} \tau^{\theta} \rpa \sgc \\
2 P^{\theta}(\tau) &= \frac{1}{r^3} \pr_r \lpa r^3 \pr_r \tau^{\theta} \rpa - \frac{1}{r} \pr_r \lpa \frac{\pr_{\theta} \tau^r}{r} \rpa \sgd
\end{align}
Then we perform a Fourier series expansion for both of the components. Solutions to this are vector fields
\begin{equation}
\tau= \frac{\pr}{\pr \theta} \cg \tau= \frac{1}{r^2} \frac{\pr}{\pr \theta} \cg \tau= f(r) \frac{\pr}{\pr r} \sgc
\end{equation}
together with the two infinite towers
\begin{align}
\tau &=\sum_{n=1} -\frac{\dot{g}_n(r)}{n} \cos(n \theta) \frac{\pr}{\pr r} + \frac{g_n(r)}{r^2}  \sin(n \theta)\frac{\pr}{\pr \theta} \sgc \\
\tau &= \sum_{n=1} \frac{\dot{f}_n(r)}{n} \sin(n \theta) \frac{\pr}{\pr r} +  \frac{f_n(r)}{r^2} \cos(n \theta) \frac{\pr}{\pr \theta} \sgd
\end{align}
Note that these resemble, but indeed different than the generators of 2d conformal transformations
\begin{align}
\tau &=\sum_{n=1} r^{n+1} \cos(n \theta) \frac{\pr}{\pr r} + r^n \sin(n \theta)\frac{\pr}{\pr \theta} \sgc \\
\tau &= \sum_{n=1} - r^{n+1} \sin(n \theta) \frac{\pr}{\pr r} +  r^n \cos(n \theta) \frac{\pr}{\pr \theta} \sgd
\end{align}

\subsection{Hamiltonian Constraint}
We also would like to decompose the Hamiltonian Constraint. Since this equation non-linear, to simplify the affair we assume right from the beginning the coordinate basis.
\begin{assumption}
Assume $\lb n,\eb \rb=0$ so that
\begin{equation}
\lb n, \eb \rb = - \accbnd_b n + \lp \twoconnbnd^b_a - \extbnd^b_a \rp \eb =0 \quad  \ra \accbnd_b=0 \gag \twoconnbnd^b_a = \extbnd^b_a \sgd
\end{equation}
\end{assumption}
\noindent
Then the Hamiltonian constraint becomes
\begin{align}
\ham &= \left( \extbnd_{a b} \extbnd^{a b} - \extbnd^2 \right) (\chi^\perp)^2 - 2 \extbnd \dot{\chi}^\perp \chi^\perp + \frac{1}{2} \conbnd_{a} \chi^\perp \conbnd^{a} \chi^\perp \nonumber \\ 
 &+ \conbnd_{a} \chi^\perp \dot{\chi}^{a} - 2 \left( \dot{\chi}^\perp + \extbnd \chi^\perp \right) \conbnd_{a} \chi^{a} + 2 \chi^\perp \conbnd_{a} \chi_{b} \extbnd^{a b} \nonumber \\ 
 & + \frac{1}{2} \gamma_{a b} \dot{\chi}^{a} \dot{\chi}^{b} + \frac{1}{2} \left( \conbnd_{b} \chi^{a} \conbnd_{a} \chi^{b} + \conbnd_{b} \chi_{a} \conbnd^{b} \chi^{a} \right) - \left( \conbnd_{a}\chi^{a} \right)^2
\end{align}
where we again define $n(f)=\dot{f}$. Let us focus on the case where $d=3$ again.
\begin{case}
Let $d=3$ so that on any of the leaf of the foliation, any vector field can be written as
\begin{equation}
\chi^a= \conbnd^a f + \vol_k^{ab} \conbnd_b \omega \sgd
\end{equation}
\end{case}
Let us also take the integration over an arbitrary closed surface within the spatial manifold that is topologically $S^{(d-1)}$. Note that if we take this surface to be the boundary and take the boundary preserving condition $\left. \chi^\perp \right|_{\bnd}=0$ we should get back the action we have derived. Nevertheless on an arbitrary leave of foliation we get \lattech{fix the alignment}
\begin{align}
0 &= \int_{\bnd} \vol_{\metbnd}  \Bigl(  - \breve{R} \dot{f}^2 + 2 \dot{f} \lp -  \extbnd \ddot{f}  -  \conbnd^2 \dot{f} + 3 \extbnd^{a b} \conbnd_{a} \conbnd_{b} f + 2  \conbnd_{a} \extbnd \conbnd^{a} f \rp \nonumber \\  & \quad \quad \quad - 2 (\ddot{f} + \extbnd \dot{f} ) \conbnd^2 f  + \lp 2 \extbnd^{c a} \extbnd_{a}^{d}  - \breve{R}^{cd} \rp \conbnd^{c}f \conbnd^{d}f \Bigr)  \nonumber \\
 &+ \int_{\bnd} \vol_{\metbnd}   \Bigl( 2 \dot{f} \conbnd_{a} \extbnd \conbnd_{d} \omega \epsilon^{d a} + 2 \dot{f} \extbnd^{a b} \conbnd_{b} \conbnd^{c} \omega \epsilon_{c a}  + 2 \extbnd_{a}^{c} \conbnd_{c}f \left( \extbnd \conbnd_{d} \omega - \conbnd_{d} \dot{\omega} \right) \epsilon^{d a} \nonumber \\ & \quad \quad \quad  - 2 \breve{R}_{c d}  \epsilon^{c a} \conbnd^{d} \omega \conbnd_{a}f \Bigr) \nonumber \\
  &+ \int_{\bnd} \vol_{\metbnd}  \Bigl( \frac{1}{2} \extbnd^2 \conbnd_{a} \omega \conbnd^{a} \omega - \extbnd \conbnd^{a} \omega \conbnd_{a} \dot{\omega} + \frac{1}{2} \conbnd^{a} \dot{\omega} \conbnd_{a} \dot{\omega} + \frac{1}{2} \conbnd^2 \omega \conbnd^2 \omega \nonumber \\ & \quad \quad \quad  - \breve{R}_{a b} \conbnd^{a} \omega \conbnd^{b} \omega \Bigr) \sgd
\end{align}
This still rather complicated to solve. We now focus on the case of a ball again.
\begin{case}
Let us further assume that each leaf of foliation is a round 2-sphere with with radius $r$.
\end{case}
\noindent
We get
\begin{align}
0 & = \int d \Omega_r \left( -\frac{2}{r^2} \dot{f}^2 - \frac{4}{r} \ddot{f} \dot{f} - 2 \dot{f} \conbnd^2 \dot{f} + - 2 \ddot{f} \conbnd^2 f + \frac{2}{r} \dot{f} \conbnd^2 f - \frac{1}{r^2} f \conbnd^2 f \right) \\ 
& + \int d\Omega_r \left( -\frac{1}{r^2} \omega \conbnd^2 \omega + \frac{2}{r} \dot{\omega} \conbnd^2 \omega - \frac{1}{2} \dot{\omega} \conbnd^2 \dot{\omega} + \frac{1}{2} \conbnd^2 \omega \conbnd^2 \omega \right)=0 \sgd
\end{align}
Now plug in the solution to $\omega$ coming from the momentum constraint and also expand $f$ in terms of the spherical harmonics, i.e. let 
\begin{align}
\omega &= Y_{lm} c_{lm} r^{(l+1)} \sgc \\
f &= Y_{lm} f_{lm}(r) \sgd
\end{align}
Noting $d \Omega_r= d \Omega_{unit} r^2$ we get
\begin{equation}\label{eq:Hamcons1}
0= \sum_{lm}\frac{d}{dr}  \left( -2 \dot{f}_{lm}^2 r  + l (l+1) \left( 2 \dot{f}_{lm} f_{lm} - \frac{f_{lm}^2}{r} + \frac{1}{2} c_{lm}^2 r^{2l+1} (l-1) \right) \right) \sgd
\end{equation}
Now taking the r integral from 0 to R , should give us back the metric form.  Using $\dot{f}_{lm}(R)=0$ this can easily be seen if we assume
\begin{equation}
\lim_{r \rightarrow 0} \frac{f_{lm}^2(r)}{r}=0 \sgd
\end{equation}
Assume now for simplicity that each term in \eqref{eq:Hamcons1} vanishes identically, and try to see if there exists a solution that satisfies this. Defining $g_{lm}(r)= f_{lm}(r)/ r^{(l+1)}$ equation becomes
\begin{equation}
-2r^2 \dot{g}_{lm}^2 + 2 (l+1)(l-2) \dot{g}_{lm} g_{lm} r + (l+1)(2l^2-l-2) g_{lm}^2 = -l\frac{(l^2-1)}{2} c_{lm}^2 
\end{equation}
with the boundary condition $g_{lm}(R)= \pm \sqrt{\frac{l-1}{2}} c_{lm}$. First few solutions look like this:
\begin{align}
g_{00} &= a_{00} \frac{R}{r} \sgc \\
g_{1m} &= a_{1m} \frac{(2R \pm r)}{r} \sgc \\
g_{2m} &= \mp \frac{a_{2m}}{2( \pm 2 + \sqrt{6})} \left( \left( \frac{r}{R}\right) ^{\sqrt{6}} \left( \pm 5 + 2 \sqrt{6} \right) \mp \left( \frac{R}{r} \right)^{\sqrt{6}} \right) \sgd
\end{align}
However one can check these do not solve the full Hamiltonian equation, thus our assumption fails. We conclude our investigation here. Later on we will take up the case of the 3-ball, since this was the only case we could explicitly solve the momentum constraint and discuss the form of the action. Note that our failure to solve the Hamiltonian equation only causes us not to have the explicit solution in the bulk, but since as we have proved that the action is completely determined by the data on the spatial boundary we can still comment on it.

%% file: slnspace.tex
\chapter{Structure of The Space of Vacua}

In the previous chapter we have investigated the adiabatic solutions to Einstein equations by writing down the equations for a metric obtained by the Manton approximation and bringing them into the a form composed of the constraints together with a reduced action. We showed that the momentum constraint gives a one-to-one relationship between vector fields on the boundary of the spatial slice up to exact forms which we conjectured to be uniquely solved in terms of the coexact part of the solution by the Hamiltonian constraint, and showed what the reduced action is and proved that it is solely composed of the boundary data. As we have pointed out, this action defines a metric on the space of vacua, and the solutions are to be geodesics on the space of vacua with respect to this metric, moreover Hamiltonian constraint forces these solutions to be null geodesics.

In this chapter we investigate the structure of the space of vacua via the information we have gathered. First in Section \ref{sec:wdw} we compare our results with results obtained in the literature on the Wheeler-deWitt (\wdw) metric, more specifically we discuss the signature of the metric on the space of vacua. In the Section \ref{sec:homSpofGR} we discuss the space of vacua in terms of the action of diffeomorphisms on it, and identify it as a (pseudo-)Riemannian homogeneous space.

\section{Signature of the Metric on the Space of Vacua}\label{sec:wdw}
\noindent
Metric on the configuration space $\consp$ we have written down in \eqref{eq:wdw} as
\begin{equation}
\metonvacua_h(\delta_1 h,\delta_2 h) = \int_M d^dx \sqrt{\det h} \lp h^{i[k} h^{j]l} \rp \delta_1 h_{ij} \delta_2 h_{kl} \sgc
\end{equation}
is known as the Wheeler-deWitt (\wdw) metric in the literature \cite{giulini}. Note that this metric is not positive definite, indeed by York's decomposition in \cite{york}, one can separate the tangent space into orthogonal parts with respect to this basis via
\begin{equation}
\delta h_{ij}= \delta h_{ij}^{TT} + \lp \consubm_i \xi_j + \consubm_j \xi_i -  \frac{2}{d} h_{ij} \consubm_i \xi^i \rp + h_{ij} \Omega
\end{equation}
where 
\begin{equation}
\consubm^i \delta h_{ij}^{TT}= 0 \gag h^{ij} \delta h_{ij}^{TT}=0 \sgd
\end{equation}
\wdw~metric is positive on first two parts and negative on the third part. By definition we are interested in the part of the tangent space written as
\begin{equation}\label{eq:diffeo}
\consubm_i \xi_j + \consubm_j \xi_i \sgd
\end{equation}
This will be a mixture of second and third parts and thus its norm will change in accordance with how much of it projected to which part of the tangent space \cite{giulini}.

However the form of \eqref{eq:diffeo} is not the only restriction we have on our tangent space. Remember the reduced form of the action we had in \eqref{eq:vacmet1}:
\begin{equation}
\metonvacua_h(\delta_{\chi} h,\delta_{\xi} h) = \frac{1}{2}\int_{\bnd} d^dx \sqrt{\det \metbnd} \lp \chi_a \ortder \xi^a - \extbnd_{ab} \xi^a \chi^b  \rp \sgd
\end{equation}
Let us think this metric now as an h-dependent inner product on the space of boundary vector fields:
\begin{equation}
\ket \chi,\xi \bra_h = \frac{1}{2}\int_{\bnd} d^dx \sqrt{\det \metbnd} \lp \chi_a \ortder \xi^a - \extbnd_{ab} \xi^a \chi^b  \rp \sgd
\end{equation}
One can check that
\begin{equation}
\ket \chi,\chi \bra_h = \begin{cases}
             \geq 0  & \text{if } \chi=df \sgc\\
             \leq 0  & \text{if } \chi=\dad \beta \quad \mbox{where} \, \beta \in \Lambda^{(d-2)}(\bnd) \sgd
       \end{cases}
\end{equation}
However these parts will not be orthogonal with respect to each other in general. Let us find the orthogonal complement of the exact part to illustrate this, using equation \eqref{eq:ortderexact}:
\begin{equation}
\ket df, \xi \bra=  \int_{\bnd} d^d x \sqrt{\det \metbnd} \lp - \extbnd_{ab} D^a f \xi^b \rp=0
\end{equation}
So vector fields that are orthogonal to exact parts should satisfy $\conbnd_a \lp \extbnd^a_b \xi^b \rp=0$, which will not be in general the case for an coexact part. In the next chapter however we will see that they will for the case of $M=\bar{B}_R^3$.
\section{Homogeneous Space Structure}\label{sec:homSpofGR}
Remember that when we first wrote down the Manton approximation in equation \eqref{eq:Manapp}, we started by the proposed solutions of the form
\begin{equation} \label{eq:Manapp2}
\vacspmet(t)= \sigma_t^* \refmet \sgc
\end{equation}
where $\sigma(t,x)$ is a flow on a spatial slice $\subm$ with the metric $\metsubm$, and $\refmet$ is the spatial part of a given solution of the Einstein equations in \gnc. Thus the space of metrics in the form of \eqref{eq:Manapp2} can be considered as the orbit of $\refmet$ under the \hyperlink{group action}{group action} $\cdot: \bdiffm \times \subm \ra \subm$ of \hypertarget{bnd prev}{boundary preserving} diffeomorphisms $\bdiff(\subm)$ on the configuration space $\consp$ given as (see Appendix \ref{subsec:appGrAct} for definitions related to group actions)
\begin{equation}
\phi \cdot \metsubm= \phi^* \metsubm \sgc
\end{equation}
and note that condition of boundary preservingness is needed since we have a manifold with boundary and would like to keep this boundary fixed. Without going into too much mathematical detail, we take this orbit to be a properly embedded submanifold \expl{elaborate?}. Note that boundary preserving diffeomorphisms do form a group, which we will call $\bdiff(\subm)$. We will show this at the Lie algebra level: we show that vector fields that are boundary preserving, i.e. that have a vanishing normal component on the boundary, form a closed Lie algebra; i.e.
\begin{equation}
\left. \lb \chi, \zeta \rb^\perp \right|_{\bnd}=0 \gig \chi,\zeta \in \vfs(m) \gstg \left. \chi^\perp \right|_{\bnd}=0 \gag \left. \zeta^\perp \right|_{\bnd}=0 \sgd
\end{equation} 
\begin{proof}
We need to see whether
\begin{equation}
\left. \lb \chi, \zeta \rb^\perp \right|_{\bnd}= \left. n_j \lp \chi^i \pr_i \zeta^j - \zeta^i \pr_i \chi^j \rp \right|_{\bnd}
\end{equation}
is equal to zero. Using $\left. \chi^\perp \right|_{\bnd}=0 \gag \left. \zeta^\perp \right|_{\bnd}=0$ one can show that
\begin{equation}
\left. \lb \chi, \zeta \rb^\perp \right|_{\bnd} = \left. \chi^a \zeta^b n_j \lp e^i_a \pr_i e^j_b - e^i_b \pr_i e^j_a \rp \right|_{\bnd} = \left. \chi^a \zeta^b n \lp \lb e_a, e_b \rb \rp \right|_{\bnd} \sgc
\end{equation}
which is equal to zero due to involutivity of the distribution we are using to foliate our spatial manifold $\subm$.
\end{proof}
Now remember that not any metric in the form of \eqref{eq:Manapp2} can solve Einstein equations, as a first step we have seen one needs to impose constraint equations on \eqref{eq:Manapp2}. Let us denote the set of boundary preserving diffeomorphisms that such that \eqref{eq:Manapp2} satisfy the constraint equations as $\bdiffc(\subm)$. The actual space of solutions should actually be the orbit of $\refmet$ under the action of this ``set". However for this to be meaningful one first needs to show that $\bdiffc(\subm)$ is actually a group. 

We do this via the following set of arguments: Let us denote the set of boundary preserving diffeomorphisms that are identity on the boundary as $\bdiff_0(M)$. Note that by definition $\bdiff_0(M)=\diff_0(M)$. This is a normal subgroup in $\bdiffm$. This can be seen as follows: Vector fields corresponding to $\bdiff_0(M)$ are boundary preserving vector fields that vanish on the boundary. Let us denote these by $\abdiffz$. These form an ideal in the set of boundary preserving vector fields, which we will denote as $\abdiff$. For this one should only check if
\begin{equation}
\left. \lb \chi, \xi \rb \right|_{\bnd} = 0 \quad \forall \chi \in \abdiff \gag \forall \xi \in \abdiffz \sgc
\end{equation} 
which can be easily done by again using the orthogonal decomposition on $\subm$. Here we note that this argument will not work for $\diffm$, i.e. $\diff_0(M)$ is not a normal subgroup of $\diffm$.

Since $\bdiff_0(\subm)$ is a normal subgroup of $\bdiff(\subm)$, the quotient $\bdiff(\subm)/\bdiff_0(\subm)$ is a Lie group. Note that this is the space of equivalence classes of boundary preserving diffeomorphisms that are related by a diffeomorphism that is identity on the boundary, and thus equivalent to $\bdiffc(\subm)$. The reason for this is that we have partially proved and partially conjectured in the preceding chapter that for a given boundary vector field, there exists a unique vector field that solves the constraint equations. Note also that $\bdiffc(\subm) \sim \diff(\bnd)$.

So as a result we identify our space of vacua as the orbit of $\refmet$ under the action of $\bdiffc(\subm)$. Then this action is transitive on the orbit of the $\refmet$ which we now declare to be the true space of vacua $\truevac$, so that $\truevac$ is a homogeneous space. Because of this, by Proposition \ref{prop:homsp}
\begin{equation}
\truevac \sim \bdiffc(\subm) / \bdiffc(\subm)_{\refmet} \sim \diff(\bnd) / \bdiffc(\subm)_{\refmet}
\end{equation}
\expl{does this fail for infinite dimensional groups such as our case, see pg 65 of Montogomery}\expl{also make comment that $\bdiffc(\subm)$ is equivalent to boundary diffeo}where $\bdiffc(\subm)_{\refmet}$ is the \hyperlink{isot}{isotropy group} of $\refmet$. By definition these should be the isometries of $\refmet$ that are in $\bdiffc(\subm)$. One can see indeed that all of the isometries of $\refmet$ is in $\bdiffc(\subm)$ since
\begin{equation}
\consubm_{(i} \xi_{j)}=0 \ra \momi(\xi)=0 \gag \ham(\xi)=0 \sgd
\end{equation}
Remember that the reduced action we have obtained is actually a metric on the space of vacua $\truevac$. Now we move onto check whether $\truevac$ is a Riemannian homogeneous space as was the case for Yang-Mills. For this we will use the equivalence of conditions of Proposition \ref{prop:RieHom}. i.e. We will show that the scalar product we have defined is $\Ad^{\bdiffc(\subm) / \bdiffc(\subm)_{\refmet}}$ invariant. For this we retreat back to a bulk form of the scalar product:
\begin{equation}
\ket \zeta, \lambda \bra = \int_{\bnd} n_i \nabla^{(i} \xi_\zeta^{j)} \xi^{\lambda}_j \sgc
\end{equation}
where $\chi^{\zeta}$ should be considered as the solution to constraint equations given a boundary vector field $\zeta$. Then showing  $\Ad^{\bdiffc(\subm) / \bdiffc(\subm)_{\refmet}}$ invariance is equivalent to showing the invariance
\begin{equation}
\ket \zeta, \lambda \bra = \ket \xi_\zeta, \xi_\lambda \bra= \ket \phi_* \xi_\zeta, \phi_* \xi_\lambda \bra \sgc
\end{equation}
where $\phi$ is an isometry that is boundary preserving. We show this at the infinitesimal level, i.e. let
\begin{equation}
\ket \xi_\zeta, \xi_\lambda \bra = \int_{\bnd} \lp \lie_{\xi_\zeta} \metsubm \rp \lp n, \xi_{\lambda} \rp \sgd
\end{equation}
Then we need to show
\begin{equation}
\delta_\chi \ket \xi_\zeta, \xi_\lambda \bra = \ket \lie_\chi \xi_\zeta , \xi_\lambda \bra + \ket \xi_{\zeta}, \lie_\chi \xi_\lambda \bra=0 \sgc
\end{equation}
where $\lie_\chi \metsubm=0$ and $\left. \chi^\perp \right|_{\bnd}=0$. But we have
\begin{align}
\delta_\chi \ket \xi_\zeta, \xi_\lambda \bra &= \int_{\bnd} \lie_{\lie_{\xi_\zeta}} \metsubm (n, \xi_\lambda) + \lie_{\xi_\zeta} \metsubm \lp n, \lie_\chi \xi_\lambda \rp \nonumber \\
&= \int_{\bnd} \lie_\chi \lp \lie_{\xi_\zeta} \metsubm (n, \xi_\lambda) \rp - (\lie_{\xi_\zeta} \metsubm) \lp \lie_\chi n, \xi_\lambda \rp \nonumber \\
&= -\int_{\bnd} \conbnd_a \chi^a \lie_{\xi_\zeta} \metsubm (n, \xi_\lambda) - ( \lie_{\xi_\zeta} \metsubm ) \lp \lie_\chi n, \xi_\lambda \rp \sgd
\end{align}
Now we use the properties of a bulk Killing vector under orthogonal decomposition:
\begin{align}\label{eq:bulkkill}
\consubm_{(i} \chi_{j)}= 0 \ra & \ortder \chi^\perp = \accbnd_a \chi^a \sgc \\
& \acc_a \chi^\perp + \conbnd_a \chi^\perp + \ortder \chi_a - \extbnd_{ab} \chi^b =0  \sgc\\
& \conbnd_{(a} \chi_{b)} + \extbnd_{ab} \chi^\perp=0 \sgc
\end{align}
and also
\begin{align}
\lie_\chi n= \lp \accbnd_a \chi^a - \ortder(\chi^\perp) \rp n - \lp \ortder \chi^b - \extbnd_a^b \chi^a \rp e_b  \sgd
\end{align}
So we get
\begin{equation}
\lie_\chi n=0 \gag \left. \conbnd_a \chi^a \right|_{\bnd}=0 \ra  \delta_\chi \ket \xi_\zeta, \xi_\lambda \bra =0 \sgd
\end{equation}
Thus we conclude that $\truevac$ is a Riemannian homogeneous space. Here we note one subtlety however. Note that the scalar product $\ket,\bra$ is identically zero for vector fields that satisfy
\begin{equation*}
\ortder \zeta^a - \extbnd_b^a \zeta^b= \mathbb{D} \zeta= 0 \sgd
\end{equation*}
From \eqref{eq:bulkkill} we see that the isotropy group is necessarily is in the kernel of $\mathbb{D}$ but not vice verse. If the kernel of $\mathbb{D}$ turns out to be a bigger set then the isotropy group then the metric on $\truevac$ will be degenerate. We will leave further discussion of this point to a future study. \expl{check if you can see this}

%% file: GRonBall.tex
\chapter{General Relativity on a Ball} \label{ch:ball}
In the previous chapter we have seen which equations the adiabatic solutions we have proposed needs to satisfy and what the action reduces into. However we have not been able to solve the constraint equations explicitly except the case where we take the spatial slice as the three dimensional ball of radius $R$ and take the flat Euclidean metric as the reference metric. For this case we were able to solve the momentum constraint but not the Hamiltonian. In this chapter we analyze this case in more detail.

Since for each spatial slice, spatial metric is diffeomorphic to Euclidean space, equations for this case can be written in flat coordinates, and usual 3-d vector analysis can be utilized. In fact in this special spacetime, many operators and quantities will simplify; e.g. operations of $\dad$ and $d$ will become divergence and curl, hodge decomposition will become Helmholtz theorem, and by choosing a coordinate system acceleration will vanish and two connection will be equal to extrinsic curvature, while extrinsic curvature itself will be proportional to the flat metric. This analysis will help us better understand the procedure we carried out and results we obtained by turning them into an easier to visualize 3-d vector analysis.
In the first section we discuss Vector Spherical Harmonics as a tool to solve vectorial differential equations. Any vector in three dimensional Euclidean space can be uniquely decomposed into radial, ``longitudinal" and ``transverse" components that are defined through spherical harmonics. These components will be orthogonal. Vectorial differential equations like divergence and curl nicely factorizes in terms of these, and turn into a one dimensional differential equation in terms of functions of the radius as we will explain in the first section. The momentum constraint will become the curl of curl of a vector which we solve. 

In Section \ref{sec:Spheremom} we discuss this solution and what does the action look like in this case. We discuss the isotropy subgroup of the symmetries of this action, and the structure of the space of vacua. We comment on the consequences of the Hamiltonian constraint.
 
\section{Vector Spherical Harmonics}\label{sec:vecsph}
It has been shown in \cite{barrera} that any vector in three dimensional Euclidean space can be expanded in terms of spherical harmonics as
\begin{equation}
\vec{\xi}(r,\theta,\phi)= \sum_{lm} \left( \xr(r) \vec{Y}_{lm} + \xd(r) \vec{\Psi}_{lm} + \xt(r) \vec{\Phi}_{lm} \right) \sgc
\end{equation}
where 
\begin{align}
\vec{Y}_{lm} &= \hat{r} Y_{lm} \sgc \\
\vec{\Psi}_{lm} &= r \vec{\nabla} Y_{lm} \sgc \\
\vec{\Phi}_{lm} &= \vec{r} \times \vec{\nabla} Y_{lm} \sgc
\end{align}
where $Y_{lm}$s represent spherical harmonics. These three quantities are called vector spherical harmonics which we will further name as radial, longitudinal and transverse spherical harmonics respectively. These satisfy some very useful properties such as
\begin{equation}
\vy \cdot \vphi=0 \cg \vy \cdot \vpsi=0 \cg \vphi \cdot \vpsi=0 \sgc
\end{equation} 
and also
\begin{align}
\int \vol_{S^2} \vy \cdot \vec{Y}^*_{l'm'} &= \delta_{ll'} \delta_{mm'} \sgc \\
\int \vol_{S^2} \vpsi \cdot \vec{\Psi}^*_{l'm'} &= \int \vol_{S^2} \vphi \cdot \vec{\Phi}^*_{l'm'} = l(l+1) \delta_{ll'} \delta_{mm'} \sgc \\
\int \vol_{S^2} \vy \cdot \vec{\Psi}^*_{l'm'} &= \int \vol_{S^2} \vy \cdot \vec{\Phi}^*_{l'm'} = \int \vol_{S^2} \vpsi \cdot \vec{\Phi}^*_{l'm'} = 0 \sgd
\end{align}
Now we would like to see how some vector differential equations are solved by using this expansion. We start by writing down some expressions using this expansion, using equations (3.11a-3.14) in \cite{barrera}:
\begin{align}
\vn.\vec{\xi} &= \left( \frac{1}{r^2} \frac{d}{dr} \left( r^2 \xr \right) - \frac{l (l+1)}{r} \xd \right) Y_{lm} \sgc \label{div}\\ 
\vn \times \vec{\xi} &=- \frac{\xr}{r} \vphi + \frac{1}{r} \frac{d}{dr} \left( r \xd \right) \vphi - \frac{l(l+1)}{r} \xt \vy - \frac{1}{r} \frac{d}{dr} \left(r \xt \right) \vpsi \sgc \label{curl}\\
\vn \times \vn \times \vec{\xi} &= \left( \frac{l (l+1)}{r^2} \xt  -  \frac{d^2 \left( r \xt \right)}{r dr^2} \right) \vphi - \frac{l (l+1)}{r} \tilde{\xi}_{lm} \vy - \frac{d \left( r \tilde{\xi}_{lm} \right)}{r dr}  \vpsi \sgc \label{curlofcurl}
\end{align}
where we use Einstein summation convention and define
\begin{equation}
\tilde{\xi}_{lm}= - \frac{\xr}{r} + \frac{1}{r} \frac{d}{dr} \left( r \xd \right) \sgd
\end{equation}
As one can already see using this expansion vectorial differential equations will become one dimensional differential equations in $r$. Lets now investigate the solutions to the equations where the above differential expressions are set to zero.
\paragraph{Solution to divergence equation:} 
Divergence equation $\vn.\vec{\xi}=0$ is solved by setting equation \eqref{div} to zero. This equation gives
\begin{equation}\label{dive}
\xd l (l+1) =2 \xr + r \partial_r \xr \sgc
\end{equation}
thus the general solution to this equation can be written as (except l=0 term)
\begin{align}
\vec{\xi} &= \xt \vphi + \frac{1}{l (l+1)} \partial_r \left(r^2 \xr \right) \vn Y_{lm} + \xr Y_{lm} \hat{r} \\
&= \xt \vphi + \frac{1}{l (l+1)} \left( \hat{r}. \vn \left(r^2 \xr \right) \vn Y_{lm} - r^2 \xr \nabla^2 Y_{lm} \hat{r} \right) \sgc
\end{align}
which upon using (\cite{jackson})
\begin{equation}
\vn \times \left( \hat{r} f(r) \times \vn g \right) = \hat{r} f \nabla^2g - \frac{\partial f}{\partial r} \vn g
\end{equation}
takes the form
\begin{equation}
\vec{\xi}= \xt \vphi - \frac{1}{l (l+1)} \vn \times \left( \hat{r} r^2 \xr \times \vn Y_{lm} \right) + \frac{a}{r^2} \hat{r} \sgc
\end{equation}
where we have also added $l=0$ part.
\paragraph{Solution to curl equation:} 
Curl equation $\vn \times \vec{\xi}=0$ can be solved by setting equation \eqref{curl} to zero. This gives
\begin{equation}
\xr= \frac{d}{dr} \left( r \xd \right) \quad \mbox{and} \quad \xt=0 \sgc
\end{equation}
thus the general solution can be written as
\begin{align}
\vec{\xi} &= \frac{d}{dr} \left( r \xd \right) \hat{r} Y_{lm}+ \xd r \vec{\nabla}Y_{lm} \\
&= \vn \left( r \xd Y_{lm} \right) = \vn f
\end{align}
for arbitrary $f$.
\paragraph{Solution to curl of curl equation:} 
Curl of curl equation $\vn \times \vn \times \vec{\xi}=0$ can be solved by setting \eqref{curlofcurl} to zero. The coefficients of $\vec{\xi}$ should satisfy
\begin{align}
\tilde{\xi}_{lm} \rightarrow \quad & \xr = \frac{d}{dr} \left( r \xd \right) \sgc\label{vhe1}\\
& \frac{l (l+1)}{r^2} \xt  - 2 \frac{\partial_r \xt}{r} - \partial^2_r \xt =0  \sgc\label{vhe2}
\end{align}
where the last equation is solved as
\begin{equation}
\xt= c_{lm} r^l + d_{lm} \frac{1}{r^{(l+1)}} \sgd
\end{equation}
Thus solution has one undetermined function of r, and written as
\begin{equation}
\vec{\xi} = \left(c_{lm} r^l + d_{lm} \frac{1}{r^{(l+1)}} \right) \vphi + \vn \left( r \xd Y_{lm} \right) \sgd
\end{equation} 

We end this section by writing down the vector spherical harmonic in the spherical coordinates:
\begin{align}
\Phi_{lm} &= - \frac{1}{r sin \theta} \partial_{\phi} Y_{lm} \frac{\partial}{\partial \theta} +  \frac{1}{r sin \theta} \partial_{\theta} Y_{lm} \frac{\partial}{\partial \phi} \sgc \\
\Psi_{lm} &=  \frac{1}{r} \partial_{\theta} Y_{lm} \frac{\partial}{\partial \theta} +  \frac{1}{r sin^2 \theta} \partial_{\phi} Y_{lm} \frac{\partial}{\partial \phi} \sgd
\end{align}
Thus a general spatial vector field can be written as
\begin{align}
\xi(r,\theta,\phi)=  \xi^r_{lm}(r) Y_{lm} \frac{\partial}{\partial r} & + \left(  - \frac{\xi^{\bot}_{lm}(r)}{r sin \theta} \partial_{\phi} Y_{lm}  +\frac{\xi^{||}_{lm}(r)}{r} \partial_{\theta} Y_{lm} \right) \frac{\partial}{\partial \theta} \nonumber \\ &+  \left(   \frac{\xi^{\bot}_{lm}(r)}{r sin \theta} \partial_{\theta} Y_{lm}  +\frac{\xi^{||}_{lm}}{r sin^2 \theta} \partial_{\phi} Y_{lm} \right) \frac{\partial}{\partial \phi} \sgd
\end{align}

\section{Momentum Constraint and Action} \label{sec:Spheremom}
Now we turn back to our problem. Recall that we are interested in solving \gr~for Manton approximation, and concluded that in \gnc~they are described by spatial vector fields on each slice, where the metric of each slice is
\begin{equation}
\vacspmet(t)= \sigma_t^* \refmet \sgd
\end{equation}
For the case at hand we will take $\refmet_{ab}= \delta_{ab}$. Here we stress again that each spatial slice is diffeomorphic to flat space, but the spacetime will not be diffeomorphic to a flat spacetime. So for a slice, written in its flat coordinates the momentum constraint is
\begin{equation}\label{eq:flatmomentum}
(\dad d \chi)_i= (*d*d \chi)_i= \epsilon_{ijk} \pr_j \lp \epsilon_{kmn} \pr_m \chi_n \rp =0 \sgd
\end{equation}
Then we switch to the spherical coordinates and let
\begin{equation}
\vacspmet= dr^2 + r^2 \lp d\theta^2 + \sin^2(\theta) d\phi^2 \rp \sgc
\end{equation}
so that
\begin{equation}
n= \frac{\pr}{\pr r} \cg \lbr \ea \rbr = \lbr \frac{\pr}{\pr \theta}, \frac{\pr}{\pr \phi} \rbr \sgc
\end{equation}
and
\begin{equation}
\twoconnbnd^a_b = \extbnd^a_b \cg \accbnd^a=0 \sgd
\end{equation}
Boundary is the 2-sphere with radius R and the induced metric on it is
\begin{equation}
\metbnd= R^2 \lp d\theta^2 + \sin^2 \theta d\phi^2 \rp \sgd
\end{equation}
The metric on the space of vacua becomes
\begin{equation}
\metonvacua_{\delta} \lp \zeta, \xi \rp= \int_{S_R} \vol_{S^2_R} \lp \zeta_a \pr_r \xi^a \rp \sgd
\end{equation}

Now let us explicitly solve the momentum constraint. We will utilize the vector spherical harmonics for this. We know from Euclidean vector calculus equation \eqref{eq:flatmomentum} becomes the curl of curl so that momentum constraint becomes
\begin{equation}
\momi(\chi)= \vn \times \vn \times \vec{\chi}=0 \sgd
\end{equation}
From the previous section we know the solution to this. Inside the ball we require regularity, and for this to be true at zero we set coefficients with inverse power of r to zero. Thus solution to \eqref{eq:flatmomentum} inside the ball is
\begin{equation}\label{eq:sphsol}
\vec{\chi}= a_{lm} r^l \vphi + \vn f(r,\theta, \phi) \sgd
\end{equation}
Upon this one should impose the condition $\no \chi= \left. \chi^r \right|_{r=R} =
0$. Since $\vphi$ is orthogonal to the radial direction this only imposes
\begin{equation}
\left. \pr_r f \right|_{r=R} = 0 \sgd
\end{equation}
Let us expand this function in spherical harmonics as
\begin{equation}
f(r,\theta,\phi)= f_{lm}(r) Y_{lm}(\theta,\phi) \sgc
\end{equation}
then the condition becomes
\begin{equation}
\dot{f}_{lm}(R)=0 \sgd
\end{equation}
Note that equation \eqref{eq:sphsol} reflects the results we have obtained in section \ref{sec:boundarydata}: Solution to the momentum constraint is composed of two parts: one exact, one coexact. This is because
\begin{equation}
\vec{\nabla} \cdot \vphi=0 \sgd
\end{equation}
The coexact part is uniquely determined by its value on the boundary while the exact part is completely arbitrary. This arbitrary part is expected to be fixed by the Hamiltonian constraint. In section \ref{sec:decconst} we have attempted to solve this constraint without success. However since the norm of the vector field corresponding to the solution with respect to the metric on the vacua is Hamiltonian constraint integrated on the boundary, we will be able to gather some information from the expression for the metric on the vacua.To write it down first we note
\begin{align}
\mathbb{D}\chi^a= \left. \pr_r \chi^a \right|_{r=R}= R \left[\sqrt{\metbnd}\epsilon^{ab}\partial_b\left(\sum_{l,m} (l-1) a_{l m} Y_{l m}\right)-2 \metbnd^{ab}\partial_b \lp f_{lm}(R) Y_{lm} \rp \right]\label{dsphere}.
\end{align}
Note that $\mathbb{D}$ has a kernel such that
\begin{equation}
ker(\mathbb{D})= \lbr \chi^a= \sum_{m=-1,01}  \frac{a_{1m}}{R} \volk^{ab} \pr_b Y_{1m}\rbr \sgd
\end{equation}
Note that this exactly corresponds to the set of rotations. For this specific case the isotropy subgroup will be equal to $\ker \mathbb{D}$. The reason for this is that the isotropy group is formed of boundary preserving bulk isometries, and isometries of 3d Euclidean space are rotations and translations, however translations are not boundary preserving. So we conclude that our space of vacua for this case is
\begin{equation}
\boxed{ \truevac= \diff(S^2) / \SO(3) \sgd}
\end{equation}
Using the information above for the metric on the vacua we will finally have
\begin{equation}
\langle\zeta_{(1)},\zeta_{(2)}\rangle=R^3\sum_{l,m} l (l+1)\left( (l-1) a^{(1)}_{l m} a^{(2)}_{l m} - 2 f_{lm}(R)^{(1)} f_{lm}(R)^{(2)} \right) \, .\label{spmetric}
\end{equation}
Hamiltonian constraint forces a solution $\zeta$ to be null with respect to this metric:
\begin{equation}
\langle\zeta,\zeta \rangle=R^3\sum_{l,m} l (l+1)\left( (l-1) a^2_{l m}  - 2 f^2_{lm}(R) \right)=0 \, .
\end{equation}
Note because of this, coexact and exact part are not independently solutions to the constraint equations, a solution needs to have both of these parts.

%% file: discussion.tex
\chapter{Discussion}\label{ch:discussion}

Having studied a specific example in the last chapter we conclude the thesis by summarizing our results, discussing open issues and further directions for investigation. In this thesis we have studied some adiabatic solutions of the theory of general relativity by using Manton approximation. Starting out with a space-time splitted form of the \eh~action in Gaussian Normal Coordinates where the spacetime is solely defined by the metric of the spatial slices, we proposed solutions that are described as slow motion on the space of vacua or equivalently solutions that are spacetimes whose a spatial slice at an instant ``t" is the pullback of a reference metric by a flow on it at the instant t, i.e.
\begin{equation}
\vacspmet(t)= \sigma_t^* \refmet \sgd
\end{equation}
By studying the form of the constraint equations we have concluded-with arguments that partially relied on a conjecture about the Hamiltonian constraint- that the solutions are completely defined as the null geodesics on the space of vacua with respect to the metric we have defined. As we have discussed this metric is pseudo-Riemannian, and the space of vacua is an infinite dimensional pseudo-Riemannian homogeneous space diffeomorphic to
\begin{equation}
\truevac \sim \diff(\bnd)/\bdiff_{\refmet}(\subm) \sgd
\end{equation}
Note that with this analysis one can also conclude that diffeomorphisms that do not vanish on the boundary are not to be thought as gauge equivalences but rather as global symmetries. This conclusion is somewhat intuitive, though not always emphasized, since one should expect that the coordinate transformations that changes the boundary values of the fields of the theory should not be considered as redundancies in the description. 
\paragraph{Open issues.} Our analysis have resulted in a number of open issues. First we remind the reader that in applying Manton approximation we ignored solutions that could be generated by local boosts. This however was mostly due to the conceptual difference between them. Here we note that just like given a soliton solution, one can get other solutions by making a Lorentz boost \cite{srednicki}; taking a solution we have proposed one might get another solution by making a local boost.

Another open issue is that we have only conjectured that Hamiltonian constraint uniquely fixes the part of solutions to the momentum constraint that were in the form of an exact form and was completely arbitrary. Even though we could use the theory of differential forms and \hmf~decomposition on a manifold with boundary to prove the existence and uniqueness of solutions to the momentum constraint since it could be brought into a exterior  calculus form, doing this does not look to be trivial for the Hamiltonian constraint.

This brings us to the issue of solving the constraints explicitly. Note that we could solve the momentum constraint for 3-dimensional and 2-dimensional Euclidean spaces, we did not get specific solutions for the more general cases. This is closely related to finding explicit harmonic forms for a manifold with general Ricci flat metric, and it might be tackled using the results of \cite{duff}. Note that for our analysis not only solving the constraints in the bulk is important, one also needs to impose boundary conditions such as boundary preservingness. For example instead of a round ball of chapter \ref{ch:ball}, one can consider an ellipsoidal boundary which should be tractable.Here we also remind that we were not able to solve explicitly the Hamiltonian for neither 3d nor 2d Euclidean case yet, but these might also be solvable either by some more guesses or via the tools of mathematical computer programs.

Related to the construction on the configuration space and space of vacua, one question remained to be answered was the equivalence of $\ker \mathbb{D}$ to the isotropy group, or the isometries of the bulk. We have shown that all of the bulk isometries are in $\ker \mathbb{D}$ but it was not clear whether the inverse is also true or not. If it is not true then we would have a degenerate metric on the space of vacua.

We conclude this part by pointing out to potential subtleties related to the fact that the symmetry groups, i.e. diffeomorphism groups, we considered for the study of the homogeneous space structure. We have taken these to be Lie groups and proceeded with the homogeneous space analysis. However in the literature it was pointed out that the proper setting for infinite dimensional groups is what is called ``Lie-Frechet" groups \cite{oblakthesis}. These are infinite dimensional groups that are taken locally to be Frechet spaces. Note that this subtlety is related to differentiability properties and not group properties. Because of this we expect group related properties of our analysis, such as the procedure of double quotient to still hold, but there might be some issues with the smoothness of the resulting spaces. See \cite{khesinwendt} for a discussion of infinite dimensional groups.
\paragraph{Future Directions.} We have discussed the open issues and possible direct extensions of the procedure we have proposed for finding adiabatic solutions of general relativity. Let us now elaborate on connections to the literature and future directions. 

We were motivated by the studies of asymptotic symmetries, but we have studied a theory on finite spatial volume. An obvious question is whether this procedure can be extended to infinite volume case where soft theorems appear by a limiting procedure. One can also do this by reconsidering the procedure on a manifold without a boundary, but a set of asymptotic fall off conditions as was considered in asymptotic symmetry analysis.

Another direction might be the investigation of the specific relationship between our analysis and adiabatic modes in cosmology. As we have mentioned our propositions of solution is very much in line with the proposition of adiabatic solutions in \cite{weinbergadiabatic}. Note that this specific form we have used caused, together with the momentum constraint, the theory to reduce to a theory on the boundary of the spatial slice. One might inquire if something along similar lines happens for adiabatic modes. Note however that one crucial difference is that while we consider slow-time solutions, adiabatic modes are solutions that have small inhomogeneity but that turns our to be slow time as a result of Einstein equations.

%% file: mathbackgnd.tex
%This I have directly copied from my mathematics latex note. I will manipulate it as I write the thesis

\chapter{Introduction}

In this appendix, we summarize the mathematical definitions and theorems one needs to be familiar with for a study of modern theoretical physics, but the focus will be somewhat on the tools needed to study gravitation in some form. To be more precise, we will discuss manifolds, specifically the differentiable ones, and various objects on them such as tensor fields, flows, forms, symplectic forms, metrics and so on. General relativity tells us that one should consider the spacetime as a differentiable manifold and most of the objects aforementioned has a corresponding physical reality or use on that spacetime manifold. However in a study of theoretical physics one also often encounters more abstract manifolds that describe things like a phase space (like a symplectic manifold) or set of symmetry transformations (like a Lie group, which is a differentiable manifold that happens to be a group) or a configuration space (which can be considered as a Riemannian manifold if one can define a metric on it). 

This variety in the applications of the concept of differentiable manifolds to different physical concepts makes it desirable to study each quantity that can be defined on a manifold carefully, paying attention to which prerequisites each needs e.g. it is important to notice even though one can define a Lie derivative in terms of a covariant derivative, existence of a covariant derivative is not necessary actually, one can define the Lie derivative on a manifold without a metric. In this note I try to keep attention to this aspect, and definitions that have more prerequisite concepts are defined later on e.g. Lie derivative is defined before the covariant derivative.

This appendix is thus written with the aim of collecting together the basic definitions we will need during the main text and set the notations used. Most of the following is a form of summary of John M. Lee's books, Introduction to Smooth Manifolds \cite{LeeISM} and Riemannian Manifolds \cite{LeeRM}; however ordering of the topics and examples are original. Other references we use for this includes Robert Wald's General Relativity \cite{wald}, Wu Ki Tung's Group Theory in Physics \cite{wukitung}. I have also utilized unpublished lecture notes of Sadık Değer's differentiable manifolds class and Dieter Van den Bleeken's Group Theory class I have taken, and indebted to them for the knowledge I have of these topics. Other references are cited at the places relevant in the document. Note that we skip certain definitions that are considered self-explanatory and most of the proofs of the theorems and propositions stated, so whenever the need for a more precise definition more rigor is felt, references aforementioned should be consulted.

\chapter{Differentiable Manifolds}\label{ch:mathbackgnd}

In this chapter we start with describing differentiable manifolds: smooth spaces that, in the vicinity of each of its points, look like Euclidean space of the same dimension. These spaces will be our basic construction, and in the following chapters we will build upon them by adding extra structures; note for example that a differentiable manifold does not have a metric on it, so at this stage it does not carry any information of length. After defining them in the first section, we introduce the concepts of vectors, vector fields, tensors, antisymmetric tensors, flows, Lie derivatives, orientation, integration and foliations.

\section{Definitions and Basic Elements}
In this section we first go over the basic definitions of manifolds: topological manifolds and differentiable manifolds-topological manifolds with a differentiable set of maps defined on it. Then we introduce basic elements on a differentiable manifold: tangent vectors, tensors -direct products of tangent vectors- and fields of those quantities. We begin our inspection with definition of certain types of maps, that will be frequently used thereafter.

\subsection{Maps} 

Consider a map $F: M \ra N $, where M and N are two sets, then the following hold:
\begin{itemize}
\item $F$ is surjective, if it covers all of N.
\item $F$ is injective, if no two elements of M are mapped to the same element of N.
\item $F$ is bijective, if it is both surjective and injective.
\item $F$ is continuous at $ p \in M $ if any open set containing $ F(p)$ contains the image of an open set of $M$ containing $p$.
\item Let $ M,N \subset \R^n$ and $F=(F_1,...,F_n)$. If all $F_i$ are k-times differentiable, i.e. each $F_i$ is continuous and all first k derivatives of it exists, then $F$ is called a $C^k$ map. A $C^\infty$ map is also called a smooth map.
\item If $F$ is continuous, and if its inverse exists and also continuous, then $F$ is a homeomorphism.
\item Let $ M,N \subset \R^n$. If $F$ is differentiable, and if its inverse exists and also differentiable, then $F$ is a diffeomorphism.
\end{itemize}

\subsection{Manifolds}

\paragraph{Topological Manifolds}
\begin{defn}
For any set $K$ and any collection $ T = \lbr U_i \rbr  $ of some subsets $U_i$ of $K$, $(K,T)$ called a topology if\noclub[4]
\begin{enumerate}
\item $\emptyset$, $K \in T$,
\item Any union of $U_i$'s is an element of $T$,
\item Any intersection of any finite number of $U_i$ is an element of $T$.
\end{enumerate}
\end{defn}
\noindent
If $(K,T)$ is a topology then $K$ is called a topological space and $U_i$ is called an open set. A subset $V$ of $K$ is called a closed set if $V^c \equiv K-V$ is an open set.

%Topological space X is said to be compact if every open cover (cover that is composed of open subsets) of X has a finite ( countable: condition of second countability we have imposed) subcover (subcollection of a cover that is still a cover).

%Closed manifold: Compact manifold (without boundary)
%Open manifold: non-compact manifold (without boundary)

\begin{defn}
A topological manifold $M$ of dimension n is a topological space with the following properties:
\begin{enumerate}
\item $M$ is Hausdorff: $\forall \ p,q \in M$ there exists some open subset $U_p,U_q$ of $M$, where $p \in U_p$ and $ q \in U_q$ such that $U_p \cap U_q = \emptyset$.
\item Each point of $M$ has a neighborhood $U$ homeomorphic to an open set of $\R^n$ with a homeomorphic map $\phi$.
\item $M$ has a countable basis of open sets.
\end{enumerate}
\end{defn}
$(U,\phi)$ is called a coordinate neighborhood and also sometimes a chart. Two coordinate neighborhoods are said to be $C^\infty$ compatible if the change of coordinates from one to the other on an overlapping region is $C^\infty$.
\paragraph{Differentiable Manifolds}
\begin{defn}[Smooth Manifold]
A differentiable (or smooth) manifold $(M,A)$ is a topological manifold $M$ together with a smooth structure $A$ on it which is a family $ A=\lbr (U_\alpha, \phi_\alpha) \rbr $ of coordinate neighborhoods, called an atlas, such that
\begin{enumerate}
\item The union of $U_\alpha$'s cover $M$,
\item For any $\alpha,\beta$ the neighborhoods $\left(U_\alpha,\phi_\alpha\right)$ and $\left(U_\beta,\phi_\beta\right)$ are $C^\infty$ compatible,
\item Any coordinate neighborhood compatible with every coordinate neighborhood in $A$ is itself in $A$ (Maximality).
\end{enumerate}
\end{defn}
If $(M,A)$ is a smooth manifold, a coordinate neighborhood contained in the given smooth atlas is called a smooth chart.\paragraph{Differentiability on Manifolds:}
A function $f:M \ra \R$ is called smooth if for every point $p \in M$, there exists a smooth chart $(U,\phi)$ s.t. $p \in U$ and $ \hat{f} \equiv f \circ \phi^{-1}$ is smooth on $\phi(U)$. Set of all smooth functions on $M$ is a \hyperlink{vectorspace}{vector space} called $C^\infty(M)$. 

Let $M$ and $N$ be smooth manifolds and $F:M \ra N$ be map. $F$ is smooth if $\forall \, p \in M$ there exists smooth charts $(U,\phi)$ containing $p$ and $(V,\psi)$ containing $F(x)$ s.t $F(U) \subset V$ and $\psi \circ F \circ \phi^{-1}$ is smooth from $\phi(U)$ to $\psi(V)$. 

\subsection{Derivation on Manifold and Tangent Space}
\begin{defn}
Let $M$ be a smooth manifold. The map $X: C^\infty (M) \ra \R $ is called a \hypertarget{derivation}{derivation} at $p \in M$ if it satisfies the linearity and Leibniz properties, i.e. if
\begin{enumerate}
\item $ X(\alpha f+\beta g)= \alpha X(f) + \beta X(g)$,
\item $ X(fg)= X(f) \left. g \right|_p + \left. f \right|_p X(g)$,
\end{enumerate}
where $\alpha,\beta \in \R$, $f,g \in C^\infty(M)$. Set of all derivatives of $C^\infty(M)$ at $p$ is called the tangent space $T_pM$ and it is a vector space since
\begin{enumerate}
\item $ (X+Y) (f)= X(f)+Y(f)$,
\item $(\alpha X) (f)= \alpha X(f)$.
\end{enumerate}
An element of $T_pM$ is thus called a tangent vector at $p$.
\end{defn}
\begin{defn}[Pushforward]
Let $F:M \ra N$ be a smooth map, and X be a derivation at $p \in M$. The pushforward of $F$, $F_*:T_pM \ra T_{F(p)}N$ is defined by
\begin{equation}
(F_*X) f \equiv X(f \circ F) \quad \quad f \in C^\infty(N) \sgd
\end{equation}
Note that
\begin{enumerate}
\item $F_*X$ is a derivation at $F(p)$,
\item $F_*$ is also called the differential of $F$.
\end{enumerate}
\end{defn}
\paragraph{Pushforward of Tangent Vectors in $\R^n$ to $M$:}
Note that $\frac{\partial}{\partial x_i}$ at a point in $\R^n$ is a \hyperlink{derivation}{derivation}; and it is well defined on all points in $\R^n$. Let $(U,\phi)$ be a smooth chart on $M$. Since $\phi$ is a diffeomorphism, it has an inverse $\phi^{-1}: \R^n \ra M$. Thus at a point $p \in M$ one can define:
\begin{equation}
\left. \frac{\pr}{\pr x_i} \right|_p \equiv \left( \phi^{-1} \right)_*  \left. \left( \frac{\pr}{\pr x_i} \right) \right|_{\phi(p)} \sgd
\end{equation}
Note the slight abuse of notation here, this notation will be/is used throughout this document/in the literature, i.e. when someone talks about $\frac{\pr}{\pr x_i}$, generally this definition is meant; the cases where it is not should be clear from the context. Note also that for every ``geometric vector" at a point in $\R^n$ there exists a directional derivative and vice versa. Thus the tangent space at a point in $\R^n$ is n-dimensional. Since $\phi$ is a diffeomorphism, $\phi_*$ is an isomorphism, and thus the tangent space at point in $M$ is also n-dimensional.

\subsection{Vector Fields and Tangent Bundles}
\begin{defn}[Tangent Bundle]
Tangent Bundle $TM$ is the \hyperlink{disjointunion}{disjoint union} of tangent spaces at all points of $M$:
\begin{equation}
TM = \coprod_{p \in M} T_pM \sgd
\end{equation}
Elements of $TM$ are denoted as $(p,v)$. $TM$ is a 2n dimensional smooth manifold. (Dimensionality can be seen from $\R^n$ case with an argument close to above, for the proof of the smoothness see \cite{LeeISM}, pg.66.)
\end{defn}
\begin{defn}[Vector Field]
Vector field is a continuous map $Y: M \ra TM$ with the property $ \Pi \circ Y = \1_M$; where $ \mathbb{1}_M: M \ra M $ is the identity map on $M$ and $\Pi:TM \ra M$ is the projection $\Pi(p,v)=p$. $Y$ maps $p \ra Y_p$ s.t. $Y_p \in T_pM$.
\end{defn}
\noindent
Vector fields possess the following properties:
\begin{itemize}
\item Action of a vector field $Y$ on a function $f$ is another function given by
\begin{equation}
Yf(p)=Y_pf \sgd
\end{equation}
\item Set of all vector fields on $M$, denoted by $\vfs(M)$, is a vector space 
\begin{enumerate}
\item linear under pointwise addition and scalar multiplication: $ \left( a Y + Z \right)_p= a Y_p + b Z_p$, where $a,b \in \R$,
\item whose zero element is the vector field which is zero everywhere,
\item and on which a multiplication by a smooth function is defined pointwise: $\left. fY \right|_p= f(p) Y_p $ .
\end{enumerate}
\item $ Y \in \mathfrak{X}(M) $ is a derivation on $M$ since it satisfies the linearity and Leibniz rules:
\begin{equation}
Y(fg)=f Y(g) + g Y(f) \cg
\end{equation}
where $f,g$ functions on $M$.
\end{itemize}
\paragraph{F-related Vector Fields:} Let $F:M \ra N$ be a map and $X$ be a smooth vector field on $M$. If there exists a smooth vector field $\tilde{X}$ on $N$ such that
\begin{equation}
F_*(X_p)= \tilde{X}_{F(p)} \quad \quad \forall p \in M \sgc
\end{equation}
$X$ and $\tilde{X}$ are said to be $F$-related. Existence of $\tilde{X}$ is guaranteed only if $F$ is a diffeomorphism, i.e. if $F$ is a diffeomorphism, then for every smooth vector field on $M$ there is a unique $F$-related smooth vector field on $N$. This vector field is called the \textit{pushforward of vector field X by F}.
\begin{defn}
A vector bundle of rank-k over $M$ is a topological space
that is a collection of k-dimensional real vector spaces at each point of the manifold such that at each point there exists a neighborhood $U$ so that collection of these vector spaces at this neighborhood is diffeomorphic to $U \times \R^k$.
\end{defn}
Note that $TM$ is a smooth vector bundle of rank n over $M$.
% $E$ together with a surjective continuous map $\pi: E \ra M$ satisfying:
%\begin{enumerate}
%\item For each $p \in M$; the fiber $E_p= \pi^{-1}(p)$ over p is endowed with the structure of a k-dimensional real vector space, 
%\item and there exists a coordinate neighborhood $U$ of $p$ with a homeomorphism between $\pi^{-1}(U)$ and $U \times R^k$ that is a vector space isomorphism at each $q \in U$.
%\end{enumerate}

\subsection{Covectors and Tensors}

\begin{defn}
Cotangent space at a point $p$ in a manifold $M$ is the dual space of the $T_pM$, denoted by $T_pM^*$. Disjoint union of all cotangent spaces on a manifold is cotangent bundle. An element of the cotangent bundle is called a covector field.
\end{defn}

\begin{defn}[Pullback]
Let $F: M \ra N$ be a smooth map, then $F^*: T_{F(p)}N^* \ra T_pM^*$ is called the pullback of $F$ and its action given by
\begin{equation}
(F^* \omega)X \equiv \omega (F_*X) \sgc
\end{equation}
where $X \in T_pM$ and $\omega \in T_{F(p)}N^*$. Generalization of this to covector field is trivial. Note that for a smooth map $F: M \ra N$ pullback of a covector field on $N$ is always guaranteed to exists, i.e. inverse of $F$ is not needed to exists. 
\end{defn}

\begin{defn}[Differential of a function]
Let $f$ be a smooth function on $M$. Differential of $f$ is a covector field, denoted $df$, defined as
\begin{equation}
df(X_p) = X_p(f) \sgd
\end{equation}
This definition asserts that df is a smooth covector field.
\end{defn}
Let $\phi=( x^1,...,x^n)$ be a coordinate chart of $M$. Then $x^i$ is a smooth function on $M$. $dx^i$ is dual to the basis $\frac{\partial}{\partial x^i}$, i.e. $dx^j \left( \frac{\partial}{\partial x^i} \right)= \delta^j_i$.

\begin{defn}[Tensors]
A type (k,l) tensor on a vector space $ V $ is a \hyperlink{multilinear}{multilinear} map
\begin{equation*}
\underbrace{ V^{*}\otimes \dots \otimes V^{*}}_{k}\otimes \underbrace{ V \otimes \dots \otimes V}_{l} \rightarrow \mathbb{R} 
\end{equation*}
denoted by $ T^{k}_{l} $. The space of type (k,l) tensors on V is defined as
\begin{equation*}
T^{(k,l)}(V)= \underbrace{ V \otimes \dots \otimes V}_{k}\otimes \underbrace{ V^{*} \otimes \dots \otimes V^{*}}_{l} \sgd
\end{equation*}
\end{defn}

\begin{defn}[Tensor Bundles and Tensor Fields]
The bundle of type (k,l) tensors on a smooth manifold $M$ is defined by
\begin{equation}
T^{(k,l)}TM= \prod_{p \epsilon M}{T^{(k,l)}(T_p M)} \sgd
\end{equation}
A section of a type (k,l) tensor bundle is a type (k,l) tensor field.
\end{defn}
\paragraph{Pullback of Tensor Fields:}
Let $F:M \ra N$ be a smooth map and $A$ be a (0,k) type tensor field on $N$. The pullback of $A$ by $F$ is defined by
\begin{equation}
(F^*A)_p(v_1,...,v_k)= A((F_*v_1)_p,...,(F_*v_k)_p) \sgd
\end{equation}
If $F$ is a diffeomorphism then this can be generalized to type (k,l) tensor fields:
\begin{equation}
(F^*A)_p(\omega_1,...,\omega_k,v_1,...,v_l)=A( (F^{-1})^* \omega_1,...(F^{-1})^* \omega_k,F_* v_1) \sgd
\end{equation}

\subsection{p-forms and Exterior Derivative}

\begin{defn}[p-Covectors]
A p-covector is a (0,p) tensor which is completely antisymmetric. A 0-covector is a function and a 1-covector is a (0,1) tensor. The vector space of all p-covectors on $ V $ is denoted by $ \Lambda^{p}\left(V^{*}\right) $.
\end{defn}

\begin{defn}[Wedge Product]
Wedge product of a p-covector A and a q-covector B is defined as
\begin{equation}
(A \wedge B )_{\mu_{1} \dots \mu_{p+q}}= \frac{(p+q)!}{p! q!} A_{[ \mu_1 \dots \mu_p} B_{\mu_{p+1} \dots \mu_{p+q}]} \sgd
\end{equation}
\end{defn}

\begin{defn}
Let $V$ be a finite dimensional vector space. For each $ v \in V $ we define the interior multiplication by $v$ to be the map $ i_v : \Lambda^p(V^{*}) \rightarrow \Lambda^{p-1} (V^{*}) $ such that
\begin{equation}
i_v \alpha( w_1, \dots, w_{p-1})= \alpha(v,w_1, \dots, w_{p-1})
\end{equation}
where $ \alpha \in \Lambda^p(V^{*}) $, and we interpret $ i_v \alpha=0 $ when k=0, i.e. when $\alpha$ is a number. $i_v \alpha$ is also denoted as $ v \inc \alpha $.
\end{defn}

\begin{prop}
Let $\alpha$ be a covariant k-tensor on a finite-dimensional vector space
V . The following are equivalent:
\begin{enumerate}
\item $\alpha$ is a k-covector.
\item $\alpha(v_1,...,v_k)=0$ whenever the k-tuple $\lbr v_1,...,v_k \rbr $ is linearly dependent.
\end{enumerate}
\end{prop}

\begin{defn}[p-form]
The subset of $ T^{(0,p)}TM $ consisting of p-covectors is denoted by $ \Lambda^p T^{*}M $:
\begin{equation}
\Lambda^k T^{*}M = \prod_{p \epsilon M}{\Lambda^{k}(T_{p}^{*} M)} \sgd
\end{equation}
A section of $ \Lambda^p T^{*}M $ is called a (differential) p-form: it is a tensor field which is a p-covector at each point.
\end{defn}
Wedge product and interior product are generalized for p-forms as pointwise operators.

The set of smooth p-forms on a manifold $M$ form a vector space denoted by $\pfs(M)$.

\begin{defn}[Exterior Derivative]
Exterior derivative is a map $ d : \pfs(M) \rightarrow \Omega^{p+1}(M)$ given by
\begin{equation}
(dA)_{\mu_1 \dots \mu_{p+1}}= (p+1) \partial_{ [ \mu_1} A_{\mu_2 \dots \mu_{p+1}]} \sgd
\end{equation}
For a 0-form exterior derivative is the differential.
\end{defn}

A p-form A is called
\begin{itemize}
\item  closed if $dA=0$,
\item exact if $A=dw$, for $w$ some p-1 form.
\end{itemize}
Note that every exact form is closed, but not every closed form is exact. Closed forms that differ from each other by a exact form constitutes equivalence classes. Quotient by this equivalence relation is has a special definition:
\begin{defn}[deRham Cohomology Group]
deRham Cohomology Group is the quotient
\begin{equation}
\derhamp(M) \equiv \frac{Z^p(M)}{B^p(M)} \equiv \frac{\mbox{closed p-forms on M}}{\mbox{exact p-forms on M}} \sgd
\end{equation}
\end{defn}

It turns out that deRham cohomology groups are closely related with the topological properties of manifolds. More specifically,
\begin{lemma}[Poincare Lemma]\cite{nakahara}
\hypertarget{poincare}{On a manifold that is contractible to a point any closed form is exact.}
\end{lemma}
This means all deRham cohomology groups of a manifold contractible to a point is trivial. Further results on deRham cohomology groups are as follows \cite{LeeISM},\cite{nakahara}:
\begin{itemize}
\item $\derham^0(M) \cong \mathbb{R}$ if $M$ is connected.
\item $\derham^0(M) \cong \mathbb{R} \oplus ... \oplus \mathbb{R}$ if $M$ has n-connected pieces.
\item $\derham^1(\mathbb{R})=0$
\item for $n \geq 1$ 
\begin{equation}
\derhamp(S^n) \cong \left\{
                \begin{array}{ll}
                  \mathbb{R} \quad \quad \mbox{if} \quad p=0 \quad \mbox{or} \quad p=n\ \sgc \
                  0 \quad \quad \mbox{if} \quad 0<p<n \sgd
                \end{array}
              \right.
\end{equation}
\end{itemize} 

\section{Flows and Lie Derivative}
\subsection{Integral Curves and Flows} 
\begin{defn}[Integral Curve of a Vector Field]
Let $X$ be a smooth vector field on a smooth manifold $M$. Integral curve of $X$ is a smooth curve $ \gamma: J \ra M$, where $J$ is an open subset of $\R$, such that 
\begin{equation}
\gamma'(t_0) \equiv \gamma_{*} \left( \left. \frac{d}{dt} \right|_{t_0}  \right)=X_{\gamma(t_0)} \sgc
\end{equation}
where $X_{\gamma(t)}$ is the vector field evaluated at the point $\gamma(t)$-which is a tangent vector.
\end{defn}
This is a first order differential equation, thus given an initial condition, i.e. an initial point on manifold $\gamma$ maps to, there is a unique solution; thus integral curves of a vector field do not cross. But there is no guarantee that solution exists everywhere on manifold, thus finding a global integral curve may not always be possible. A vector field is called complete if it has a global integral curve.

A collection of integral curves are often called a congruence. We will now define `` flow", which is a particular congruence of curves.

\begin{defn}[Flow]
A \hypertarget{flow}{flow} is a continuous map $\sigma: \R \times M \ra M$ such that
\begin{enumerate}
\item $\sigma(t,\sigma(s,p))=\sigma(t+s,p)$ ,
\item $\sigma(0,p)=p$ .
\end{enumerate}
\end{defn}
Every global flow, also called the one parameter group action is derived from the integral curves of some vector field, i.e. let $\sigma$ be a smooth global flow, and let $V$ the vector field obtained by setting $V_p$ to be the tangent vector of $\sigma(p,t)$, then $V$ will be a smooth vector field. But a flow may not be defined everywhere on the manifold for every vector field. Following can be defined by using flows:
\begin{enumerate}
\item Fixing $t$ define $\sigma_t(p) \equiv \sigma(t,p) : M \ra M$, then $\sigma_t$ is a homeomorphism; if $\sigma$ is smooth then $\sigma_t$ is a diffeomorphism.
\item Fixing $p$ define $\sigma^p(t) \equiv \sigma(t,p) : \R \ra M$, then $\sigma^p$ is an integral curve starting at $p$.
\end{enumerate}

Let $\sigma: \R \times M \ra M$ be a smooth global flow. One can define a tangent vector at each point $p \in M$ by $V_p = {\sigma^{p}}'(0)$. The assignment $p \ra V_p$ is a vector field, called the infinitesimal generator of $\sigma$.

\begin{example}[Active and Passive Transformations, Diffeomorphisms, Infinitesimal Diffeomorphisms]
Let's consider a manifold M with coordinate chart $\phi(p)=(t,x,y,z)$, and let $F: M \rightarrow N$ be a diffeomorphism. Also let $\psi(p')=(\bar{t},\bar{x},\bar{y},\bar{z})$ be a coordinate chart on N so that
\begin{equation}
\psi \circ F \circ \phi^{-1} (t,x,y,z) = (\bar{t}_F(t,x,y,z),\bar{x}_F(t,x,y,z),\bar{y}_F(t,x,y,z),\bar{z}_F(t,x,y,z)) \sgc
\end{equation}
Note that this illustrates what is known as active and passive transformations in physics: One can think at the manifold level and say what we are doing is a diffeomorphic mapping $F$ from one manifold to the other(active) or at the chart level and say we are doing a coordinate transformation $\psi \circ F \circ \phi^{-1}$ (passive). \\
\hfill\\
Let us now pick the coordinate charts such that an identity diffeomorphism does not change the values of coordinates, i.e.
\begin{align}
\psi \circ \1 \circ \phi^{-1} (t,x,y,z) &= (\bar{t}_{\1}(t,x,y,z),\bar{x}_{\1}(t,x,y,z),\bar{y}_{\1}(t,x,y,z),\bar{z}_{\1}(t,x,y,z)) \nonumber \\ &=(t,x,y,z) \sgd
\end{align}
Now I will take $F(p)= \sigma_{\beta}(p) \quad \forall p \in M$ for a given, where $\sigma$ is a flow. Note that $\sigma_{\beta}$ is a one-parameter set of diffeomorphisms. Then
\begin{equation}
\lim_{\beta \ra 0} (\bar{t},\bar{x},\bar{y},\bar{z})=(t,x,y,z)+ \beta \left. \frac{d \sigma_p(\beta)}{d \beta} \right|_{ \beta =0}= (t,x,y,z)+ \beta V_p \sgd
\end{equation}
For example consider the transformation which is known as boost in physics:
\begin{equation}
\psi \circ \sigma_{\beta} \circ \phi^{-1}(t,x,y,z)= ( \cosh(\beta) t + \sinh(\beta) x, \cosh (\beta) x + \sinh( \beta) t, y, z ) \sgd 
\end{equation}
Note that $\beta=0$ corresponds to identity transformation. Corresponding vector field that generates this transformation is then
\begin{align}
\left. \frac{d \sigma_p(\beta)}{d \beta} \right|_{ \beta =0} &= \left. ( \sinh(\beta) t + \cosh(\beta) x, \sinh (\beta) x + \cosh(\beta) t, y, z ) \right|_{ \beta =0} \nonumber \\ &= (x,t,0,0)\sgc
\end{align}
i.e.
\begin{equation}
V= x \frac{\partial}{\partial t} + t \frac{\partial}{\partial x} \sgd
\end{equation}
One can also start with the vector field on $M$ and get a complete flow defined on $M$.
\begin{align}
\sigma_{(t_0,x_0,y_0,z_0)} (\beta ) &= (t(\beta),x(\beta),y(\beta),z(\beta)) \sgc \\
\left. \frac{d \sigma_p(\beta)}{d \beta} \right|_{ \beta =0} &= (t'(\beta),x'(\beta),y'(\beta),z'(\beta))=(x,t,0,0) \sgc
\end{align}
solving $t'=x$ and $x'=t$ with the initial conditions $t(0)=t_0$ and $x(0)=x_0$ will give us back the full flow given above. 
\end{example}

\subsection{Lie Derivatives}
Recall that a vector field $X$ acting on a smooth function $f$ on a manifold $M$ is a derivation at all points on the manifold. One can express this action in terms of the flow $\sigma_X$ of $X$:
\begin{equation}
X(f)= \sigma_* \left( \frac{d}{dt} \right) f = \frac{d}{dt}\left( f \circ \sigma_X \right)= \frac{f(\sigma_X(t+dt))-f(\sigma_X(t))}{dt} \sgd
\end{equation}
This gives us a more familiar expression of the derivation of f. Note that it is also possible to get a similar expression by writing $X$ as a pushforward of $\frac{\partial}{\partial x_i}$.

To take the ``derivative" of a vector field, we can try acting on it by another vector field. i.e. Let us try to define $D_X$ acting on $\mathfrak{X}(M)$ such that for $ Y \in \mathfrak{X}$,
\begin{equation}
D_X Y(f)= X(Y(f)) \sgd
\end{equation}
However $D_X Y$ does not satisfy the rules of derivation:
\begin{equation}
D_X Y(fg)= X(Y(fg))=X(f)Y(g)+X(g)Y(f)+ g XY(f)+f XY(g) \sgd
\end{equation}
If we instead define
\begin{equation}
\lie_X Y(f)= X(Y(f))- Y(X(f)) \sgc
\end{equation}
then $\lie_X Y$ satisfies the rules of derivation and thus an element of $\mathfrak{X}(M)$ if $X,Y \in \mathfrak{X}(M)$.\\
\iffalse
To express this in terms of flows , let $\chi(t,p)$ be the flow of $X$ and $\gamma(s,p)$ be the flow of $Y$ and note that:
\begin{equation}
\chi^{q}_{*} \left( \left. \frac{d}{dt} \right|_{t} \right) = X_{\chi^q(t)}
\end{equation}
and if we let $\chi^q(t)=p$ then $q=\chi^p(-t)$.Thus \expl{Check this}
\begin{align*}
\left( \lie_X Y \right)_p f &= X_p \left( Y(f) \right) - Y_p \left( X(f) \right) \\
&= \chi^p(-t)_{*} \left( \frac{d}{dt} \right) Y(f) - \gamma^p(-s)_{*} \left( \frac{d}{ds} \right) X(f) \\
&= \frac{d}{dt} \left( Y(f)\circ \chi^p(-t) \right) -  \frac{d}{ds} \left( X(f)\circ \gamma^p(-s) \right)  \\
&= \frac{d}{dt} \left( Y_{\chi^p(-t)}(f) \right) -  \frac{d}{ds} \left( X_{\gamma^p(-s)}(f) \right)  \\
&= \frac{d}{dt} \left( \gamma^{\chi^p(-t)}_{*}(-s) \left(\frac{d}{ds} \right) f \right) - \frac{d}{ds} \left( \chi^{\gamma^p(-s)}_{*}(-t) \left(\frac{d}{dt} \right) f \right)  \\
&= \frac{d}{dt} \frac{d}{ds} \left( f\left( \gamma^{\chi^p(-t)}(-s)-\chi^{\gamma^p(-s)}(-t) \right) \right)  \\
&= \frac{d}{dt} \frac{d}{ds}  f\left( \gamma \left( \chi(p,-t),-s\right)- \chi \left( \gamma(p,-s),-t \right) \right) \sgd\\
\end{align*}
\fi
One can also see that an equivalent definition is
\begin{equation}
\left. \lie_X Y \right|_p = \lim_{t \ra 0}{\frac{(\chi_{-t})_{*} Y_{\chi_{t}(p)} - Y_p}{t}} \sgd
\end{equation}
\paragraph{Lie Derivative of Tensors:} Above definition can be generalized to arbitrary (k,l) tensors:
\begin{equation}
\left. \lie_X A \right|_p = \lim_{t \ra 0}{\frac{(\chi_{t})^{*} A_{\chi_{t}(p)} - A_p}{t}} \sgd
\end{equation}
Some properties of the Lie derivative are as follows: 
\begin{enumerate}
\item It is linear in second argument over $\R$ i.e. $\lie_X(aY+bZ)= a \lie_X(Y) + b \lie_Y(Z)$.
\item It is not linear in second argument over $C^{\infty}(M)$ i.e. $\lie_X (f Y)= X(f) Y + f \lie_X(Y)$.
\item It is linear in first argument over $\R$ i.e. $\lie_{aX+bY}(Z)= a \lie_X(Z) + b \lie_Y(Z)$.
\item It is \hypertarget{Lie nonlinear in first}{not linear in first argument} over $C^{\infty}(M)$ i.e. $\lie_{fX} (Y)= -Y(f) X + f \lie_X(Y)$.
\item Let V be a smooth vector field on $M$ and $ \kappa, \eta $ some arbitrary forms on $M$. Then
\begin{equation}
\mathcal{L}_{V}( \kappa \wedge \eta ) = \mathcal{L}_{V}( \kappa) \wedge \eta +  \kappa \wedge \mathcal{L}_{V}( \eta ) \sgd
\end{equation}

\end{enumerate}

\begin{thm}[Cartan's Magic Formula] \hypertarget{cartan}
On a smooth manifold $M$, for any smooth vector field $V$ and any smooth differential form $\alpha$,
\begin{equation}
\mathcal{L}_V \alpha = V \lrcorner \ d\alpha + d(V \lrcorner \ \alpha) \sgd
\end{equation}
\end{thm}

\section{Orientation and Integration}
\subsection{Orientation}
Consider a real vector space V with two ordered bases $\lbr {E}_i \rbr$ and $\{ \tilde{E}_i \}$. These two bases are called consistently oriented if the transition matrix $B$ defined as
\begin{equation}
E_i= B^j_i \tilde{E}_j
\end{equation}
between them has positive determinant. A choice of equivalence classes of consistently oriented bases are called an orientation. Note that by definition only two different choice of orientation is possible for a vector space.
\begin{defn}[Orientation of a Vector space]
A vector space together with a choice of orientation is called an oriented vector space.
\end{defn}
To define an orientation for a manifold one simply chooses a set of ordered vector fields such that at each point they have a positive orientation.
\begin{prop}
Any non-vanishing n-form $\omega$ on $M$ determines a unique orientation of $M$ for which $\omega$ is positively at each point.
\end{prop}
A manifold for which a choice of orientation cannot be made is called a non-orientable manifold. Following facts are useful:
\begin{itemize}
\item A smooth manifold is called parallelizable, if there exists a global frame on it. Every parallelizable smooth manifold is orientable.
\item Every Lie group-to be defined later- has precisely two left-invariant orientations, corresponding to the two orientations of its Lie algebra.
\end{itemize}

\subsection{Integration}
Only meaningful integration on the manifold is integration over n-forms, if manifold is n dimensional. We begin by introducing line integrals, then we define multidimensional integrals over n-forms.
Covectors give a coordinate independent notion of line integrals as explained in the following.
\paragraph{Line Integrals:} Consider a 1-dimensional manifold $M=\lb a,b \rb$, and let $t$ be coordinates on it. Then a covector on $M$ can be written as $\omega= f(t) dt$, and an integral of $\omega$ can be defined as
\begin{equation}
\int_{\lb a,b \rb} \omega \equiv \int_a^b f(t) dt \sgd
\end{equation} 
\begin{prop}
Let $\phi: \lb c,d \rb \ra \lb a,b \rb$ be an increasing diffeomorphism. Then
\begin{equation}
\int_{\lb c,d \rb} \phi^* \omega = \int_{\lb a,b \rb} \omega \sgd
\end{equation}
\end{prop}
\begin{defn}
Let $\gamma: \lb a,b \rb \ra M$ be a smooth curve in  $M$. Then line integral of a smooth covector field $\omega$ over $\gamma$ is defined as
\begin{equation}
\int_{\gamma} \omega \equiv \int_{\lb a,b \rb} \gamma^* \omega \sgc
\end{equation}
which also is equal to
\begin{equation}
= \int_a^b \omega_{\gamma(t)} \lp \gamma'(t) \rp dt \sgd
\end{equation}
\end{defn}

\begin{thm}[Fundamental Theorem of Line Integrals]
Let $f$ be a smooth function and $\gamma:\lb a,b \rb \ra M$ be a curve on $M$. Then 
\begin{equation}
\int_{\gamma} df=f(\gamma(b))-f(\gamma(a)) \sgd
\end{equation}
\end{thm}
Note that from this theorem we see, integral of an exact one-form over a closed curve is zero. In general a one-form that has an integral over a closed curve that is zero is called a conservative one-form. Next theorem tells us this is an iff statement.
\begin{thm}
A smooth one-form on $M$ is conservative if and only if it is exact.
\end{thm}
This is the familiar law in classical mechanics that, work done by a force $\vec{F}$ is path independent, if and only if $\vec{F}=\vec{\nabla} V$, where $V$ denotes the potential.

Thus far we were only able to define an integral of a one-form over a curve on a manifold in a coordinate invariant way. The reason for this was that the one form defines a measure of length for a vector at a point,and does this through the curve: it takes a vector and produces a number. Thus integral over a manifold of n-dimensions can only be done in a coordinate invariant way if a type (0,n) tensor is given. We want these integrals to define volume accurately.

Consider for simplicity the manifold $\R^n$. Note that the property of multilinearity, these integrals over (0,n) tensors on $\R^n$ will satisfy the property that if a vector is scaled by a number, then the volume is scaled by the same number; also if two vectors are added, resulting volume will be the addition of corresponding two volumes.

However (0,n) tensors will lack one important property that an n-volume should satisfy: if the set of n vectors acted by the tensor is linearly dependent then the volume should be zero. Thus instead of (0,n) tensors, n-forms should be used. So a n-dimensional integral over o manifold can only be over n-forms.

The generalization of n-dimensional integrals to arbitrary manifolds is made as follows:
\begin{defn}[Integral over a manifold]
Integral of an n-form over a region U of an n-dimensional manifold is defined as
\begin{equation}
\int_U \omega \equiv= \pm \int_{\phi_i(U)} (\phi^{-1})^* \omega \sgc
\end{equation}
where $+/-$ is for a positively/negatively oriented chart.
\end{defn}

\begin{thm}[Stokes' Theorem]
Let $M$ be a smooth manifold of dimension n, and $\pr M$ its boundary with the map $i: \pr M \ra  M$, then
\begin{equation}
\int_{M}{ d\alpha} = \int_{ \partial M }{i^* \alpha} \sgc
\end{equation}
where $\alpha$ is n-1 form.
\end{thm}
Note that this is the analogue of fundamental theorem of line integrals for n-dimensional manifolds.

\section{Submanifolds and Foliations}\label{sec:appfol}
\subsection{Submanifolds}
Let $F: M \ra N$ be a smooth map, rank of $F$ at a $p \in M$ is the rank of the linear map $\left( F_{*} \right)_p : T_pM \ra T_{F(p)}N$. If $F$ has the same rank r at every point on $M$, then we say $F$ has a constant rank r. An F that has a full rank everywhere is said to be:
\begin{itemize}
\item a submersion if $ rank(F)=dim(N) $,
\item an immersion if $ rank(F)=dim(M) $,
\item an embedding of $M$ into $N$ if $F$ is a smooth immersion that is also a topological embedding.
\end{itemize}

\begin{defn}[Embedded Submanifold]
An embedded submanifold of $M$ is a subset $S \subseteq M$ that is a manifold in the subspace topology, endowed with a smooth structure with respect to which the inclusion map $i:S \ra M$ is a smooth embedding. 
\end{defn}
\uline{For the purposes of this document, when we say submanifold we always mean an embedded submanifold unless otherwise stated.} A submanifold $S$ is called an hypersurface if it has a dimension 1 less than of the ambient manifold $M$.

\begin{prop}[Tangent Space of a Submanifold] \label{prop:tansubm}
Suppose $M$ is a smooth manifold with or without boundary, $S \subseteq M$ an embedded submanifold, and $p \in S$. Following are two ways of characterizing the tangent space of a submanifold:
\begin{itemize}
\item A \hypertarget{vfonsubman}{vector} $v \in T_pM$ is in $T_pS$ if and only if there is a smooth curve $\gamma:\R \ra M$ whose image is contained in S, and which is also smooth as a map into S, such that $\gamma(0)=p$, and $\gamma'(0)=v$.
\item $T_pS= \left\lbrace v \in T_pM: vf=0 \ \textit{whenever} \ f \in C^{\infty}(M) \ \textit{and} \ \left. f \right|_S=0 \right\rbrace$ .
\end{itemize}
\end{prop}
\noindent
Similarly if $X$ is a smooth vector field on $M$ then $X$ is a vector field of $S$ if and only if $ \left. (Xf) \right|_S =0 $ for every $\ f \in C^{\infty}(M)$ such that $\left. f \right|_S=0 $.

\paragraph{Vector fields along a submanifold:}
Given an embedded submanifold $S \subseteq M$, a section of the ambient tangent bundle $\left. TM \right|_S $ is called a \hypertarget{vector field along}{vector field along} S. It is a map $X: S \ra TM$ such that $X_p \in T_pM$ for each $p \in S$. Note that this is different from a vector field on $S$, where $X_p \in T_pS$ at each point.

\subsection{Orientation of a Submanifold} 
Let $M$ be an oriented n dimensional manifold, and $S$ be a hypersurface in it. Let $N$ be a vector field along $S$ that is nowhere tangent to $S$. Then one can give $S$ a unique orientation such that at every point $p$ on $S$ via this vector field $N$ such that: $\lbr E_1,...E_{n-1}\rbr $ is an oriented basis for $T_pS$ if and only if $\lbr N_p,E_1,...E_{n-1} \rbr$ is an oriented basis for $T_pM$. If $\omega$ is an orientation form for $M$ then inclusion of $N \inc \omega$ is an orientation form for S.

\begin{example}
Consider $S^2$ embedded in $\R^3$:  
\begin{equation}
S^2=\lbr (x,y,z) \st x^2+y^2+z^2 =1 \rbr \sgd
\end{equation}
Normal vector
\begin{equation}
N =r \frac{\partial}{\partial r} = x \frac{\partial}{\partial x}+y \frac{\pr}{\pr y}+ z \frac{\pr}{\pr z} \\
\end{equation} 
is a vector field along $S^2$ that is nowhere tangent to it. Right handed choice of orientation for $R^3$ is provided by the orientation form
\begin{equation}
\omega = dx \wedge dy \wedge dz \sgc
\end{equation}
then
\begin{equation}
N \inc \omega = x dy \wedge dz + y dz \wedge dx + z dx \wedge dy
\end{equation}
is a choice of orientation for $S^2$. Going to the spherical coordinates you get
\begin{equation}
N \inc \omega = \sin \theta d\phi \wedge d\theta \sgc
\end{equation} 
the standard choice of orientation for the sphere.
\end{example}

\begin{remark}[An application of the Stokes' Theorem]
If $S \subseteq M$ is a oriented compact smooth k-dimensional submanifold with or without boundary, and $\omega$ is a closed k-form on $M$ and $\int_S \phi^*\omega \neq 0$, where $\phi$ is the embedding of $S$ in $M$ then
\begin{enumerate}
\item $\omega$ is not exact on $M$,
\item $S$ is not the boundary of an oriented compact smooth submanifold with boundary
in $M$.
\end{enumerate}
\end{remark}
\subsection{Distribution and Integral Manifold}

\begin{defn}[Distribution]
A smooth distribution $D$ on $M$ of rank-k is a rank-k subbundle of $TM$, i.e. it is the smooth disjoint union of set of linear subspaces $D_p \in T_pM$ of dimension k over all points $p \in M$. Smoothness is ensured if around all points on $M$ there exists a coordinate neighborhood such that there exists k vector fields that form a basis for all points in that neighborhood.
\end{defn}

\begin{defn} [Integral Manifold]
A non-empty immersed submanifold $N \subset M$ is called an \hypertarget{integman}{integral manifold} of $D$ if $T_pN=D_p$ for all $p \in N$. 
\end{defn}
\begin{example}
A no-where vanishing smooth vector field $V$ is a smooth rank-1 distribution and image of any integral curve of $V$ is an integral manifold of $V$.
\end{example}
\begin{defn}[Integrable Distribution]
$D$ is called an \hypertarget{integrabledistribution}{integrable distribution} if each point of $M$ is contained in an integral manifold of $D$.
\end{defn}
\begin{defn}[Involutive Distribution]
If for any vector fields that $X,Y$ that are sections of $D$, $[X,Y]$ is also a section of $D$ then D is called an involutive distribution. If $D$ is involutive then space of smooth global vector fields of $D$ is a Lie subalgebra of the space of smooth global vector fields of $M$, $\vfs(M)$.
\end{defn}

\begin{prop}
Every integrable distribution is involutive.
\end{prop}
\begin{proof}
Let $X,Y$ be smooth vector fields of $D$, and let $N$ be an integral manifold of $D$ containing $p$. Then $X,Y$ are tangent to $N$, so $[X,Y]$ is tangent to $N$.
\end{proof}

An alternative definition of a smooth distribution is given via forms as follows:
\begin{prop}[Defining Forms]
D is a smooth distribution of rank-k if and only if at each point $p \in M$ has a neighborhood $U$ on which there are smooth 1-forms $\omega_1, ...,\omega_{n-k}$ such that for each $q \in U$,
\begin{equation}
D_q = \left. Ker(\omega_1) \right|_q \cap ... \cap \left. Ker(\omega_{n-k}) \right|_q \sgd
\end{equation}
\end{prop}
\begin{proof}
Let $ Y_1, ....,Y_k$ span $D$. Complete these to a smooth local frame $ Y_1,...,Y_n $ for $M$. Then there exists a smooth dual coframe $ \omega_1,...,\omega_n$ such that $ \omega_i(Y_j)=\delta_i^j$, then $Y_1,...,Y_k$ will be in $ Ker(\omega_{k+1}) \cap ... \cap Ker(\omega)_{n}$. 
\end{proof}
Such set of $\omega_1, ...,\omega_{n-k}$ are called local defining forms for $D$.

A p-form $\eta$ is said to annihilate $D$ if $\eta(X_1,...,X_p)=0$ whenever $X_1,...,X_p$ are local sections of $D$; and this is the case if and only if $\eta$ can be written in the form
\begin{equation}
\eta =\sum_{i=1}^{n-k}{\omega_i \wedge \beta_i}
\end{equation}
where $\beta_1,...,\beta_{n-k}$ are some smooth (p-1) forms, and $\lbr \omega_i \rbr$ are defining forms.

\begin{prop}\label{prop:annihform}
$D$ is involutive if and only if for any smooth 1-form $\eta$ that annihilates $D$ on a  $U \subset M$, $d \eta$ also annihilates $D$ on $U$.
\end{prop}

\begin{proof}
Let $D$ be involutive, and $ \eta$ annihilates $D$. Then for any smooth sections $X,Y$ of $D$ 
\begin{equation}
d \eta (X,Y) = X(\eta(Y))-Y(\eta(X))- \eta( [X,Y]))=0 \sgd
\end{equation}
Conversely let $\omega_1, ...,\omega_{n-k}$ be local defining forms for $D$, then for each i
\begin{equation}
\omega_i( [X,Y])) = X(\omega_i(Y))-Y(\eta(X))- d \omega_i (X,Y)=0
\end{equation}
thus $D$ involutive.
\end{proof}
A coordinate chart $(U,\phi)$ on $M$ is called flat for $D$ if $\phi(U)$ is a cube in $\R^n$; and at a point that belong to $U$, $D$ is spanned by,lets say, first k coordinate vector fields $\frac{\partial}{\partial x_1}...\frac{\partial}{\partial x_k}$. In any such chart, each slice of the form $x^{k+1}=c_{k+1} \& ...\& x^{n}=c_{n}$ for constants $c_{k+1}... c_n$, is an integral manifold of $D$. We say $D$ is completely integrable, if such chart exists for $D$ in a neighborhood of each point of $M$. Note that every completely integrable distribution is integrable, thus involutive; the theorem we next give without proof shows this relation is an if and only if statement, i.e.

\begin{thm}[Frobenius Theorem]
\hypertarget{frob}{Every involutive distribution is completely integrable}.
\end{thm}

\chapter{Lie Groups} \label{ch:applie}
In this chapter we discuss a special type of differentiable manifold: manifolds with group structure. Concept of a group is used to describe symmetries in physics. In the first section we define what a group is and discuss properties of finite groups. Then we move onto Lie groups which has infinite number of elements i.e. groups that have a manifold structure. In gauge theories the solutions that are related to each other via gauge transformations are called physically equivalent. Gauge symmetries are considered as redundancies in the description, and thus sometimes it will be desirable to get rid of them. The procedure to define a space where equivalent elements are represented by a single quantity is called quotienting. We will see how this procedure works for finite and Lie groups. By virtue of group actions, which we define in the following, one can also quotient manifolds with a group action on them, and if the action has appropriate properties the resulting space will be a smooth manifold, as discussed in Section \ref{sec:apphom}. We conclude with homogeneous spaces, the manifolds with a transitive action on them.

\section{Groups}
\subsection{Introduction}
\begin{defn}
A group $ (\G,.) $ is a set $ \G $ with a multiplication law .\,: $ \G \times \G \ra \G $ such that
\begin{enumerate}
\item $\exists$ an identity element of . , called $ \e $;
\item $\exists$ inverse of each element;
\item . is associative.
\end{enumerate}
\end{defn}

\begin{defn}
A group homomorphism $\alpha$ is a map $ (\G_1,.) \ra (\G_2,*) $ such that
\begin{equation}
\alpha( g_1 . g_2)=\alpha(g_1)*\alpha(g_2)  \quad \textit{for} \quad g_1,g_2 \in \G
\end{equation}
a bijective homomorphism called an isomorphism.
\end{defn}

\begin{defn}[Subgroup]
A subgroup $H$ of $\G$ is a subset of $\G$ that is a group under the same multiplication  ``$.$".

\end{defn}

\subsection{Cosets and Normal Subgroups}

\begin{defn}[Conjugacy of group elements]
$g_1 \in \G$ is conjugate to $g_2 \in \G$ if there exists a $h \in \G$ such that $g_1=h g_2 h^{-1}$.Note that conjugacy is an \hyperlink{equivalence}{equivalence relation}.
\end{defn}

\begin{defn}[Coset]
If $H$ is a subgroup of $\G$ it's (left) coset with respect to $g \in \G$ is
\begin{equation}
gH= \left\{ gh_1,...,gh_m \right\} \sgd
\end{equation}
\end{defn}
(Left) Cosets define an equivalence relation, they either completely overlap or are completely disjoint; i.e. they partition $\G$. The set $g_1 H, ...,g_k H$ of all distinct cosets of $\G$ is denoted by $\G/H$. Note that:
\begin{enumerate}
\item A coset of a subgroup might not be a subgroup itself.
\item $\G/H$ might not be a group itself.
\item Left coset partitioning might be different than right coset partitioning.
\end{enumerate}

\begin{defn}[Normal Subgroup]
A subgroup whose left cosets equal to its right cosets is called a normal subgroup, i.e. a subgroup $H$ is normal if $gHg^{-1}=H$ for all $g \in \G$. This means that H is made up of complete conjugacy classes i.e. if $a \in H$ then $gag^{-1} \in H$ for all $g \in \G$. Note that this is not necessarily true for a coset of any subgroup.
\end{defn}

\begin{thm}
If $H$ is a normal subgroup of $\G$ then $\G/H$ is a group. This group is called a quotient group.
\end{thm}

\begin{thm}[First Isomorphism Theorem]
Let $f$ be an \hypertarget{firstiso}{homomorphism} on $\G$ then, $Ker(f)$ is a normal subgroup and $\G/Ker(f)$ is isomorphic to $Im(f)$, i.e. $\G/Ker(f) \cong Im(f)$.
\end{thm}

\begin{defn}[Center of a group]
The set
\begin{equation}
Z(\G) \equiv \lbr g \in \G : gh=hg \quad \forall h \in \G \rbr
\end{equation}
is called the center of a group. Note that it is a normal subgroup of $\G$.
\end{defn}
\subsection{Representations}
Representations are concrete realizations of groups, more specifically: a representation $T$ of a group $\G$ is a homomorphism to $GL(V,F)$, where $V$ is a vector space over the field $F$.
\begin{itemize}
\item Two representations $T_1$ and $T_1$ are equivalent if there exists a $S \in GL(V,F)$ such that $T_1(g)=ST_2(g)S^{-1}$ for all $g \in \G$.
\item A representation is called faithful if it is injective.
\item dim of representation=dim(V).
\end{itemize}

\begin{defn}
Unitary representation is a representation on the Hilbert space such that $T(g)^{\dagger}=T(g^{-1})=T(g)^{-1}$.
\end{defn}
Every representation of a finite group can be made into a unitary representation, By redefining the inner product.

\begin{defn}
A representation $T$ of a $\G$ on $V$ called reducible iff there exists a nontrivial linear subspace $W \subset V$ such that $T(g)W \subseteq W$ for all $g \in \G$, otherwise $T$ is called irreducible.
\end{defn}

\begin{thm}[Shur's Lemma]
Let $T$ and $U$ be two irreducible representations of a group on vector spaces $V$ and $W$ respectively. If there exists an operator $S: V \ra W$ such that $S T = U S$, then either $S=0$ or $S$ is invertible.
\end{thm}
For $T=U$ (and thus $V=W$), this will give us that for an irreducible representation $T$ if there exist an operator S such that it commutes with all of matrices of $T$ then $S=c \mathbb{1}$.\\
Using Shur's Lemma one can also show that every unitary representation is completely reducible i.e. can be written as direct sum of irreducible representations. Hence every representation of a finite group is also completely reducible.

\begin{thm}[Fundamental Orthogonality Theorem]
 Take two irreducible representations $T_{\alpha}$,$T_{\beta}$ of a group $\G$ then
\begin{equation}
\sum_{g \in \G}{T_{\alpha}(g)_{ik} T_{\beta}(g^{-1})_{lj}} = \frac{n_{\G}}{dim(T_{\alpha})} \delta_{\alpha \beta} \delta_{ij} \delta_{kl}
\end{equation}
where $n_{\G}$ is the number of elements in $\G$.
\end{thm}
\begin{cor}
Number of inequivalent irreducible representations are finite.
\end{cor}

\iffalse
\begin{defn}
The character $\chi_T$ for a representation $T$ is a map $\chi_T: \G \ra \C$ such that $\chi_T(g)=Tr \left( T(g) \right)$
\end{defn}
\begin{itemize}
\item If two representations are similar iff their characters are the same.
\item Characters are class functions i.e. $\chi_T(g_1)=\chi_T(g_2)$ if $g_1=h g_2 h^{-1}$.
\item $\chi_T(e)=dim T$.
\item Using fundamental orthogonality theorem:
\begin{equation}
\frac{1}{n_{\G}} \sum_{g \in \G}{\chi_{\alpha}(g) \chi_{\beta}^{*}(g)}= \delta_{\alpha \beta}
\end{equation}
\item $ \sum_{ \alpha}{(dim T_{\alpha})^2}= n_{\G}$.
\end{itemize}

\begin{defn}
(Left) Regular Representation is defined on $V=\lbr f | f: \G \ra \C \rbr $ such that $ \left( T(g)f \right) (h)=f(g^{-1}h)$. It is helpful to denote elements of $V$ as $ \left|f \right>$, and group elements as $ \left< g \right|$ so that $ \left< h \right| T(g) \left| f \right> = f(g^{-1}h)$. One can choose a basis $\left|h\right>$ such that $ \left< g \right| \left. h \right>= \delta_{hg}$.  
\end{defn}
\begin{itemize}
\item $\chi_R (g)= \sum_{h}{ \left< h \right| T(g) \left| h \right>} = \sum_{h}{h(g^{-1} h)} = \sum_{h}{\delta_{h,(g^{-1} h)}}= \delta_{\e g} n_{\G}$
\item $ T_R= \underset{ \alpha}{\oplus} d_{\alpha} T_{\alpha}$
\end{itemize}

\fi
\section{Lie Groups}
\subsection{Lie Groups and Action of Group on a Manifold} \label{subsec:appGrAct}
\begin{defn}
Let $\G$ be a group and a differentiable manifold without boundary, then $\G$ is a Lie group, if group multiplication and the inversion maps are smooth.
\end{defn}

An easy example of Lie Groups are set of invertible matrices of any dimension, called the matrix group, where the elements are square matrices, the group multiplication operation is matrix multiplication, and the group inverse is the matrix inverse. Set if invertible real matrices of dimension n is denoted by $\GL(n,\R)$.

Lie group homomorphism from a Lie group $\G$ to another Lie group $\mathcal{H}$ is a smooth map $F:\G \ra \mathcal{H}$ that is also a group homomorphism. It is called a Lie group isomorphism if it is also a diffeomorphism.

\begin{defn}[\cite{LeeISM} pg.156]
A Lie subgroup of a Lie group $\G$ is a subgroup of $\G$ endowed with a topology and smooth structure making it into a Lie group and an immersed submanifold of $\G$.

\end{defn}

\begin{defn}[Action of a Group on a Manifold]
Let $\G$ be a group with identity element $e$ and let $M$ be a set. (Left) \hypertarget{group action}{Action} of $\G$ on $M$ is a map $\theta: \G \times M \ra M$ : $(g,p) \mapsto g * p$ that satisfies
\begin{enumerate}
\item $g_1*(g_2*p)=(g_1.g_2)*p$ where ``." is the group multiplication of $\G$,
\item $e*p=p$.
\end{enumerate}
If $\G$ is a Lie group, $M$ a smooth manifold and $\theta$ smooth then action is called to be smooth.
\end{defn}
One can also define $\theta_g:M \ra M$ such that $\theta_g(x)=g*x$. Then $(\theta_g)^{-1}= \theta_{g^{-1}}$; thus if $\theta$ is smooth then $\theta_g$ is a diffeomorphism.

Note that a flow on a manifold, as defined in chapter on Differentiable Manifolds, also named as one-parameter group action is an action $\phi$ of the group $(\R,+)$ on $M$. 

If $M$ is also a Lie group $\phi$ is called the one-parameter subgroup of $M$.
\paragraph{Some properties of group action:} Let $\theta:\G \times M \ra M$:$(g,p) \ra g*p$ be an action of a group $\G$ on a set $M$, then 
\begin{enumerate}
\item Orbit of p= $\left\{ g*p | g \in \G \right\}$,
\item If orbit of p=p, then p is a fixed point of $\G$ on $M$,
\item If orbit of p=$M$ for all $p \in M$, then $\theta$ is transitive on $M$,
\item \hypertarget{isot}{Isotropy group} or stabilizer of $p \in M$ is $G_p = \left\{ g \in \G  | g*p=p \right\}$,
\item $\theta$ is said to be free if $G_p={e}$ for all $p \in M$.
\end{enumerate}

\begin{defn}[Maps between manifolds that are Equivariant under some Group action]
Let $\G$ be a group and $\theta$ and $\alpha$ its (left) action on $M$ and $N$ respectively. A map $F:M \ra N$ is said to be equivariant if
\begin{equation}
\alpha(g, F(p))= \theta(g,p)
\end{equation}
for all $g \in \G$, $p \in M$.
\end{defn}

\begin{defn}[Left Translation on Lie groups]
Let $L: \G \times \G \ra \G$ be an action of the Lie group $\G$ on itself so that $L(a,g)=a.g$ where . is the group multiplication of $\G$. Then $L_a$ is a diffeomorphism, called the left translation.
\end{defn}
\begin{enumerate}
\item Every Lie group acts smoothly, freely, transitively on itself by left translation.
\item Every Lie group acts smoothly on itself by conjugation: $\theta(g,h)=ghg^{-1}$.
\end{enumerate}
A right translation can also be defined the same way where $R(a,g)=g.a$ or more suggestively $R_a(g)=g.a$.

\begin{defn} [Left-invariant Vector Field]
A vector field $X$ on $\G$ is left-invariant if it is $L_a$ related to itself for all $a \in \G$; i.e.
\begin{equation}
\left. (L_a)_{*} X \right|_g = \left. X \right|_{a.g} \sgd
\end{equation}
\end{defn}

\subsection{Lie Algebra of a Lie Group}
Set of all smooth left-invariant vector fields is a linear subspace of $\mathfrak{X}(M)$, and it is closed under lie bracket i.e. if $X,Y$ are smooth left-invariant vector fields, then $[X,Y]$ also is; thus they form a \hyperlink{LieA}{Lie algebra} under the usual Lie bracket.
\begin{defn}[Lie Algebra of a Lie Groups]
Lie algebra of all smooth left-invariant vector fields of a Lie group $\G$ with the usual Lie bracket of vector fields is called the Lie algebra of $\G$ and denoted by $Lie(\G)$.
\end{defn}
Note that set of all smooth vector fields is also a Lie algebra, but what we call $Lie(\G)$ is a subset of this.\\
The Lie algebra of the group $\GL(n,\R)$ are set of $n \times n$ real matrices $M(n,\R)$, also denoted by $\gl$.
\begin{thm}
Let $\G$ be a Lie group. The map $\E: Lie(\G) \ra T_e\G$, given
by $\E(X)=X_e$, is a vector space isomorphism. Thus, $\dim(\Lie(\G))= \dim(\G)$. 
\end{thm}
\begin{proof}[A Skecth of the proof]
Show that 1) $\ker(\E)=0$ i.e. if $\E(X)=0$ then $X=0$ 2) for any $v \in T_e\G$ there exists a $V \in \Lie(\G)$ so that $V_e=v$, i.e. left-invariant vector
fields are uniquely determined by their values at the identity. 
\end{proof}

\subsection{Exponential Map}

\begin{defn}
A one-parameter subgroup of a Lie group $\G$ is a Lie group homomorphism $\gamma: \R \ra \G$ where $\R$ is a Lie group under addition. 
\end{defn}
\begin{thm}
The one-parameter subgroups of $\G$ are precisely the maximal integral curves of left-invariant vector fields starting at the identity.
\end{thm} 

\begin{defn}[Exponential Map]
Let $\G$ be a Lie group with Lie algebra $\Lie(\G)$, the exponential map of $\G$, $ \exp: \Lie(\G) \ra \G$, is defined by
\begin{equation}
\exp(X)=\gamma(1) \sgc
\end{equation}
where $X \in \Lie(\G)$,$\gamma$ is the one-parameter subgroup generated by $X$, or equivalently the integral curve of $X$ starting at the identity.
\end{defn}
For any $X \in \Lie(\G)$, $\tilde{\gamma}(s) \equiv \exp(sX)$ is the one parameter subgroup of $\G$ generated by $X$.

If $\phi: \G \ra H$ is a Lie group homomorphism, then $ \phi^*_e: \Lie(\G) \ra \Lie(H)$ is a Lie algebra homomorphism. Moreover
\begin{equation}
\phi(\exp(X))=\exp ( d\phi_e(X)) \sgd
\end{equation}

\iffalse
\subsection{Covering Space/Group}
\begin{itemize}
\item Let $ \G $, $K$ be two Lie groups that are homomorphic via $f$. Recall $ \G / Ker(f) \cong Im(f) $. An example of this is $ SU(2) / \mathbb{Z}_2 \cong SO(3)$. Also; $S^3=SU(2)$ while $SO(3)=RP^2$. $SU(2)$ is covering space of $SO(3)$. First homotopy group of $SU(2)$ is $ \left\lbrace e \right\rbrace$;First homotopy group of $SO(3)$ is $\mathbb{Z}_2$
\end{itemize}
\fi

\subsection{Haar Integral Measure on Lie Groups}
Let $\G$ be a compact Lie group endowed with a left-invariant orientation. Then $\G$ has a unique positively oriented left invariant n-form $\omega_{\G}$ with the property
\begin{equation}
\int_{\G} \omega_{\G}=1 \sgd
\end{equation}
Let $I: C(\G) \ra \R$, where $C(\G)$ is set of all continuous real functions on $\G$, such that
\begin{equation}
I(f) \equiv \int_{\G} f(g) \, \omega_{\G} \sgc
\end{equation}
then $I$ is left and right translation invariant, i.e.
\begin{equation}
I(f) = I(f \circ L_g)=I(R_g \circ f) \quad \forall g \in \G \sgd
\end{equation}
The map I is called the Haar integral.

\subsection{Adjoint Representation} \label{subsec:adrep}

The representation of a Lie Group is defined the same way with finite groups.

Given a representation $\phi: \G \ra \Aut(V)$ on a compact Lie group $\G$, there exists a $\G$ invariant inner product on V i.e. there exists a inner product $(.,.)$ where
\begin{equation}
(u,v) = \int_{\G} \ket \phi(g)u,\phi(g)v \bra \, \omega_{\G} \sgc
\end{equation}
such that 
\begin{equation}
( \phi(h) u, \phi(h) v )=(u,v) \quad \forall h \in \G.
\end{equation}
Here $\ket,\bra$ is an inner product on V.

Now define the conjugation map $C_g: \G \ra \G$ on a Lie group $\G$ such that $C_g(h)=ghg^{-1}$, i.e $C_g = R_{g^{-1}} \circ L_g$. Then $C_g$ is a Lie group homomorphism.
\begin{defn}[Adjoint Representation of a Lie Group and its Lie Algebra]
Let $\G$ be a Lie group, and $\g$ its Lie algebra. The map $\Ad: \G \ra \Aut(\g)$ , where 
\begin{equation}
\Ad(g)= {C_g}_*
\end{equation}
is called the adjoint representation of $\G$. The map $\ad:\g \ra \End(\g)$ ,where
\begin{equation}
\ad= \Ad_*
\end{equation}
is called the adjoint representation of $\g$.
\end{defn}

Several properties are at hand:
\begin{enumerate}
\item $\ad(X)Y=[X,Y]$, for $X,Y \in \g$. 
\item If $\G$ is a matrix group then $\Ad(g) X= g X g^{-1}$ , where on the right hand side $g$ and $X$ should be considered as matrices, and multiplied by matrix multiplication. 
\item $\ker(\Ad)= Z(\G)$, i.e. where $Z(\G)$ is the center of $\G$.
\item $\ker(\ad)= Z(\g)$, i.e. where $Z(\g)$ is the center of $\g$.
\end{enumerate}

\begin{defn}
\hypertarget{killing form}{Killing form} of a Lie algebra $\g$ is a map $B: \g \times \g \ra \R$, such that
\begin{equation}
B(X,Y)= \tr( \ad X \circ \ad Y)
\end{equation}
where $\circ$ is composition map of two endomorphisms. Properties of the Killing form are as follows:
\begin{enumerate}
\item $B$ is symmetric and bilinear.
\item $B(X,Y)=B( \Ad(g) X, \Ad(g) Y)$ i.e. Killing form is $\Ad$ invariant.
\item $B(\ad(Z)X,Y)=-B(X,\ad(Z)Y)$.
\item $\G$ is semisimple if $B$ is non-degenerate.
\item If $\G$ is semisimple then $Z(\g)=0$ and $Z(\G)$ is discrete.
\item If $\G$ is compact semisimple then $B$ is negative definite.
\end{enumerate}
\end{defn}

\section{Homogeneous Spaces}\label{sec:apphom}
Consider a group action $\theta: \G \times M \ra M$. Then one can define a relation $\sim$ on $M$ such that $p \sim q$ if they $p,q$ are in each other's orbits. This relation is an equivalence relation, thus defines a partition. The set of equivalence classes i.e. the set of orbits is denoted by $M/\G$ with the quotient topology (see \cite{LeeISM} for a definition). If the action $\theta$ satisfies certain conditions -smoothness,freeness,properness- that we will not define here, $M/\G$ is a smooth manifold.

This most general definition of quotient manifold will not be of interest to us, instead we will focus on a more special case, the homogeneous spaces. For this we generalize the notion of quotient for finite groups to Lie groups.

\begin{prop}
Let $\G$ be a Lie group and $K$ a closed subgroup of $\G$. Then there is a unique way to make $\G/K$ a manifold so that the projection $\pi: \G \ra \G/K$ is a submersion. Then $\G/K$ is called a coset manifold.
\end{prop}
An analogous definition is in order for manifolds on which a Lie group $\G$ acts transitively.

\begin{prop}\label{prop:homsp} \cite{arvan}
Let $\theta: \G \times M \ra M$ be a transitive action of a Lie group $\G$ on a manifold $M$. Let $K=\G_p$ be the isotropy subgroup of a point p. Then
\begin{enumerate}
\item The subgroup $K$ is a closed subgroup of $\G$,
\item Orbit of $p$ is diffeomorphic to $\G / K$,
\item The dimension of $\G/K$ is $\dim \G - \dim K$.
\end{enumerate}
\end{prop}

\begin{defn}
A homogeneous space is a manifold with a transitive action of a Lie group $\G$.
\end{defn}

\chapter{Riemannian Manifolds}\label{ch:apprie}
In this chapter we continue building extra structures on manifolds. In the last chapter we studied manifolds that has symmetry properties, i.e. a group structure, whereas in this one we will talk about manifolds that has a measure of length and angles is defined via what we call a metric. Einstein's theory tells us that spacetime is a pseudo-Riemannian metric, but use of study of Riemannian manifolds goes beyond that. Most theories can be visualized as taking place on a space of configurations with a metric on it, so that equations of motions turns out to be the geodesics with respect to that metric.

In the following we start by giving basic definitions: after the metric, most basic element on a Riemann manifold, we define connections, which can be independent of the metric but generally taken with such properties that there exists unique connection for a given metric. Concepts of curvature is then built by connection and metric.
\section{Definitions and Basic Elements}
\subsection{Metric on a Manifold}
\begin{defn}[Riemannian metric and manifold]
A Riemannian metric on a smooth manifold $M$ is a (0,2) tensor field g which satisfies the following:
\begin{enumerate}
\item $g_p(U,V)=g_p(V,U)$,
\item $g_p(U,U) \ge 0 $ where the equality holds only when $U=0$,
\end{enumerate}
on all points $p$ on $M$ where $U,V \in T_pM$. The manifold $M$ endowed with such a metric is called a Riemannian Manifold.
\end{defn}

Note that
\begin{enumerate}
\item Every smooth manifold admits a Riemannian metric. 
\item The Riemannian metric determines an inner product on each tangent space.
\item The Riemannian metric is invertible, since it's null space is just $\{0\}$.
\item For every tangent vector (covector) $X$ at a point $p$ on the manifold, one can define a corresponding covector (tangent vector) using the (inverse) metric. In the coordinate representation this is called lowering (raising) indices.
\end{enumerate} 

\begin{defn}[Isometry]
A diffeomorphism $\phi$ from a Riemannian manifold $M_1$ with metric $g_1$ to another Riemannian manifold $M_2$ with metric $g_2$ is called an isometry if $\phi^* g_2 = g_1$. Set of isometries of a metric forms a Lie group.
\end{defn}

\subsection{Connection,Geodesic,Curvature}
\begin{defn}[Connection] \label{def:appconn}
Connection is a map $ (X,Y) \ra \nabla_X Y$ satisfying the following:
\begin{enumerate}
\item It is linear in second argument over $\R$ i.e. $ \nabla_X(a Y+bZ)=a \nabla_X(Y)+ b \nabla_X(Z)$.
\item It is not linear in the second argument over $C^{\infty}(M)$ and satisfies $\nabla_X (fY)= f \nabla_X(Y)+ X(f)Y$. This second equation is called the product rule.
\item It is linear in first argument over $C^{\infty}(M)$ \hyperlink{Lie nonlinear in first}{(as opposed to Lie Derivative)} i.e. $\nabla_{f X} Y= f \nabla_X Y$
\end{enumerate}
\end{defn}
Let $\left\lbrace E_i \right\rbrace$ be a frame for $TM$ and define the Christoffel symbols $\Gamma_{ij}^{k}$ of $\nabla$ by
\begin{equation}
\nabla_{E_i}{E_j}=\Gamma_{ij}^{k} E_k \sgd
\end{equation}
Then
\begin{equation}
\nabla_X Y = \left( X(Y^k)+X^i Y^j \Gamma_{ij}^{k} \right) E_k= X^i \lp E_i(Y^k) + \Gamma_{ij}^k Y^j \rp E_k \sgd
\end{equation}

Some properties of the connection are as follows:
\begin{enumerate}
\item There is a one-to-one relation between $\nabla$ and $\Gamma$.
\item Every manifold admits a linear connection.
\item Because of the product rule, if $\nabla^1$ and $\nabla^2$ are two connections; $ \frac{1}{2} \nabla^1$ or $\nabla^1+\nabla^2$ are not connections.
\end{enumerate}

Definition of connection can be extended to arbitrary (k,l) type tensor fields such that, for type (1,0) it agrees with the given connection and
\begin{enumerate}
\item $\nabla_X f=X(f)$,
\item $\nabla_X (F \otimes G )= (\nabla_X F) \otimes G + F \otimes (\nabla_X G)$,
\item $\nabla_X $ commutes with any contraction of a pair of indices.
\end{enumerate} 
For any given connection for vector fields there will be unique generalization that satisfies the properties above. Any such generalized connection will also satisfy the following rule:
\begin{align} \label{eq:Cov distr over tensor}
 X \lp  F(\omega^1, ...,Y_1,...) \rp &= \nabla_X \lp F(\omega^1, ...,Y_1,...) \rp \nonumber \\
&= (\nabla_X F)(\omega^1, ...,Y_1,...) + \sum_{i=1} F(\omega^1, ...,\nabla_X \omega^i,...,Y_1,...) \nonumber \\
&+ \sum_{i=1} F(\omega^1, ...,Y_1,...,\nabla_X Y^i,...) \sgd
\end{align}

Connections alternatively can be considered as a map that takes (k,l) type tensor fields to (k, l+1) type tensor fields. This map is called the covariant derivative and defined as
\begin{equation}
\nabla F(\omega^1,...,\omega^l,Y_1,...,Y_k,X)= \nabla_X F(\omega^1,...,\omega^l,Y_1,...,Y_k) \sgd
\end{equation}
\begin{remark}
Note that due to the property $\nabla_{fX}Y= f \nabla_X Y$ of a connection,$\nabla F$ is a tensor field, and this is why we cannot define a similar tensor using Lie derivatives.
\end{remark}
\begin{remark}
Note that in a written in a basis $(\nabla_X Y)^{\nu}=X^{\mu}{(\nabla Y)_{\mu}}^{\nu}=X^{\mu} \nabla_{\mu} Y^\nu$.
\end{remark}
Covariant derivative of a vector field $V$ along a curve $\gamma$ is $\nabla_{\gamma'(t)}V$. $V$ is said to be parallel along $\gamma$ with respect to $\nabla$ if, $\nabla_{\gamma'(t)}V=0$. Given a tangent vector at a point on the curve there is a unique vector field along the curve obtained by parallel transporting this tangent vector.

\begin{defn}[Geodesic]
A curve $\gamma$ is called a geodesic with respect to $\nabla$, if 
\begin{equation}
\nabla_{\gamma'(t)}\gamma'(t)=0.
\end{equation}
\end{defn}
Any curve satisfying $\nabla_{\gamma'(t)}\gamma'(t)= \alpha(t) \gamma'(t) $ can be repramatrized to satisfy the geodesic equation. Such parameter called the affine parameter. Given a tangent vector $V_p$ at point $p$ on $M$ there exists a unique (maximal) geodesic $\gamma$ through $p$ such that $\gamma'(t)_p=V_p$; for a given connection.

\begin{defn}[Curvature]
Curvature $\rmt$ of a connection $\nabla$ is a map $ TM \otimes TM \otimes TM \ra TM $ given by
\begin{equation}
\rmt(X,Y)Z=[\nabla_X, \nabla_Y] Z-\nabla_{[X,Y]} Z \sgd
\end{equation}
\end{defn}
\noindent
Note that R is a type $(1,3)$ tensor. If we write it in the coordinate basis
\begin{equation}
\rmt={\rmt_{ijk}}^l dx^i \otimes dx^j \otimes dx^k \otimes \partial_l,
\end{equation}
so that the coefficients are given by
\begin{equation}
{\rmt_{ijk}}^l= dx^l \left( \rmt(\partial_i,\partial_j) \partial_k \right)= dx^l \left( [\nabla_{\partial_i},\nabla_{\partial_j}]\partial_k \right).
\end{equation}
Also a type (0,4) tensor version of the curvature is defined via
\begin{equation}
\rmt(X,Y,Z,W)=g( \rmt(X,Y) Z, W) \sgd
\end{equation}

\subsection{Riemannian Connection, Riemannian Geodesic, Riemannian Curvature}
\begin{defn}[Metric Compatible Connection]
 $\nabla$ is said to be compatible with g if it satisfies the following equivalent conditions:
\begin{enumerate}
\item $\nabla_X g(Y,Z)=g(\nabla_X Y,Z)+g(Y,\nabla_X Z)$.
\item $\nabla g=0$.
\item For any $V,W$ along any curve $\gamma$
\begin{equation}
\frac{d}{dt} g(V,W) = g(\nabla_{\gamma'(t)}V,W)+ g(V,\nabla_{\gamma'(t)}W) \sgd
\end{equation}
Note that this means for the metric compatible $\nabla$ if vector fields are parallel transported along any curve, their inner product will be preserved along that curve.
\end{enumerate}
\end{defn}
\begin{defn}
$\nabla$ is said to be symmetric/torsionless if
\begin{equation}
\nabla_X Y - \nabla_Y X=\left[X,Y\right] \sgd
\end{equation}
\end{defn}
\begin{thm}
There exists a unique linear connection on a Riemannian or (pseudo-Riemannian) Manifold $(M,g)$ that is compatible with metric and symmetric; called the \hypertarget{Rie Conn}{Riemannian connection}.
\end{thm}

\begin{remark}
A useful property follows from the definitions above: All Riemannian geodesics are constant speed curves:
\begin{equation}
\frac{d}{dt} g(\gamma'(t),\gamma'(t)) =  g(\nabla_{\gamma'(t)}\gamma'(t),\gamma'(t))+ g(\gamma'(t),\nabla_{\gamma'(t)}\gamma'(t))=0 \sgd
\end{equation}
\end{remark}

\begin{remark}[Naturality of Riemannian Connection] If there is an isometry between two Riemannian manifolds then this isometry takes Riemannian connection to Riemannian connection; Riemannian geodesic to Riemannian geodesic and so on; see \cite{LeeISM}.
\end{remark}

\subsection{Killing vectors}
\begin{defn}[Killing Vectors]
A vector field on a Riemannian manifold called a Killing vector field if the metric is invariant under its flow, i.e. $\xi$ is a Killing vector field if $\lie_{\xi}(g)=0$. Note that Killing vectors generates one parameter group of isometries.
\end{defn}

Killing vectors have the following properties:
\begin{enumerate}
\item Let $\xi$ be a Killing vector field and $\gamma$ be a geodesic, then $g(\gamma'(t),\xi)$ is constant along $\gamma (t)$.
\item For a manifold of n dimensions, there can be at most $n (n+1)/2$ linearly independent Killing vector fields.
\item Killing vectors form a Lie subalgebra of set of smooth vector fields.
\end{enumerate}

\section{Riemannian Volume form and Hodge Star}
\begin{defn}[Riemannian Volume form] 
For a (pseudo-) Riemannian Manifold there exists  a unique smooth orientation form $\volg$ that satisfies $\volg(E_1,...,E_n)=1$  for every local oriented orthonormal frame $ \lbr E_i \rbr$ for $M$, called the \hypertarget{rievol}{ Riemmanian volume form}. In any oriented smooth coordinates $ \lbr x^i \rbr$ it has the local coordinate expression 
\begin{equation}
\volg = \sqrt{\pm \det(g_{ij})} \, dx^1 \wedge ... \wedge dx^n
\end{equation}
where $+,-$ is for Riemannian and Lorentzian signature respectively.
\end{defn}
This tensor is sometimes also called Levi-Civita Tensor. It also satisfies the following equation written in a coordinate basis:
\begin{equation}
\volg^{\m_1 ...\m_k \m_{k+1}... \m_n} {\volg}_{ \m_1 ... \m_k \n_{k+1} ... \n_n} = \pm (n-k)! k! \delta^{\left[\m_{k+1}\right.}_{\n_{k+1}} ...\delta^{\left. \m_{n} \right]}_{\n_{n}} \sgd
\end{equation}
Using Riemannian Volume form one also defines the following:
\begin{defn}[Hodge Star Operator]
Hodge star operator is a map $*: \Omega^k(M) \ra \Omega^{n-k}(M)$ for any $0 \leq k \leq n$, where n is the dimension of the space, such that 
\begin{equation}
\left( * \omega \right)_{\n_{k+1} ... \n_n}= \frac{1}{k! } {\volg^{\m_1 ... \m_k}}_{ \n_{k+1} ... \n_n} \omega_{\m_1 ... \m_k} \sgd
\end{equation}
\end{defn}
Note also that Hodge star defines a natural inner product, i.e. let $\omega,\eta$ be k-forms, define inner product $\dket , \dbra$ as
\begin{equation}
\dket \omega, \eta \dbra = \int_{M} \omega \wedge *  \eta \sgd
\end{equation}
Some properties of the Hodge star are
\begin{itemize}
\item $ \nabla^{\mu} \omega_{\mu}= \pm * d * \omega$,
\item $ curl X = 2 ( * dX^ {\flat} )^ {\#} $,
\item $ \nabla f = (df)^{\#} $,
\item $ \nabla^2 f = \pm *d*df $,
\item $ \nabla^{\mu} F_{\mu \nu} = \pm * d * F$ ,
\end{itemize}
where $\#$ is raising and $\flat$ is lowering operator.
\paragraph{Stokes' Theorem on a Riemannian Manifold.}
Remember the Stokes' theorem on a manifold, if $\omega$ is an (n-1)-form on an n-dimensional manifold $M$ then:
\begin{equation}
\int_M d\omega= \int_{\pr M} \omega \sgc
\end{equation}
Now let M be a Riemannian manifold with metric $g$ and let $\omega= v \inc \volg$, one can show Stokes's theorem gives:
\begin{equation}
\int_{M}{ \left( \nabla_{\mu} v^{\mu} \right) \volg }=\int_{M}{ \left( *d*v \right) \volg } = \int_{ \partial M }{ v \inc \volg } = \int_{\pr M} g(v,n) \mathlarger{\epsilon}_{\tilde{g}}
\end{equation}
where $\tilde{g}$ is the induced metric on $\pr M$.
\begin{example}
Consider $\R^3$ with the Euclidean metric. Riemannian volume form for it will be
\begin{equation}
\omega_{\R^3}=dx \wedge dy \wedge dz \sgd
\end{equation}
Apply the above equation for the vector
\begin{equation}
v= r \frac{\partial}{\partial r}= x \frac{\partial}{\partial x}+ y \frac{\partial}{\partial y}+z \frac{\partial}{\partial z} \sgd
\end{equation}
Then the divergence will be
\begin{equation}
*d*v=\nabla_a v^a=3 \sgd
\end{equation}
Thus on a sphere with radius $R$,
\begin{equation}
\int_{R}{ \left( *d*v \right) \epsilon } = 4 \pi R^3 \sgd
\end{equation}
On the other hand using the Stokes' theorem,
\begin{align}
\int_{R}{ \left( *d*v \right) \epsilon } = \int_{ \partial R }{ v \lrcorner \epsilon } &= \int_{ \partial R }{\epsilon_{ba_1 a_2... } v^b} \nonumber \\
& = \int_{ \partial R }{ \epsilon_{123} v^1 dy \wedge dz + \epsilon_{231} v^2 dz \wedge dx + \epsilon_{312} v^3 dx \wedge dy } \nonumber \\
&= \int_{ \partial R }{ x dy \wedge dz +...} \nonumber \\
& =\int_{ \partial R }{R^3 sin \theta d\phi \wedge d\theta}= 4 \pi R^3 \sgc
\end{align}
confirming our results.
\end{example}

\section{Laplace de Rham Operator and Hodge Decomposition}

\begin{defn}
On an n-dimensional Riemannian manifold $M$, adjoint exterior derivative is a map $\dad: \Omega^k(M) \rightarrow \Omega^{k-1}(M)$ defined as
\begin{equation}
d^{\dagger} \equiv (-1)^{nk+n+1} *d* \sgd
\end{equation}
\end{defn}
Note that
\begin{align}
\mbox{(i)}\,\,& (d^{\dagger})^2=0, \\
\mbox{(ii)}\,\,& \left( d^{\dagger} \omega \right)_{\m_1 ... \m_{k-1}} = \pm \left(-1 \right)^{nk+n+1} \nabla^{\alpha} \omega_{\alpha \m_1 ... \m_{k-1}} \sgc \label{d-adj} 
\end{align}
where $\nabla$ is the covariant derivative belonging to $g$.

\begin{defn}[Laplace de Rham Operator and Harmonic Form]
Laplace de Rham operator is a map $\Delta:\Omega^k(M) \ra \Omega^k(M)$ where
\begin{equation}
\Delta \omega= (d d^{\dagger} + d^{\dagger} d) \omega \sgd
\end{equation}
A k-form $\alpha$ is called a harmonic form if $\Delta \alpha=0$.
\end{defn}
\iffalse
\noindent
Using \ref{d-adj} one can see
\begin{align}
\left( d d^{\dagger} \omega \right)_{\mu_1 ... \mu_k} &= \pm \left(-1 \right)^{nk+n+1} k \nabla_{[\mu_k} \nabla^{\alpha} \omega_{|\alpha| \mu_1 ... \mu_{k-1}]} \sgc \\
\left( d^{\dagger} d \omega \right)_{\mu_1 ... \mu_k} &= \pm \left(-1 \right)^{nk+1} (k+1) \nabla^{\alpha} \nabla_{[\alpha} \omega_{ \mu_1 ... \mu_k]} \sgd\\
\end{align}
We also claim \expl{CHECK FOLLOWING FORMULA-BUT WE DO NOT USE THEM}:
\begin{equation}
(r+1) \nabla^{\alpha} \nabla_{[\alpha} \omega_{\mu_1... \mu_r]} = \nabla^{\alpha} \nabla_{\alpha} \omega_{\mu_1 ... \mu_r} - r \nabla^{\alpha} \nabla_{[\mu_1} \omega_{|\alpha| \mu_2 ... \mu_r ]} \sgc
\end{equation}
thus
\begin{equation}
\Delta \omega_{\mu_1 ... \mu_k} =  \left(-1 \right)^{nk+1+s} \left( (-1)^d k \nabla_{[\mu_k} \nabla^{\alpha} \omega_{|\alpha| \mu_1 ... \mu_{k-1}]} +   \nabla^{\alpha} \nabla_{\alpha} \omega_{\mu_1 ... \mu_k} - k \nabla^{\alpha} \nabla_{[\mu_1} \omega_{|\alpha| \mu_2 ... \mu_k ]} \right) \sgd
\end{equation}
In the literature this is known as Weitzenbock formula. In $n=3$ for some simple cases it becomes:
\begin{align}
k =0 \quad \Delta f&= - \nabla^\alpha \nabla_{\alpha} f \sgc \\
k =1 \quad \Delta \omega_\mu &= \nabla^\alpha \nabla_\alpha \omega_\mu - \nabla_\mu \nabla^\alpha \omega_\alpha - \nabla^\alpha \nabla_\mu \omega_\alpha \sgd
\end{align}
\fi
\begin{thm}[Hodge Decomposition]
A p-form on a compact, oriented Riemannian manifold can be uniquely decomposed as
\begin{equation}
\omega=\alpha + d \beta + d^{\dagger} \gamma \sgc
\end{equation}
where $\alpha$ is a harmonic form. 
\end{thm}
\begin{lemma}
$\derhamp(M) \cong Harm^p(M)$, where $Harm^p(M)$ is set of harmonic p-forms.
\end{lemma}

\section{Riemannian Submanifolds}
For many reasons it will be useful to consider tensors along a submanifold of a given Riemannian manifold and split these into parts tangent to the submanifold and orthogonal to submanifold. In the main text we will perform a similar analysis for a foliation, and will see additional structure; here we treat a single submanifold and see how structure imposed on it by its ambient manifold.

Let $(\amb,\metamb)$ be a Riemannian Manifold and $(\subm, \metsubm)$ an embedding.Let $X,Y \in T\subm$ and extend them to vector fields on $\amb$ such that they are also defined off $\subm$. Note that on $\subm$, extension of a vector field should be equal to that vector field. We use same letters for extensions.

Let $\consubm$ and $\conamb$ be Riemannian connections on $\subm$ and $\amb$ respectively. Now $\conamb_X Y$ will, in general, have components in $T\subm$ and $N\subm$, the orthogonal complement of $T\subm$ in $\left. T\amb \right|_M$. Our claim is that the tangential component will be equal to $\consubm_X Y$, i.e. :

\begin{thm}[Gauss Formula]
\begin{equation}\label{eq:gauss}
\tilde{\nabla}_X Y = \nabla_X Y + \left( \tilde{\nabla}_X Y \right)^{\bot} \sgd
\end{equation}
\end{thm}
\noindent
The map $\Pi: \tansubm \times \tansubm \ra \orthsubm $ such that   $\Pi(X,Y) \equiv \left( \tilde{\nabla}_X Y \right)^{\bot}$ is called the second fundamental form.

\begin{proof}
This is shown by defining a map $D: \tansubm \times \tansubm \ra \tansubm$ such that
\begin{equation}
D(X,Y)= (\conamb_X Y)^\top
\end{equation}
and showing that it is a symmetric and $\metsubm$ compatible connection.
\end{proof}
\noindent
The second fundamental form is
\begin{enumerate}
\item independent of extensions of $X,Y$,
\item bilinear over $C^\infty(\subm)$,
\item symmetric.
\end{enumerate}

\begin{thm}[The Weingarten Equation]\label{eq:wein}
Let $X,Y \in \tansubm$ and $Q \in \orthsubm$ (i.e. $Q$ is a vector field \hyperlink{vector field along}{along} $\subm$ that happens to be orthogonal to its tangent space at every point on $\subm$). Then following holds at points of $\subm$:
\begin{equation}
\tilde{g}( \tilde{\nabla}_X Q, Y)=- \tilde{g}(Q,\Pi(X,Y)) \sgd
\end{equation}
\end{thm}

\begin{thm}[Gauss Equation]
Let $X,Y,Z,W \in \tansubm$ with extensions. Then
\begin{equation}
\rieamb(X,Y,Z,W)=\riesubm(X,Y,Z,W)- \metamb(\Pi(X,W),\Pi(Y,Z))+\metamb(\Pi(X,Z),\Pi(Y,W)) \sgd
\end{equation}
\end{thm}
\paragraph{Induced volume form of a Riemannian submanifold:} As already have been used in Stokes' theorem for Riemannian manifolds, given a submanifold of codimension-1 one has the relationships
\begin{equation}\label{eq:appvol}
\vol_{\metamb}= n \wedge \vol_{\metsubm} \gag \vol_{\metsubm}= n \inc \vol_{\metamb} \sgc
\end{equation}
where $n$ is the unique normalized one form along $\subm$ that is normal to it. Note that this relation will not be applicable to the cases of a submanifold with a degenerate induced metric in a pseudo-Riemannian manifold.

\chapter{Symplectic Manifolds}

\section{Symplectic Form}
\begin{defn}[Symplectic Tensor]
A 2-covector $ \omega $ on a finite dimensional vector space V is said to be nondegenerate if for each nonzero $ v \in V $, there exists a $ w \in V $ such that $ \omega(v,w)\neq 0 $. A non-degenerate 2-covector is called a symplectic tensor.
\end{defn}

\begin{prop}
Let $ \omega $ be a symplectic tensor on an m-dimensional vector space $V$. Then m is even and there exists a basis for $V$ in which $\omega$ has the form
\begin{equation}
\omega= \sum_{i=1}^{n}{\alpha^i \wedge \beta^i} \sgc
\end{equation}
where n=m/2 and $ \left\lbrace \alpha^1,\beta^1, ..., \alpha^n, \beta^n \right\rbrace $ is the a basis of dual space $V^{*}$.
\end{prop}

The subset of $ T^{(0,k)}TM $ consisting of k-covectors is denoted by $ \Lambda^k T^{*}M $:
\begin{equation}
\Lambda^k T^{*}M = \prod_{p \in M}{\Lambda^{k}(T_{p}^{*} M)} \sgd
\end{equation}
A section of $ \Lambda^k T^{*}M $ is called a (differential) k-form: it is a tensor field which is a k-covector at each point.

\begin{defn}[Symplectic form and Symplectic Manifold]
A non-degenerate 2-form on a smooth manifold $M$ is a 2-form $\omega$ such that $\omega_p$ is a nondegenerate 2-covector for each $ p \in M $. A symplectic form ( or structure) on $M$ is a closed nondegenerate 2-form. A manifold endowed with a symplectic form is called a symplectic manifold.
\end{defn}

\begin{thm}[Darboux Theorem]
 Let $M$ be a 2n-dimensional manifold and $\omega$ be a nondegenerate 2-form on $M$.Then  $d\omega=0$ if and only if at each point $p \in M$ there are coordinates $(U,\psi)$ where $p \in M $ and $ \psi: U \rightarrow \left\{x_1, \dots, x_n,y_1, \dots , y_n \right\} $ satisfies $\psi(p)=0$ and
\begin{equation}
(\psi^{-1})^{*} \omega= \sum_{j=1}^{n}{dx_j \wedge dy_j} \sgd
\end{equation}
Such coordinates are called symplectic or Darboux coordinates.
\end{thm}

\section{Hamiltonian Vector Fields and Poisson Bracket}
\begin{defn}
Hamiltonian vector field of a smooth function f on a symplectic manifold $(M,\omega)$ is the smooth vector field $X_f$ defined by the property
\begin{equation}
\omega(X_f,Y)=df(Y)=Y(f)
\end{equation}
for $Y$ any smooth vector field.
\end{defn}
\noindent
One can show that in symplectic coordinates $\left\lbrace x^i,y^i \right\rbrace$ the Hamiltonian vector field $X_f$  is written as
\begin{equation}
X_f= \sum_{i=1}^{n}{\frac{\partial f}{\partial y^i} \frac{\partial}{\partial x^i}-\frac{\partial f}{\partial x^i} \frac{\partial}{\partial y^i}}
\end{equation}
In conjunction with a symplectic form one can define what is called Poisson bracket on smooth functions.

\begin{defn}
Given smooth functions $f,g$ on a symplectic manifold $(M,\omega)$ we define their Poisson bracket $\left\{ f,g \right\} \in C^{\infty}(M)$ by:
\begin{equation}
\left\{ f,g \right\} = \omega(X_f,X_g) = df(X_g) = X_g f
\end{equation}
Note that this gives the rate if change of f along the Hamiltonian flow of g.
\end{defn}

Some properties of Poisson bracket are as follows:
\begin{enumerate}
\item $\left\{ .,. \right\}$ is bilinear over $\mathcal{R}$,
\item $\left\{ f,g \right\} = -\left\{ g,f \right\}$,
\item $\left\{ \left\{ f,g \right\} ,h \right\}+\left\{ \left\{ g,h \right\} ,f \right\}+\left\{ \left\{ h,f \right\} ,g \right\} =0 $,
\item $ \lb X_f, X_g \rb= - X_{\lbr f,g \rbr }$
\end{enumerate}
for $f,g,h \in C^{\infty}(M)$.

\section{Hamilton's Equations from Hamiltonian Vector Fields}
Symplectic manifolds are a very nice and abstract way of describing the Hamiltonian prescription of physical systems. A classical system in physics is described as follows:

\uline{A classical dynamical system is described by a symplectic manifold $(M,\omega)$ and a function $H$ called the Hamiltonian. Each integral curve of $X_H$, the Hamiltonian vector field of $H$, describes a solution to the system with a different initial condition. We will call these systems Hamiltonian systems.}

Now we move onto illustrating this statement. A physical system is described by a set of ``generalized coordinates", called $\left\{ Q^i \right\}$, $i=1, \dots,n$. A Hamiltonian description of system is given by an even dimensional space described with set $\left\{ Q^i,P^i \right\}$ , where $P^i$ is the corresponding ``canonical momentum", which can be considered as new coordinates. Now this is an $\mathbb{R}^{2n}$ manifold with the Darboux coordinates $\left\{ Q^i,P^i \right\}$. The Hamiltonian $H$, which basically gives the energy of the system is a smooth function on this manifold. Let us consider the Hamiltonian vector field belonging to $H$, called $X_H$. Let it have a flow $\gamma(t)=(Q^i(t),P^i(t))$. Then by the definition of flows we must have
\begin{equation}
X_H= \frac{d \gamma(t)}{dt} \sgd
\end{equation}
Written in Darboux coordinates,
\begin{equation}
X_H= \frac{d Q^i(t)}{dt} \frac{\pr}{\pr Q^i} + \frac{d P^i(t)}{dt} \frac{\pr}{\pr P^i} \sgd
\end{equation}
However by the definition of a Hamiltonian vector field we also know
\begin{equation}
X_H= \frac{\partial H}{\partial P^i} \frac{\partial}{\partial Q^i}-\frac{\partial H}{\partial Q^i} \frac{\partial}{\partial P^i} \sgd
\end{equation}
Matching the coefficients will give us \uline{Hamilton's equations}:
\begin{equation}
\frac{dQ^i}{dt} = \frac{\partial H}{d P^i} \cg \frac{dP^i}{dt} = -\frac{\partial H}{d Q^i} \sgd
\end{equation}

\section{Symmetries and Noether's Theorem}
\begin{defn}
A smooth vector field $X$ on a symplectic manifold $(M,\omega)$ is said to be a symplectic vector field if $\omega$ is invariant under the flow of X, i.e. if $\lie_X \omega=0$. It is said to be globally Hamiltonian if there exists a smooth function on $M$ such that $X=X_f$, and locally Hamiltonian if at each point on $M$, $X=X_{f_i}$ on a neighborhood $U_i \cap M$.
\end{defn}

\begin{thm}
Let $(M,\omega)$ be a symplectic manifold. A smooth vector field on $M$ is symplectic if and only if it is locally Hamiltonian. Every locally Hamiltonian vector field on $M$ is globally Hamiltonian if and only if $\derham^1(M)=0$.
\end{thm}

\begin{proof}
Let us first show that $X$ being symplectic means $X$ is locally Hamiltonian.\\
A smooth vector field $X$ is symplectic if $ \mathcal{L}_X \omega =0$. Now use Cartan's Magic formula:
\begin{equation}
\mathcal{L}_X \omega = d ( X \lrcorner \ \omega) + X \lrcorner (d\omega)= d ( X \lrcorner \ \omega)=0 \sgc
\end{equation}
where we used $d\omega=0$ that comes from the definition of the symplectic form. By \hyperlink{poincare}{Poincare's lemma} then $( X \lrcorner \ \omega)$ is locally exact i.e. there exists a function $f$ such that $( X \lrcorner \ \omega)=df$, thus $X$ is locally Hamiltonian.\\
The fact that $X$'s local Hamiltonianness implies that $X$ is symplectic can be trivially shown by following the steps above backwards.\\
If $\derham^1(M)=0$, then same arguments hold globally.
\end{proof}

A smooth function $h$ on a Hamiltonian system $(M,\omega,H)$ is said to be conserved in time if it is constant on every integral curve of $X_H$ i.e. if $X_H(h)=0$. Note that $h$ is conserved in time if and only if $ \left\{ h,H \right\} = 0 $; by the definition of Poisson bracket. 

\begin{defn}[Infinitesimal Symmetry]
A smooth vector field $V$ on a Hamiltonian system $(M,\omega,H)$ is called an infinitesimal symmetry of the system if both $\omega$ and $H$ are invariant under the flow of $V$. Note that if $V$ is an infinitesimal symmetry then it is a symplectic vector field satisfying $V(H)=0$.
\end{defn}

Let $(M,\omega,H)$ be a Hamiltonian system. If $\theta: \R \times M \ra M$ is the flow of an infinitesimal symmetry and $\gamma: \R \ra M$ is a solution to the system then for each $ s \in \R $, $ \theta_s \circ \gamma$, where $\theta_s(p)=\theta(s,p)$ for all $p \in M$, is also a solution on its domain of definition.

\begin{thm}[Noether's Theorem]
Let $(M,\omega,H)$ be a Hamiltonian system. If f is conserved in time, then its Hamiltonian vector field is an infinitesimal symmetry. Conversely, if $\derham^1(M)=0$,i.e. de Rham cohomology group of degree 1 is zero, then each infinitesimal symmetry is the Hamiltonian vector field of a conserved quantity, which is unique up to addition of a function that is constant on $M$.
\end{thm}

\begin{proof}
\hfill\\
Let us first show f being conserved implies $X_f$ is an infinitesimal symmetry:\\
If f is conserved in time, then $ \left\{f,H\right\}=0$ which implies $ \left\{f,H\right\}= -\left\{H,f\right\}= X_f H=0$, thus $H$ is constant along the flow of $X_f$. It can be easily shown by using Cartan's magic formula that a globally Hamiltonian vector field is a symplectic vector field. Because of this $\omega$ is also conserved under the flow of $X_f$. Thus $X_f$ is an infinitesimal symmetry.\\
\hfill \\ 
Now let's show that if $V$ is an infinitesimal symmetry and $\derham^1(M)=0$ then $V=X_f$ where f is conserved:\\
Since V infinitesimal symmetry, it is a symplectic vector field, and since $\derham^1(M)=0$, it is globally Hamiltonian i.e. $V=X_f$ for some f. f is not exactly unique: Let $V=X_g=X_f$ for some $g \neq f$. Then considering the definition of $X_f$, $df=dg$ and thus $g=f+c$ where $c$ is a constant on each component of $M$.
\end{proof}

\begin{example}[Free Particle in 3D]
A free particle in 3D is described by the Hamiltonian
\begin{equation}
H=\frac{p_x^2}{2}+\frac{p_y^2}{2}+\frac{p_z^2}{2} \sgd
\end{equation}
Symplectic form in Darboux coordinates is
\begin{equation}
\omega= \sum_i dx_i \wedge dp_i \sgc
\end{equation}
where Hamiltonian vector field for H is
\begin{equation}
X_H= \frac{\partial H}{\partial p_i} \frac{\partial}{\partial x_i} - \frac{\partial H}{\partial x_i} \frac{\partial}{\partial p_i}=p_i \frac{\partial}{\partial x_i} \sgd
\end{equation}
Conserved charges should satisfy
\begin{equation}\label{eq:charge}
\lbr Q,H \rbr=0 \quad \ra \quad \frac{\pr Q}{\pr x_i} p_i=0 \sgd
\end{equation}
For the solution to this make a simple quess:
\begin{equation}
\vec{\nabla}Q_1=  \lp -p_y , p_x, 0 \rp \sgd
\end{equation}
Remember that $Q$ is a function on the symplectic manifold-i.e. it depends on the 6-dimensional phase space coordinates. Then solving the above differential equation one finds
\begin{equation}
Q_1= -p_y x + p_x y \sgd
\end{equation}
Note that this is the $z$ component of the angular momentum. Other conserved quantities will be other components of the angular momentum, i.e.
\begin{align}
\vec{\nabla}Q_2 &=  \lp -p_z ,0, p_x\rp \sgc\\
\vec{\nabla}Q_3 &=  \lp 0,-p_z , p_y\rp \sgd
\end{align}
Note that other solutions to \eqref{eq:charge} are arbitrary functions of the momentum: this is clearly expected, as for a free particle momentum is also conserved. Since we take our symplectic manifold to be trivial topologically, all deRham cohomology groups are trivial, and thus all symplectic vector fields will be Hamiltonian vector fields.
\end{example}
\section{Constrained Systems}\label{sec:appconstr}

In the previous section we have seen how dynamical systems are described by a symplectic manifold with a Hamiltonian function. However in physics one encounters cases where the system has to satisfy additional conditions. These are called constraints and in a Hamiltonian system will be simply described as functions $C_a$ of phase space that are set to zero. Then the actual dynamics will take place on the submanifold defined by $C_a=0$ of the original phase space. To have consistent system \expl{???} we will also impose that constraints are time independent.
 
To describe the dynamics on the submanifold one needs to have a symplectic form on this submanifold. However it is not clear that the pullback of a symplectic form on the submanifold is itself a symplectic manifold on the submanifold. \uline{Thus the problem of constrained systems in physics can actually be considered as the problem of having a well defined symplectic form on the constraint submanifold.}

The degeneracy properties of the induced symplectic manifold is determined by the Poisson brackets of the constraint functions. Let us consider the two extreme cases first:
\paragraph{First class only case:}
Let us consider a symplectic manifold $(M,\omega)$ of dimension $m$. Let us consider that there exists $n$ functions $C_a$ on this manifold and let the us call the submanifold such that $C_a=0$,$a=1,...,n$ as $\Sigma$. Let $i:\Sigma \ra M$ be the embedding map. Assume that 
\begin{equation}
\left. \{ C_a, C_b \} \right|_{\Sigma}=0 \sgd
\end{equation}
Constraints that has zero Poisson bracket with all other constraints on $\Sigma$ are called first-class constraints. Note that here we assumed all of constraints to be first class. Let the Hamiltonian vector field on $M$ corresponding to function $C_a$ be called $X_a$. Then the following holds:
\begin{enumerate}
\item $\left. X_a \right|_{\Sigma}=0$ are tangent to $\Sigma$. 

This can be shown easily by using proposition \ref{prop:tansubm}: $X_a$ is tangent to $\Sigma$ if and only if 
\begin{equation}
\left. X_a(C_b) \right|_{\Sigma}= 0 \quad b=1,...,n \sgd
\end{equation} 
However by definition of Hamilton vector fields this is equivalent to demanding
\begin{equation}
\left. \{ C_a, C_b \} \right|_{\Sigma}=0 \sgc
\end{equation}
which was exactly our assumption for the constraints.
\item $X_a$ are degenerate vector fields of the induced symplectic form $i^* \omega$.

Let $Y$ be a vector field on $M$ that is tangent to $\Sigma$ on $\Sigma$. By definition of Hamiltonian vector fields
\begin{equation}
\left. \omega( X_a, Y) \right|_{\Sigma}= \left. Y(C_a) \right|_{\Sigma} = 0 \sgc
\end{equation}
where the last equality again follows from proposition \ref{prop:tansubm}. Then
\begin{equation}
i^* \omega(X_a,Y)=\left. \omega(X_a,Y) \right|_{\Sigma}=0 \quad \forall Y \in T\Sigma \sgd
\end{equation}
\item $X_a$ are involutive on $\Sigma$.

This directly follows from the properties of the Poisson bracket.
\end{enumerate}
Since $X_a$ are involutive on $\Sigma$ by Frobenius theorem they foliate $\Sigma$ into $n$ dimensional integral manifolds. Integral manifolds do not intersect and thus forms equivalence classes.\expl{this is vague?} By taking the quotient of $\Sigma$ by this equivalence relation, one can have a non-degenerate symplectic form. This quotient space is known as the reduced phase space. \expl{first class constr. are gauge freedom}
\paragraph{Second Class only case:}
This will be the case where non of the constraints has zero Poisson bracket on $\Sigma$. For this case $X_a$ will not be tangential to $\Sigma$, and the induced symplectic form $i^* \sigma$ will be non-degenerate.\\
\hfill\\
The equivalence relation defined by the first-class constraints are known in physics as gauge equivalence i.e. any two points on the phase space is called physically equivalent if they are on the same integral line of a flow for first-class constraints.
\expl{Gauge choice}